\documentclass[11pt]{article}
\usepackage[T1]{fontenc}
\usepackage[utf8]{inputenc}
\usepackage{color}
\usepackage{array}
\usepackage{booktabs}
\usepackage{url}
\usepackage{multirow}
\usepackage{amsmath}
\usepackage{amssymb}
\usepackage{graphicx}
\usepackage[authoryear]{natbib}
\usepackage[unicode=true,
 bookmarks=false,
 breaklinks=false,pdfborder={0 0 1},
 colorlinks=false]
 {hyperref}
\usepackage{siunitx}
\usepackage{url}
\usepackage{tabularx}
\usepackage{array}
\usepackage{mathtools}
\usepackage{booktabs}
\usepackage{subfig}
\usepackage{bm}
\usepackage{tikz}
\usetikzlibrary{shapes.geometric, arrows,shapes.misc, positioning}
\tikzset{
  half circle/.style={
      semicircle,
      shape border rotate=90,
      minimum size=1mm
      }
}
\graphicspath{{./art/}}
\allowdisplaybreaks
\usepackage{color}

\makeatletter

\usepackage{latexsym}
\usepackage{multirow}\usepackage{bm}
\usepackage{url}
\usepackage{tabularx}
\usepackage{booktabs}
\usepackage{setspace}

\usepackage{enumerate}

\newcommand{\spacingset}[1]{\renewcommand{\baselinestretch}%
{#1}\small\normalsize}

\newcommand{\E}{\mathrm{E}}

\newcommand{\pr}{\mathrm{P}}

\newcommand*{\nindep}{%
	\mathbin{
		\mathpalette{\@indep}{\not}
	}%
}
\def\bSig\mathbf{\Sigma}

\newcommand*{\indep}{%
	\rotatebox[origin=c]{90}{$\models$}
}

\newcommand{\deriv}{\big|_{\theta=0}}

\newcommand{\jtr}{\text{J2R}}

\usepackage{amsmath}
\usepackage{amssymb}
\usepackage{amsfonts}
\usepackage{multirow}
\usepackage{amsthm}

\newtheorem{theorem}{Theorem}
\newtheorem{lemma}{Lemma}
\newtheorem{corollary}{Corollary}
\newtheorem{Assumption}{Assumption}

\theoremstyle{definition}

\newtheorem{example}{Example}

\newtheorem{assumptionalt}{Assumption}[Assumption]

\newenvironment{assumptionp}[1]{
  \renewcommand\theassumptionalt{#1}
  \assumptionalt
}{\endassumptionalt}

\begin{document}


\title{\textbf{Multiply robust estimators in longitudinal studies with missing data under control-based imputation}}

\author{Siyi Liu$^{1}$, Shu Yang$^{1}$
Yilong Zhang$^{2}$, and
Guanghan (Frank) Liu$^{2}$\\
$^{1}$Department of Statistics, North Carolina State University, Raleigh, NC, USA \\
$^{2}$Merck \& Co., Inc., Kenilworth, NJ, USA}

\date{\vspace{-5ex}}

\maketitle

\spacingset{1.5} 
\begin{abstract}
Longitudinal studies are often subject to missing data. The recent guidance from regulatory agencies such as the ICH E9(R1) addendum addresses the importance of defining a treatment effect estimand with the consideration of intercurrent events. Jump-to-reference (J2R) is one classical control-based scenario for the treatment effect evaluation, where the participants in the treatment group after intercurrent events are assumed to have the same disease progress as those with identical covariates in the control group. We establish new estimators to assess the average treatment effect based on a proposed potential outcomes framework under J2R. Various identification formulas are constructed, motivating estimators that rely on different parts of the observed data distribution. Moreover, we obtain a novel estimator inspired by the efficient influence function, with multiple robustness in the sense that it achieves $n^{1/2}$-consistency if any pairs of multiple nuisance functions are correctly specified, or if the nuisance functions converge at a rate not slower than $n^{-1/4}$ when using flexible modeling approaches. The finite-sample performance of the proposed estimators is validated in simulation studies and an antidepressant clinical trial. 

\noindent \textbf{keywords:} Longitudinal clinical trial, longitudinal observational study, semiparametric theory, sensitivity analysis.
\end{abstract}
\newpage{}

\section{Introduction}\label{sec:intro} 
Missing data are a major concern in clinical studies,
especially in longitudinal settings. Participants are likely to deviate
from the current treatment due to the loss of follow-ups or a shift
to certain rescue therapy. To estimate the treatment effect precisely,
additional assumptions for the missing components are needed. It calls
for the importance of defining an estimand that can reflect the key
clinical questions of interest and take into account the intercurrent
events such as the discontinuation of the treatment \citep{international2019addendum}.

Different strategies are put forward by \citet{international2019addendum}
to deal with the intercurrent events. The \textit{hypothetical} strategy commonly envisions that participants who discontinue the
treatment are in compliance, i.e., they still take the assigned drug
throughout the entire study period. This approach, which is connected
to the unverifiable missing at random (MAR; \citealp{rubin1976inference})
assumption, frequently appears in the primary analysis to evaluate
the treatment efficacy. However, this hypothetical scenario may not be {\it realistic}, if participants lose access to the benefited test drug afterward. Under this
circumstance, those individuals are more likely to resemble the observed ones with identical historical information
in the control group, leading to control-based imputation (CBI; \citealp{carpenter2013analysis}).
CBI uses the \textit{treatment policy} strategy to construct a treatment
effect estimand that addresses a ``treatment switching'' scenario
for those individuals who drop out of the treated group.
As CBI reveals a discrepancy in outcome profiles between observed individuals and dropouts with the same history in the treated group,
a missing not at random (MNAR; \citealp{rubin1976inference}) pattern
is detected for the intercurrent events. Since the resulting estimand
is constructed under MNAR, it is often used in sensitivity analyses
(e.g., \citealp{carpenter2013analysis}; \citealp{liu2016analysis};
\citealp{cro2020sensitivity}; \citealp{di2022}; \citealp{rr2022})
to explore the robustness of results to alternative missing data assumptions
against MAR. Moreover, it has been receiving growing attention in
the primary analysis of clinical trials \citep{tan2021review} and
observational studies \citep{lee2021framework}.

Among the proposed CBI scenarios, we focus on one specific setting called jump-to-reference (J2R; \citealp{carpenter2013analysis}) throughout the paper, which has appeared in several regulatory reports (e.g., \citealp{cr2016tresiba}). In oncology trials, J2R is widely applicable since it is common for patients to shift to standard care if they quit the test therapy due to tumor progression \citep{mallinckrodt2019estimands}. Its usefulness is also revealed in the clinical trials of chronic pain treatments, where the subjects who drop out because they fail to experience pain relief may resemble the remaining ones in the control group \citep{gewandter2020improving}. The motivating example, which will be analyzed in Section \ref{sec:app}, uses an antidepressant trial conducted under the Auspices of the Drug Information Association \citep{mallinckrodt2014recent} to illustrate the usage of J2R. The trial collects the Hamilton Depression Rating Scale for 17 items (HAMD-17) scores at baseline and weeks 1, 2, 4, 6, and 8 among 100 randomly assigned participants in both the control and the treatment groups. We are interested in the average treatment effect (ATE) on the HAMD-17 score regardless of the occurrence of the intercurrent events, i.e., the ATE under the \textit{treatment policy}{} condition. As the test drug in this trial possesses a short-term effect, a reduced treatment effect is expected since the subjects taking the experimental drug are likely to experience no more treatment benefits after dropping out, indicating a J2R pattern. As a result, using the \textit{treatment policy}{} strategy and the guidelines in \citet{international2019addendum}, we define the treatment effect estimand as the mean difference of the change in the HAMD-17 score at the last time point from the baseline, assuming that the missing outcomes share the same profile as the observed ones with the identical history in the control group. The defined J2R estimand is an intent-to-treat (ITT) estimand, as it matches the goal of assessing the treatment effect in the group to which the individuals were initially assigned, regardless of the intervention \citep{lipkovich2020causal}.

The likelihood-based method and multiple imputation \citep{rubin2004multiple}
are two typical parametric approaches to handle missing data (\citealp{mallinckrodt2019estimands};
\citealp{di2022}). However, they will result in a biased estimate
of the ATE if any component of the likelihood function is misspecified.
When the parametric modeling assumptions are untenable, semiparametric
estimators based on the weighted estimating equations can be applied.
\citet{robins1994estimation} propose a doubly robust estimator for
the regression coefficients under MAR. \citet{bang2005doubly} further
develop a doubly robust estimator in longitudinal data with a monotone
missingness pattern using sequential regressions. While the robust
estimators under MAR have been well studied, they remain uncultivated
in the area of longitudinal clinical studies under MNAR-related scenarios.

Towards this end, we develop a semiparametric framework to evaluate
the ATE in longitudinal studies under J2R. As the estimand is defined
under an envisioned scenario where the outcomes have not been observed,
a potential outcomes framework is proposed to describe the counterfactuals.
The assumptions regarding treatment ignorability and partial ignorability
of missingness with causal consistency in the context of J2R are put
forward for identification. As a stepping stone, we first consider
cross-sectional studies, a special case of longitudinal studies with
one follow-up time. We discover three identification formulas for
the ATE, each of which invokes an estimator that relies on two of
the three models: 
\begin{enumerate}
\item[(a)] the \textit{propensity score}, as the model of the treatment conditional
on the observed history; 
\item[(b)] the \textit{response probability}, as the model of the response status
conditional on the observed historical covariates and the treatment; 
\item[(c)] the \textit{outcome mean}, as the model of the mean outcomes conditional
on the observed historical covariates and the treatment. 
\end{enumerate}
The three estimators assess the ATE in distinct aspects, motivating
us to construct a new estimator that combines all the modeling features.
Drawing on the semiparametric theory \citep{bickel1993efficient},
we obtain the efficient influence function (EIF) and use it to prompt
a novel estimator incorporating models (a)--(c). The proposed estimator
has a remarkable property of triple robustness (\citealp{wang2018bounded};
\citealp{jiang2020multiply}), in the sense that it is consistent
if any two of the three models are correctly specified when using
parametric models or achieves a $n^{1/2}$-consistency if the models
converge at a rate not slower than $n^{-1/4}$ when using flexible
models such as semiparametric or machine learning models. Extending
to longitudinal clinical studies, an additional model is needed for
identification: 
\begin{enumerate}
\item[(d)] the \textit{pattern mean}, as the model of the mean outcomes adjusted
by the response probability conditional on the observed history and
the treatment for any missingness pattern. 
\end{enumerate}
Even under MAR, the derivation of the EIF for longitudinal data is
notoriously challenging. The complexity is escalated 
under J2R, where the treatment group involves additional outcome information
from the control group, resulting in unexplored territory to date.
Our major theoretical contribution is to obtain the EIF in longitudinal
studies, which enables us to construct a multiply robust estimator
with the guaranteed $n^{1/2}$-consistency and asymptotic normality
if models (a)--(d) have convergence rates not slower than $n^{-1/4}$.
To mitigate the impact of extreme values in the estimator, we seek
alternative formations to obtain more stabilized estimators via normalization
\citep{lunceford2004stratification} and calibration (e.g., \citealp{hainmueller2012entropy};
\citealp{zhao2019covariate}; \citealp{lee2021improving}). Moreover,
a sequential estimation procedure that is analogous to the steps in
\citet{bang2005doubly} but under the more complex MNAR-related setting
is provided to obtain the estimator in practice. Inspired by the semiparametric
efficiency bound the estimator attains, we provide an EIF-based variance
estimator. 

The rest of the paper proceeds as follows. Section \ref{sec:1time}
constructs the semiparametric framework under J2R in cross-sectional
studies. Section \ref{sec:longi} extends it to longitudinal data.
Section \ref{sec:simu} assesses the finite-sample performance of
the proposed estimator via simulations. Section \ref{sec:app} uses
antidepressant trial data to further validate the novel estimator.
Conclusions and remarks are presented in Section \ref{sec:conclusion}.
Supporting information contains technical details, additional simulation
and real-data application results.

\section{Cross-sectional studies \label{sec:1time}}
To ground ideas, we first focus on cross-sectional studies. Let $A_{i}$
be the binary treatment, $X_{i}$ the baseline covariates, $Y_{1,i}$
the outcome, and $R_{1,i}$ the response indicator where $R_{1,i}=1$
indicates the outcome is observed and $R_{1,i}=0$ otherwise, where the subscript $1$ indicates the first post-baseline time point,
for
unit $i=1,\ldots,n$. 
Assume $\left\{ X_{i},A_{i},R_{1,i},Y_{1,i}:i=1,\cdots,n\right\} $
are independent and identically distributed. For simplicity of notation,
omit the subscript $i$ 
for the subject
. Let $V=(X,A,R_{1}Y_{1},R_{1})$
be the random vector of all observed variables and follow the 
distribution $\mathbb{P}$. 
To define the estimand unambiguously, we extend the causal framework
in \citet{lipkovich2020causal} and introduce the potential outcomes
framework by defining $R_{1}(a)$ as the potential response indicator
received treatment $a$ and $Y_{1}(a,r)$ as the potential outcome
received treatment $a$ with response status $r$. As a shorthand,
we also introduce the potential outcome $Y_{1}(a)=Y_{1}\{a,R_{1}(a)\}$
to acknowledge the equivalence between the potential outcome with
$A=a$ and the potential outcome with $A=a$ and $R_{1}$ to be the
value it would have been if $A=a$ based on the composition assumption
\citep{vanderweele2009conceptual}.

\begin{Assumption}[Treatment ignorability]\label{assump:1.1time}

$A\indep\left\{ R_{1}(a),Y_{1}(a,r)\right\} \mid X$, for all $a$
and $r$.

\end{Assumption}

Assumption \ref{assump:1.1time} is the classic treatment ignorability
in observational studies \citep{rosenbaum1983central}. In randomized
clinical trials, the treatment ignorability holds naturally. 
\begin{Assumption}[Causal consistency]\label{assump:2.1time}

$R_{1}=R_{1}(A),$ and $Y_{1}=Y_{1}\left\{ A,R_{1}(A)\right\} $.

\end{Assumption}
Assumption \ref{assump:2.1time} is the stable unit treatment value
assumption proposed by \citet{rubin1980randomization}. 
\begin{Assumption}[Partial ignorability of missingness]\label{assump:3.1time}

$R_{1}(0)\indep Y_{1}(0,r)\mid X$, for all $r$. 

\end{Assumption}
We distinguish Assumption \ref{assump:3.1time} from the conventional
MAR assumption, as it only requires conditional independence between
the potential response status and the potential outcome under any
response status in the control group. Since the control group in most clinical studies represents the placebo or standard care, the missingness ignorability matches the rationale that participants in this group still adhere to the assigned treatment after dropping out. 

\begin{Assumption}[J2R for the outcome mean]\label{assump:4.1time}

$\E\left\{ Y_{1}(1,0)\mid X,R_{1}(1)=0\right\} =\E\{Y_{1}(0)\mid X\}$.

\end{Assumption}
Assumption \ref{assump:4.1time} is vital as it specifies the outcome model under J2R. 
In the treated group, Assumptions \ref{assump:3.1time} and \ref{assump:4.1time} jointly characterize MNAR related to J2R, as the outcome distributions between observed individuals and dropouts are different based on the construction of the outcome mean. J2R is prespecified in the study protocol and belongs to a class of unverifiable assumptions on the outcome profile to target dropouts, revealing its applicability in diverse areas such as chronic diseases and oncology trials \citep{mallinckrodt2019estimands}. {In practice, one can include the outcome predictors of the control group in the outcome model to enhance the credibility of this assumption. Meanwhile, caution should be taken. Despite the prevalence of J2R, it may not be suitable for drugs with an enduring treatment benefit.}

Figure \ref{fig:demo} visualizes the four assumptions and extends the single-world intervention graph \citep{richardson2013single} to link counterfactuals with treatments. As Assumptions \ref{assump:3.1time} and \ref{assump:4.1time} imply differences in the distributions of the potential variables $R_{1}(a)$ and $Y_{1}(a)$ between treatments, we invent a graph containing both sets of the potential variables $\{R_{1}(0),Y_{1}(0,r_{1})\}$ and $\{R_{1}(1),Y_{1}(1,r_{1})\}$ and call it the double-world intervention graph (DWIG). By splitting the nodes to capture double-world distributions of the observed data, the DWIG shows different profiles for both potential variable sets and visualizes all causal assumptions.

{\large{}}
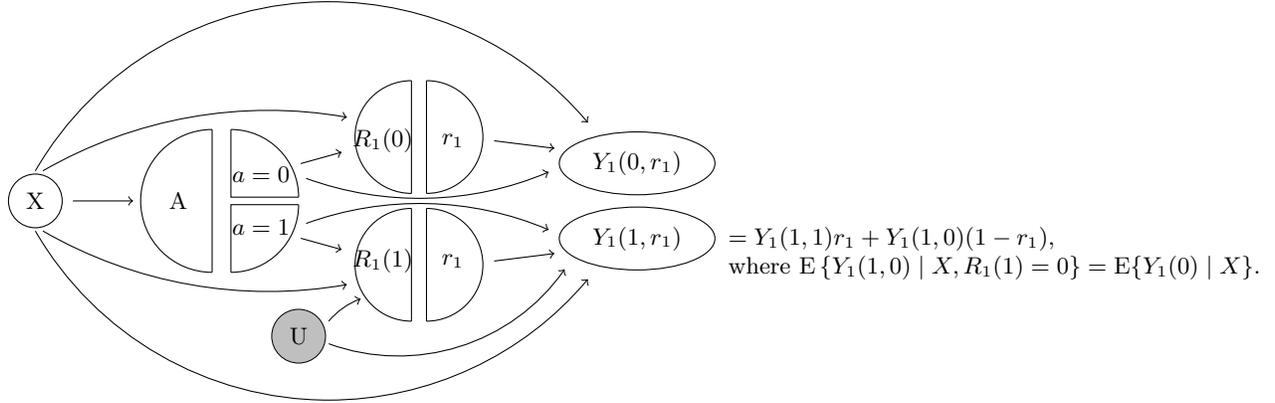
\begin{figure}
{\large{}\centering}{\large\par}
\centering{}{\large{}\caption{The DWIG encodes causal Assumptions \ref{assump:1.1time}--\ref{assump:4.1time}
and extends the single-world intervention graph \citep{richardson2013single},
visualizing double-world joint distributions $f\{X,A,R_{1}(0),Y_{1}(0,r_{1})\}$
and $f\{X,A,R_{1}(1),Y_{1}(1,r_{1})\}$. The vacancy of edges between
$A$ and the potential variables $\{R_{1}(a),Y_{1}(a,r_{1})\}$ represents
Assumption \ref{assump:1.1time}. Assumption \ref{assump:2.1time}
links $a$ with $R_{1}(a)$ and $(a,r_{1})$ with $Y_{1}(a,r_{1})$
to illustrate the causal consistency. The partial ignorability of
missingness in the control group in Assumption \ref{assump:3.1time}
only connects $R_{1}(0)$ and $Y_{1}(0,r)$ through $X$. The side
note and the additional involvement of an unmeasured confounder $U$
that links between $R_{1}(1)$ and $Y_{1}(1,r)$ indicate Assumption
\ref{assump:4.1time} and reveal an MNAR pattern invoked by J2R.}
\label{fig:demo}}
\begin{centering}
{\large{}
\begin{tikzpicture}[every node/.style=scale=1.1, node distance = 3cm, auto]
\tikzstyle{block} = [circle, draw, text centered]
\tikzstyle{outcome} = [ellipse, draw, text centered]

\tikzstyle{every node} = [font = \footnotesize]

\node [block]  at (-11,0) (X) {
X
};

\draw  (-8.65,0.95) arc (90:270:0.95);
\draw[-] (-8.65,0.95)--(-8.65,-0.95);
\node at (-9.1,0) (A) {A};
\draw  (-7.5,0.05) arc (0:90:0.9);
\draw[-] (-8.4,0.95)--(-8.4,0.05)--(-7.5,0.05);
\node at (-8,0.35) (a0) {$a=0$};
\draw  (-7.5,-0.05) arc (0:-90:0.9);
\draw[-] (-8.4,-0.95)--(-8.4,-0.05)--(-7.5,-0.05);
\node at (-8,-0.35) (a1) {$a=1$};

\draw  (-6,1.6) arc (90:270:0.75);
\draw[-] (-6,1.6)--(-6,0.1);
\node at (-6.38,0.8) (R0) {$R_1(0)$};
\draw  (-6,-1.6) arc (270:90:0.75);
\draw[-] (-6,-1.6)--(-6,-0.1);
\node at (-6.38,-0.8) (R1) {$R_1(1)$};

\draw  (-5.8,0.1) arc (-90:90:0.75);
\draw[-] (-5.8,1.6)--(-5.8,0.1);
\node at (-5.45,0.8) (r0) {$r_1$};
\draw  (-5.8,-1.6) arc (-90:90:0.75);
\draw[-] (-5.8,-1.6)--(-5.8,-0.1);
\node at (-5.45,-0.8) (r1) {$r_1$};

\fill[gray!50] (-7.5,-1.8) circle (0.35cm);
\node[block] at (-7.5,-1.8) (U) {U};

\node[outcome] (Y0) at (-3,0.5) {$Y_1(0,r_1)$};
\node[outcome] (Y1) at (-3,-0.5) {$Y_1(1,r_1)$};
\node (Y1e) at (1.75, -0.7) {\begin{tabular}{l} 	
    $=Y_1(1,1)r_1 + Y_1(1,0)(1 - r_1),$ \\
    	where $\E\left\{ Y_{1}(1,0)\mid X,R_{1}(1)=0\right\} =\E\{Y_{1}(0)\mid X\}.$\\
\end{tabular}
};

\draw[->] (-10.9,0.4) arc(120:80:6);
\draw[->] (-10.9,-0.4) arc(240:280:6);
\draw[->] (-11,0.4) arc(150:40:4.5);
\draw[->] (-11,-0.4) arc(210:320:4.5);
\draw[->] (-7.4,-0.3) arc(110:67:4.4);
\draw[->] (-7.4,0.3) arc(250:293:4.4);
\draw[->] (-10.5,0)--(-9.7,0);
\draw[->] (-4.9,0.8)--(-4.1,0.7);
\draw[->] (-4.9,-0.8)--(-4.1,-0.7);
\draw[->] (-7.1,-1.6) arc(140:110:1);
\draw[->] (-7.1,-1.9) arc(250:325:2.7);

\draw [->] (a0) edge (R0);
\draw [->] (a1) edge (R1);\end{tikzpicture}}{\large\par}
\par\end{centering}
{\large\par}
\end{figure}
{\large\par}

\subsection{Three identification formulas under J2R}

The ATE can be expressed under the potential outcomes framework as $\tau_{1}^{\jtr}=\E\{Y_{1}(1)-Y_{1}(0)\}$. 
Define the propensity score as $e(X)=\pr(A=1\mid X)$, the response probability
as ${\color{red}{\color{black}\pi_{1}(a,X)=\pr(R_{1}=1\mid X,A=a)}}$,
the outcome mean as $\mu_{1}^{a}(X)=\E(Y_{1}\mid X,R_{1}=1,A=a)$.
The following theorem provides three identification formulas of the
ATE.  

\begin{theorem}\label{thm:iden_1time}
Under Assumptions \ref{assump:1.1time}--\ref{assump:4.1time}, assume
there exists $\varepsilon>0,$ such that{} $\varepsilon<\big\{ e(X),\allowbreak\pi_{1}(a,X)\big\}<1-\varepsilon$
for all $X$ and $a$, the following identification formulas hold. 
\begin{enumerate}
\item[(a)] Based on the response probability and outcome mean, $\tau_{1}^{\text{\jtr}}=\mathbb{E}\left[\pi_{1}(1,X)\left\{ \mu_{1}^{1}(X)-\mu_{1}^{0}(X)\right\} \right].$ 
\item[(b)] Based on the propensity score and outcome mean, $\tau_{1}^{\text{\jtr}}=\mathbb{E}\Big((2A-1)\big\{ R_{1}Y_{1}+(1-R_{1})\mu_{1}^{0}(X)\big\}/\allowbreak\big[e(X)^{A}\{1-e(X)\}^{1-A}\big]\Big).$ 
\item[(c)] Based on the propensity score and response probability, $\tau_{1}^{\jtr}=\mathbb{E}\Big(AR_{1}Y_{1}/e(X)-(1-A)\pi_{1}(1,X)R_{1}\allowbreak Y_{1}/\big[\{1-e(X)\}\pi_{1}(0,X)\big]\Big).$ 
\end{enumerate}
\end{theorem}
Theorem \ref{thm:iden_1time} requires the positivity assumption of the treatment assignment \citep{rosenbaum1983central}. It means that each participant has a nonzero probability of being assigned to the control or treatment group. When missingness is involved, a positivity assumption regarding the response probability is also imposed, indicating that each individual has a chance to be observed at the study endpoint. As the missing components follow a MAR pattern in the control group, existing results (e.g., \citealp{robins1994estimation}) can help identify $\E\{Y_{1}(0)\}$. However, identifying $\E\{Y_{1}(1)\}$ requires considerable effort as the component $\E\{Y_{1}(1,0)\}$ borrows the available information from the control group to the treated group requested by J2R, which differs from the traditional approaches where the identification only relies on the observed data in the same group, resulting in one of the main contributions in our paper.

We give some intuition about the identification formulas below. The intuition also helps when we extend our framework to the longitudinal setting. Theorem \ref{thm:iden_1time} (a) describes that for any subject in the target population, the individual treatment effect will be zero when missingness is involved, as J2R entails that the individual will always take the control therapy and thus have the same outcome mean regardless of the assigned treatment; if the outcome is fully observed, the individual treatment effect given the baseline covariates will be $\mu_{1}^{1}(X)-\mu_{1}^{0}(X)$. Taking the expectation over the response status in the treatment group results in the overall marginal treatment effect. Theorem \ref{thm:iden_1time} (b) creates the pseudo-observed outcome $R_{1}Y_{1}+(1-R_{1})\mu_{1}^{0}(X)$ from imputing the missing component by the outcome mean under J2R. The standard inverse probability weighting (IPW; \citealp{imbens2004nonparametric}) method is then applied to adjust for the confounding effect using the propensity score. In Theorem \ref{thm:iden_1time} (c), the first term adjusted by $A/e(X)$ targets the participants who are still observed in the assigned treatment group, which corresponds to $\E\left\{ \pi_{1}(1,X)\mu_{1}^{1}(X)\right\} $. The second term marginalizes the multiplication between $\pi_{1}(1,X)$ and the IPW-based transformed outcome $(1-A)R_{1}Y_{1}\big/\left[\left\{ 1-e(X)\right\} \pi_{1}(0,X)\right]$, which measures the conditional control group mean $\mu_{1}^{0}(X)$, quantifies the difference between the borrowed information in the treated group from the control group and the information in the control group, and matches $\E\left\{ \pi_{1}(1,X)\mu_{1}^{0}(X)\right\} $ in Theorem \ref{thm:iden_1time} (a).

\subsection{Estimation based on the identification formulas}

We introduce additional notations for convenience. Let $\pr_{n}$
be the empirical average, i.e., $\pr_{n}(U)$ = $n^{-1}\sum_{i=1}^{n}U_{i}$
for any variable $U$. Under the parametric modeling framework, let
$e(X;\alpha)$, $\mu_{1}^{a}(X;\beta)$, and $\pi_{1}(a,X;\gamma)$
be the working models of $e(X)$, $\mu_{1}^{a}(X)$, and $\pi_{1}(a,X)$,
where $\alpha,\beta,\gamma$ are the model parameters. Suppose the
model parameter estimates $(\widehat{\alpha},\widehat{\beta},\widehat{\gamma})$
converge to their probability limits $\left(\alpha^{*},\beta^{*},\gamma^{*}\right)$.
Denote the true model parameters $\left(\alpha_{0},\beta_{0},\gamma_{0}\right)$
and the true models $\{e(X),\mu_{1}^{a}(X),\pi_{1}(a,X):a=0,1\}$
for shorthand. To illustrate model specifications, we use $\mathcal{M}$
with the subscripts ``ps'', ``om'', and ``rp'' to denote the correctly
specified propensity score, outcome mean, and response probability,
respectively. Under $\mathcal{M}_{\text{ps}}$, $e(X;\alpha^{*})=e(X)$;
under $\mathcal{M}_{\text{om}}$, $\mu_{1}^{a}(X;\beta^{*})=\mu_{1}^{a}(X)$;
under $\mathcal{M}_{\text{rp}}$, $\pi_{1}(a,X;\gamma^{*})=\pi_{1}(a,X)$.
We use $+$ to indicate the correct specification of more than one
model and $\cup$ to indicate that at least one model is correctly
specified, e.g., $\mathcal{M}_{\text{rp+om}}\cup\mathcal{M}_{\text{ps}}$
implies that the response probability and outcome mean are correct
or the propensity score is correct. The estimators are obtained by replacing 
$\{e(X),\pi_{1}(a,X),\allowbreak\mu_{1}^{a}(X):a=0,1\}$ with the
estimated models $\{e(X;\widehat{\alpha}),\pi_{1}(a,X;\widehat{\gamma}),\mu_{1}^{a}(X;\widehat{\beta}):a=0,1\}$
and the expectation with the empirical average. 
\begin{example}\label{exmp: est_1time}
The estimators motivated by
the identification formulas in Theorem \ref{thm:iden_1time} are: 
\begin{enumerate}
\item The response probability-outcome mean (rp-om) estimator: $\widehat{\tau}_{\text{rp-om}}=\pr_{n}\left[\pi_{1}(1,X;\widehat{\gamma})\left\{ \mu_{1}^{1}(X;\widehat{\beta})-\mu_{1}^{0}(X;\widehat{\beta})\right\} \right].$
The estimator is consistent under $\mathcal{M}_{\text{rp+om}}$.
\item The propensity score-outcome mean (ps-om) estimator: 
\[
\widehat{\tau}_{\text{ps-om}}=\pr_{n}\left[\frac{2A-1}{e(X;\widehat{\alpha})^{A}\left\{ 1-e(X;\widehat{\alpha})\right\} ^{1-A}}\left\{ R_{1}Y_{1}+(1-R_{1})\mu_{1}^{0}(X;\widehat{\beta})\right\} \right].
\]
The estimator is consistent under $\mathcal{M}_{\text{ps+om}}$. 
\item The propensity score-response probability (ps-rp) estimator: 
\[
\widehat{\tau}_{\text{ps-rp}}=\pr_{n}\left\{ \frac{A}{e(X;\widehat{\alpha})}R_{1}Y_{1}-\frac{1-A}{1-e(X;\widehat{\alpha})}\frac{\pi_{1}(1,X;\widehat{\gamma})}{\pi_{1}(0,X;\widehat{\gamma})}R_{1}Y_{1}\right\} .
\]
The estimator is consistent under $\mathcal{M}_{\text{ps+rp}}$. 
\end{enumerate}
\end{example} 

The estimators $\widehat{\tau}_{\text{ps-om}}$ and $\widehat{\tau}_{\text{ps-rp}}$
involve taking the inverse of the estimated propensity score or response
probability, which may produce extreme values when they are close
to 0 or 1. To mitigate the issue, we seek an alternative version of
the inverse probability weighting estimators by normalizing the weights
\citep{lunceford2004stratification}. The exact forms of the normalized
estimators $\widehat{\tau}_{\text{ps-om-N}}$ and $\widehat{\tau}_{\text{ps-rp-N}}$
are given in Web Appendix C.1. 

\subsection{EIF and the EIF-based estimators}

Based on the three different identification formulas and the motivated
estimators, it is possible to combine the three sets of model components
in one identification formula. In the subsection, we first compute
the EIF for the ATE under J2R to get a new identification formula
and then give the resulting EIF-based estimators. 

\begin{theorem}\label{thm: eif_1time}
Under Assumptions \ref{assump:1.1time}--\ref{assump:4.1time}, suppose
that there exists $\varepsilon>0,$ such that $\varepsilon<\big\{ e(X),\pi_{1}(a,X)\big\}<1-\varepsilon$
for all $X$ and $a$, the EIF for $\tau_{1}^{\jtr}$ is 
\begin{align*}
\varphi_{1}^{\jtr}(V;\mathbb{P})=\left\{ \frac{A}{e(X)}-\frac{1-A}{1-e(X)}\frac{\pi_{1}(1,X)}{\pi_{1}(0,X)}\right\} R_{1}\left\{ Y_{1}-\mu_{1}^{0}(X)\right\} -\frac{A-e(X)}{e(X)}\pi_{1}(1,X)\left\{ \mu_{1}^{1}(X)-\mu_{1}^{0}(X)\right\} -\tau_{1}^{\jtr}.
\end{align*}
\end{theorem} 

By the fact that the mean of the EIF is zero, we can
obtain another identification formula for the ATE,{} which motivates
the EIF-based estimator $\widehat{\tau}_{\text{tr}}$ as{} 
\begin{eqnarray*}
{\widehat{\tau}_{\text{tr}}}=\mathbb{P}_{n}\Bigg[\left\{ \frac{A}{e(X;\widehat{\alpha})}-\frac{1-A}{1-e(X;\widehat{\alpha})}\frac{\pi_{1}(1,X;\widehat{\gamma})}{\pi_{1}(0,X;\widehat{\gamma})}\right\} R_{1}\left\{ Y_{1}-\mu_{1}^{0}(X;\widehat{\beta})\right\} -\frac{A-e(X;\widehat{\alpha})}{e(X;\widehat{\alpha})}\pi_{1}(1,X;\widehat{\gamma})\left\{ \mu_{1}^{1}(X;\widehat{\beta})-\mu_{1}^{0}(X;\widehat{\beta})\right\} \Bigg].
\end{eqnarray*}

We provide the normalized estimator $\widehat{\tau}_{\text{tr-N}}$ to reduce the impact of extreme weights in Web Appendix C.2. We also consider employing calibration (e.g., \citealp{hainmueller2012entropy}; \citealp{zhao2019covariate}; \citealp{lee2021improving}) to improve the covariate balance and mitigate the outliers. Using the logistic link function, we estimate the weights by solving the optimization problem $\min_{w_{i}\geq0}\sum_{i=1}^{n}(w_{i}-1)\log(w_{i}-1)-w_{i}$ subject to $\sum_{i:A_{i}=1}w_{a_{1},i}h(X_{i})=n^{-1}\sum_{i=1}^{n}h(X_{i})$ to compute the weights $w_{i}=w_{a_{1},i}$ when $A=1$; subject to $\sum_{i:A_{i}=0}w_{a_{0},i}h(X_{i})=n^{-1}\sum_{i=1}^{n}h(X_{i})$ to compute the weights $w_{i}=w_{a_{0},i}$ when $A=0$; and subject to $\sum_{i:R_{i}=1}w_{r_{1},i}\allowbreak h(X_{i})=n^{-1}\sum_{i=1}^{n}h(X_{i})$ to compute the weights $w_{i}=w_{r_{1},i}$ when $R_{1}=1$. Here, $h(X)$ is any function of covariates. For example, one may incorporate the first two moments of the covariates to achieve a balance in both
means and variances. The calibration-based estimator $\widehat{\tau}_{\text{tr-C}}$
is given in Web Appendix C.2. {While $\widehat{\tau}_{\text{tr-N}}$ and  $\widehat{\tau}_{\text{tr-C}}$ enjoy superior finite-sample performance by mitigating extreme weights, the three EIF-based estimators are asymptotically equivalent with theoretical guarantees \citep{zhao2019covariate}.}

Connecting with the well-known robustness results under MAR in the missing data literature (e.g., \citealp{robins1995semiparametric,bang2005doubly}), the constructed EIF-motivated estimators distinguish themselves due to the discrepancy in outcome mean profiles between observed individuals and dropouts in the treated group envisioned by J2R, which
is further explained in Web Appendix E. Interestingly, they achieve better robust properties compared to the existing doubly robust estimators under MAR. 

As we will explain in the next subsection, the estimators reach $n^{1/2}$-consistency
if any two of the three models are correct when using a parametric
modeling strategy, or if the convergence rate of any model is not
less than $n^{-1/4}$ when using flexible models. We call this property
triple robustness. 

\subsection{Triple robustness}

We focus on investigating the asymptotic properties of $\widehat{\tau}_{\text{tr}}$.
Theorem \ref{thm: tr_1time_cons} explores the triple robustness of $\widehat{\tau}_{\text{tr}}$ under a parametric
modeling strategy on the nuisance functions. 
\begin{theorem}\label{thm: tr_1time_cons}

Under Assumptions \ref{assump:1.1time}--\ref{assump:4.1time}, suppose
that there exists $\varepsilon>0,$ such that{} $\varepsilon<\big\{ e(X;\alpha^{*}),\allowbreak e(X;\widehat{\alpha}),\pi_{1}(a,X;\gamma^{*}),\pi_{1}(a,X;\widehat{\gamma})\big\}<1-\varepsilon$
for all $X$ and $a$ almost surely, the estimator $\widehat{\tau}_{\text{tr}}$
is triply robust in the sense that it is consistent for $\tau_{1}^{\jtr}$
under $\mathcal{M}_{\text{rp+om}}\cup\mathcal{M}_{\text{ps+om}}\cup\mathcal{M}_{\text{ps+rp}}$.
Moreover, $\widehat{\tau}_{\text{tr}}$ achieves the semiparametric efficiency
bound under $\mathcal{M}_{\text{ps+rp+om}}$.

\end{theorem}

Theorem \ref{thm: tr_1time_cons} requires the true and estimated
propensity scores and response probabilities bounded away from 0 and
1 to reduce the extreme values 
\citep{robins1995semiparametric}. Given that the EIF-based estimators,
the estimators in Example \ref{exmp: est_1time}, and their normalized
versions are asymptotically linear, the variance estimators can be
computed by nonparametric bootstrap.

When the models for the nuisance functions are difficult to obtain
parametrically, one can turn to more flexible modeling strategies
such as semiparametric models like generalized additive models (GAM;
\citealp{hastie2017generalized}) or machine learning models to get
the estimated models $\left\{ \widehat{e}(X),\widehat{\pi}_{1}(a,X),\widehat{\mu}_{1}^{a}(X):a=0,1\right\} $.
To illustrate the convergence rate of the estimated models, denote
$\lVert U\rVert=\{\E(U^{2})\}^{1/2}$ as the $L_{2}$-norm of the
random variable $U$. Suppose the convergence rates are $\lVert\widehat{e}(X)-e(X)\rVert=o_{\mathbb{P}}(n^{-c_{e}}),\lVert\widehat{\mu}_{1}^{a}(X)-\mu_{1}^{a}(X)\rVert=o_{\mathbb{P}}(n^{-c_{\mu}})$
and $\lVert\widehat{\pi}_{1}(a,X)-\pi_{1}(a,X)\rVert=o_{\mathbb{P}}(n^{-c_{\pi}})$.
Denote $\widehat{\mathbb{P}}$ as the estimated distribution of the observed
data. Theorem \ref{thm: tr_1time_rate} illustrates the asymptotic
distribution of the EIF-based estimator. 

\begin{theorem}\label{thm: tr_1time_rate}
Under Assumptions \ref{assump:1.1time}--\ref{assump:4.1time}, suppose
that there exists $\varepsilon>0,$ such that{} $\varepsilon<\big\{ e(X),\widehat{e}(X),\allowbreak\pi_{1}(a,X),\widehat{\pi}_{1}(a,X)\big\}<1-\varepsilon$
for all $X$ and $a$ almost surely, and the nuisance functions and
their estimators take values in Donsker classes. Assume $\lVert\varphi_{1}^{\text{\jtr}}(V;\widehat{\mathbb{P}})-\varphi_{1}^{\text{\jtr}}(V;\mathbb{P})\rVert=o_{\mathbb{P}}(1)$.
Then, $\widehat{\tau}_{\text{tr}}=\tau_{1}^{\text{\jtr}}+n^{-1}\sum_{i=1}^{n}\varphi_{1}^{\text{\jtr}}(V_{i};\mathbb{P})+\text{Rem}(\widehat{\mathbb{P}},\mathbb{P})+o_{\mathbb{P}}(n^{-1/2}),$
where
\begin{align*}
\text{Rem}(\widehat{\mathbb{P}},\mathbb{P}) & =\E\bigg[\left\{ \frac{e(X)}{\widehat{e}(X)}-1\right\} \left\{ \pi_{1}(1,X)\mu_{1}^{1}(X)-\widehat{\pi}_{1}(1,X)\widehat{\mu}_{1}^{1}(X)\right\} +\left\{ 1-\frac{1-e(X)}{1-\widehat{e}(X)}\frac{\pi_{1}(0,X)}{\widehat{\pi}_{1}(0,X)}\right\} \widehat{\pi}_{1}(1,X)\big\{\mu_{1}^{0}(X)\\
 & \quad-\widehat{\mu}_{1}^{0}(X)\big\}+\left\{ \pi_{1}(1,X)-\widehat{\pi}_{1}(1,X)\right\} \left\{ \mu_{1}^{0}(X)-\frac{e(X)}{\widehat{e}(X)}\widehat{\mu}_{1}^{0}(X)\right\} \bigg].
\end{align*}

If $\text{Rem}(\widehat{\mathbb{P}},\mathbb{P})=o_{\mathbb{P}}(n^{-1/2})$,
then $n^{1/2}\left(\widehat{\tau}_{\text{tr}}-\tau_{1}^{\text{\jtr}}\right)\xrightarrow{d}\mathcal{N}\left(0,\mathbb{V}\left\{ \varphi_{1}^{\text{\jtr}}(V;\mathbb{P})\right\} \right)$,
where the asymptotic variance of $\widehat{\tau}_{\text{tr}}$ reaches
the semiparametric efficiency bound and $\mathbb{V}(\cdot)$ represents
the variance. \end{theorem}

The requirement of Donsker classes controls the complexity of the
nuisance functions and their estimators \citep{kennedy2016semiparametric},
which can be further relaxed using cross-fitting \citep{chernozhukov2018double}.
Theorem \ref{thm: tr_1time_rate} invokes the triple robustness in
terms of rate convergence when using flexible models, presented by
the following corollary. 

\begin{corollary}\label{cor: tr_1time_rate} Under the assumptions
in Theorem \ref{thm: tr_1time_rate}, suppose $\lVert\varphi_{1}^{\text{\jtr}}(V;\widehat{\mathbb{P}})-\varphi_{1}^{\text{\jtr}}(V;\mathbb{P})\rVert=o_{\mathbb{P}}(1)$,
and further suppose that there exists $0<M<\infty$, such that $\pr\bigg(\max\Big\{\big|\widehat{\mu}_{1}^{0}(X)\big|,\big|\widehat{\mu}_{1}^{1}(X)\big|,\allowbreak\Big|\{1-e(X)\}/\{1-\widehat{e}(X)\}\Big|\Big\}\leq M\bigg)=1$,
then $\widehat{\tau}_{\text{tr}}-\tau_{1}^{\text{\jtr}}=O_{\mathbb{P}}\left(n^{-1/2}+n^{-c}\right)$,
where $c=\min(c_{e}+c_{\mu},c_{e}+c_{\pi},c_{\mu}+c_{\pi})$.
\end{corollary} 

The additional uniformly bounded condition for the
estimated outcome means and the ratio $\{1-e(X)\}/\{1-\widehat{e}(X)\}$,
which originates from \citet{kennedy2016semiparametric} and holds
in most clinical studies, guarantees an upper bound for $\text{Rem}(\widehat{\mathbb{P}},\mathbb{P})$.
Corollary \ref{cor: tr_1time_rate} provides alternative approaches
to reach a $n^{1/2}$-rate consistency of the estimator. The nuisance
functions can converge at a slower rate no less than $n^{-1/4}$ using
flexible models. 

\section{Longitudinal data with monotone missingness \label{sec:longi}}

Next, we focus on the longitudinal setting and introduce additional
notations. Suppose the longitudinal data contain $t$ time points.
Let $Y_{s,i}$ be the outcome at time $s$, $H_{s-1,i}=(X_{i}^{\text{T}},Y_{1,i},\cdots,Y_{s-1,i})^{\text{T}}$
be the historical information at time $s$ for $s=2,\cdots,t$, and
$H_{0,i}=X_{i}$. When missingness is involved, denote $R_{s,i}$
as the response indicator at time $s$ and $D_{i}$ as the dropout
time. Let $R_{0,i}=1$, indicating the baseline covariates $H_{0,i}$
are always observed. We assume a monotone missingness pattern, i.e.,
if the individual drops out at time $s$, we would expect $R_{s,i}=\cdots=R_{t,i}=0$.
By monotone missingness, there exists a one-to-one relationship between
the dropout time $D_{i}$ and the vector of response indicators $\left(R_{0,i},\cdots,R_{t,i}\right)$
as $D_{i}=\sum_{s=0}^{t}R_{s,i}$ for all $i$. Assume the full data
$\left\{ X_{i},A_{i},R_{1,i},Y_{1,i},\cdots,R_{t,i},Y_{t,i}:i=1,\cdots,n\right\} $
are independent and identically distributed. We omit the subscript
$i$ again for simplicity. Let $V=(X,A,R_{1}Y_{1},R_{1},\cdots,R_{t}Y_{t},R_{t})$
be the vector of all observed variables and follow the observed data
distribution $\mathbb{P}$. Extending the potential outcomes framework,
we define $R_{s}(a)$ as the potential response indicator if the subject
received treatment $a$ at time $s$, $D(a)$ as the potential dropout
time if the subject received treatment $a$, $Y_{s}(a,d)$ as the
potential outcome if the subject received treatment $a$ at time $s$
with the occurrence of dropout at time $d$. Similar to the cross-sectional
setting, we simplify the potential outcome $Y_{s}\{a,D(a)\}=Y_{s}(a)$
using the composition assumption, which assumes that the potential
outcome with $A=a$ and the potential outcome with $A=a$ and the
dropout time $D$ to be the value it would have been if $A=a$ are
the same. Due to the natural constraint that future dropouts do not
affect the current and past outcomes, we have $Y_{s}(a,t+1)=Y_{s}(a,s')$
for any $s<s'<t+1$ and $D(a)=\sum_{s=0}^{t}R_{s}(a)$. We extend
Assumptions \ref{assump:1.1time}--\ref{assump:4.1time} to the context
of longitudinal data with monotone missingness. 

\begin{Assumption}[Treatment ignorability]\label{assump:1.longi}
$A\indep\left\{ R_{s}(a),D(a),Y_{s}(a,d)\right\} \mid X$, for all
$a,s$ and $d$. \end{Assumption} 
\begin{Assumption}[Causal consistency]\label{assump:3.longi}

$R_{s}=R_{s}(A)$, $D=D(A)$, and $Y_{s}=Y_{s}\left\{ A,D(A)\right\} $,
for all $s$.

\end{Assumption} 
\begin{Assumption}[Partial ignorability of missingness]\label{assump:2.longi}
$R_{s}(0)\indep Y_{s'}(0,d)\mid H_{s-1}$, for all $s'\geq s$ and
$d$. \end{Assumption} 

\begin{Assumption}[J2R for the outcome mean]\label{assump:4.longi}
$\E\left\{ Y_{s}(1,d)\mid D(1)=d,H_{d-1}\right\} =\E\{Y_{s}(0)\mid H_{d-1},R_{d-1}=1\}$,
for all $s\geq d$. 
\end{Assumption} 


In the longitudinal setting, Assumption \ref{assump:4.longi} indicates
a transition from the active treatment to the control group for the
dropouts while preserving the historical treatment benefit. 
\citet{white2020causal} develop a similar potential outcomes framework
for CBI in longitudinal clinical trials. However, their assumptions
about the causal model are much stronger, as they assume a linear
relationship between future and historical outcomes. Our proposed
framework does not rely on any modeling assumptions and is more flexible
in practice. 
In this section, all results degenerate to the ones
in cross-sectional studies when $t=1$. 

\subsection{Three identification formulas under J2R}

In most longitudinal clinical studies, the endpoint of interest is
the ATE measured by the mean difference at the last time point between
the two groups. Therefore, the ATE can be expressed 
as $\tau_{t}^{\jtr}=\E\{Y_{t}(1)-Y_{t}(0)\}$.
Define the propensity score $e(H_{s-1})=\pr(A=1\mid H_{s-1},R_{s-1}=1)$,
the response probability $\pi_{s}(a,H_{s-1})=\pr(R_{s}=1\mid H_{s-1},R_{s-1}=1,A=a)$,
the longitudinal outcome mean $\mu_{t}^{a}(H_{s-1})=\E\left\{ \mu_{t}^{a}(H_{s})\mid H_{s-1},R_{s}=1,A=a\right\} $
with $\mu_{t}^{a}(H_{t})=Y_{t}$, and the pattern mean $g_{s+1}^{1}(H_{l-1})=\E\big\{\pi_{l+1}(1,H_{l})\allowbreak g_{s+1}^{1}(H_{l})\mid H_{l-1},R_{l}=1,A=1\big\}$
for $l=1,\cdots,s-1$ with $g_{s+1}^{1}(H_{s-1})=\E\Big[\big\{1-\pi_{s+1}(1,H_{s})\big\}\mu_{t}^{0}(H_{s})\mid H_{s-1},R_{s}=1,A=1\Big]$
if we let $\pi_{t+1}(1,H_{t})=0$. The pattern mean characterizes
the weighted outcome mean in each dropout pattern under the pattern-mixture
model (\citealp{little1993pattern}). In addition, denote $\overline{\pi}_{s}(a,H_{s-1})=\prod_{k=1}^{s}\pi_{k}(a,H_{k-1})$
as the cumulative response probability for the individual observed
at time $s$, for $s=1,\cdots,t$. 
The following theorem provides three identification formulas for longitudinal
data with monotone missingness under J2R. 

\begin{theorem}\label{thm:iden_longi} Under Assumptions \ref{assump:1.longi}--\ref{assump:4.longi},
suppose that there exists $\varepsilon>0$, such that{} $\varepsilon<\big\{ e(H_{s-1}),\pi_{s}(a,H_{s-1})\big\}<1-\varepsilon$
for all $H_{s-1}$ and $a$ with $s=1,\cdots,t$, the following identification
formulas hold for the ATE under J2R: 
\begin{enumerate}
\item[(a)] Based on the response probability and pattern mean, $\tau_{t}^{\text{\jtr}}=\mathbb{E}\left[\pi_{1}(1,H_{0})\left\{ \sum_{s=1}^{t}g_{s+1}^{1}(H_{0})-\mu_{t}^{0}(H_{0})\right\} \right].$ 
\item[(b)] Based on the propensity score and outcome mean, 
\[
\tau_{t}^{\text{\jtr}}=\mathbb{E}\left[\frac{2A-1}{e(H_{0})^{A}\left\{ 1-e(H_{0})\right\} ^{1-A}}\left\{ R_{t}Y_{t}+\sum_{s=1}^{t}R_{s-1}(1-R_{s})\mu_{t}^{0}(H_{s-1})\right\} \right].
\]
\item[(c)] Based on the propensity score and response probability, 
\begin{align*}
\tau_{t}^{\text{\jtr}}=\mathbb{E}\bigg(\frac{A}{e(H_{0})}R_{t}Y_{t}+\frac{1-A}{1-e(H_{0})}\Big[\sum_{s=1}^{t}\overline{\pi}_{s-1}(0,H_{s-2})\left\{ 1-\pi_{s}(1,H_{s-1})\right\} \delta(H_{s-1})-1\Big]\frac{R_{t}Y_{t}}{\overline{\pi}_{t}(0,H_{t-1})}\bigg),
\end{align*}
where $\delta(H_{s-1})=\left\{ e(H_{s-1})\big/e(H_{0})\right\} \Big/\left[\left\{ 1-e(H_{s-1})\right\} \big/\left\{ 1-e(H_{0})\right\} \right].$ 
\end{enumerate}
\end{theorem}

\subsection{Estimation based on the identification formulas}
Similar to the cross-sectional setting, the estimators can be obtained
by replacing the functions $\big\{ e(H_{s-1}),\pi_{s}(a,H_{s-1}),\mu_{t}^{a}(H_{s-1}),g_{s+1}^{1}(H_{l-1}):l=1,\cdots,s\text{ and }s=1,\cdots,t;a=0,1\big\}$
with the estimated functions $\big\{\widehat{e}(H_{s-1}),\widehat{\pi}_{s}(a,H_{s-1}),\widehat{\mu}_{t}^{a}(H_{s-1}),\widehat{g}_{s+1}^{1}(H_{l-1}):l=1,\cdots,s\text{ and }\allowbreak s=1,\cdots,t;a=0,1\big\}$
and the expectation with the empirical average. Compared to the cross-sectional
case, obtaining the ATE estimator here involves fitting sequential
models at each time point. However, the complex iterated form of $g_{s+1}^{1}(H_{l-1})$
is infeasible to model parametrically. We consider using more flexible
models such as semiparametric or machine learning models. Denote $\widehat{\mathbb{P}}$
as the estimated distribution of the observed data $V$. Suppose the
nuisance functions have convergence rates $\lVert\widehat{e}(H_{s-1})-e(H_{s-1})\rVert=o_{\mathbb{P}}(n^{-c_{e}}),\lVert\widehat{\mu}_{t}^{a}(H_{s-1})-\mu_{t}^{a}(H_{s-1})\rVert=o_{\mathbb{P}}(n^{-c_{\mu}}),\lVert\widehat{\pi}_{s}(a,H_{s-1})-\pi_{s}(a,H_{s-1})\rVert=o_{\mathbb{P}}(n^{-c_{\pi}})$
for any $H_{s-1}$, and $\lVert\widehat{g}_{s+1}^{1}(H_{l-1})-g_{s+1}^{1}(H_{l-1})\lVert=o_{\mathbb{P}}(n^{-c_{g}})$
for any $H_{l-1}$, when $l=1,\cdots,s$; $s=1,\cdots,t$ and $a=0,1$. 

\begin{example}\label{exmp: est_longi} The estimators motivated
by the identification formulas in Theorem \ref{thm:iden_longi} are: 
\begin{enumerate}
\item The response probability-pattern mean (rp-pm) estimator: $\widehat{\tau}_{\text{rp-pm}}=\mathbb{\pr}_{n}\Big[\widehat{\pi}_{1}(1,H_{0})\big\{\sum_{s=1}^{t}\widehat{g}_{s+1}^{1}(H_{0})-\widehat{\mu}_{t}^{0}(H_{0})\big\}\Big],$
where $\widehat{g}_{s+1}^{1}(H_{l-1})=\widehat{\E}\left\{ \widehat{\pi}_{l+1}(1,H_{l})\widehat{g}_{s+1}^{1}(H_{l})\mid H_{l-1},R_{l}=1,A=1\right\} $
for $l=1,\cdots,s-1$ and $\widehat{g}_{s+1}^{1}(H_{s-1})=\widehat{\E}\left[\left\{ 1-\widehat{\pi}_{s+1}(1,H_{s})\right\} \widehat{\mu}_{t}^{0}(H_{s})\mid H_{s-1},R_{s}=1,A=1\right]$
if let $\widehat{\pi}_{t+1}(1,H_{t})\allowbreak=0$. 
\item The ps-om estimator: 
\begin{align*}
\widehat{\tau}_{\text{ps-om}} & =\pr_{n}\left[\frac{2A-1}{\widehat{e}(H_{0})^{A}\left\{ 1-\widehat{e}(H_{0})\right\} ^{1-A}}\big\{ R_{t}Y_{t}+\sum_{s=1}^{t}R_{s-1}(1-R_{s})\widehat{\mu}_{t}^{0}(H_{s-1})\big\}\right].
\end{align*}
\item The ps-rp estimator: 
\begin{align*}
\widehat{\tau}_{\text{ps-rp}}=\pr_{n}\bigg(\frac{A}{\widehat{e}(H_{0})}R_{t}Y_{t}+\frac{1-A}{1-\widehat{e}(H_{0})}\big[\sum_{s=1}^{t}\overline{\widehat{\pi}}_{s-1}(0,H_{s-2})\{1-\widehat{\pi}_{s}(1,H_{s-1})\}\widehat{\delta}(H_{s-1})-1\big]\frac{R_{t}Y_{t}}{\overline{\widehat{\pi}}_{t}(0,H_{t-1})}\bigg),
\end{align*}
where $\widehat{\delta}(H_{s-1})=\left\{ \widehat{e}(H_{s-1})\big/\widehat{e}(H_{0})\right\} \Big/\left[\left\{ 1-\widehat{e}(H_{s-1})\right\} \big/\left\{ 1-\widehat{e}(H_{0})\right\} \right]$. 
\end{enumerate}
\end{example} 

The impact of the extreme propensity score and response probability
weights is more pronounced in the longitudinal setting with an extended
long period of follow-up. To mitigate the influence, we consider the
normalized estimators $\widehat{\tau}_{\text{ps-om-N}}$ and $\widehat{\tau}_{\text{ps-rp-N}}$.
The estimation procedure is similar to the one in \citet{bang2005doubly},
which involves fitting the models recursively. The propensity score $\big\{ e(H_{s-1}):s=1,\cdots,t\big\}$ and response probability ${\color{black}\big\{\pi_{s}(a,H_{s-1}):s=1,\cdots,t\big\}}$ incorporate all the available information $H_{s-1}$. For the outcome mean ${\color{black}\big\{\mu_{t}^{a}(H_{s-1}):s=1,\cdots,t\big\}}$, we begin from the observed data at the last time point and use the predicted values to regress on the observed data recursively in backward order. For the pattern mean $\big\{ g_{s+1}^{1}(H_{l-1}):l=1,\cdots,s\text{ and }s=1,\cdots,t\big\}$, the product of the predicted values $\left\{ 1-\widehat{\pi}_{s+1}(1,H_{s})\right\} $ and $\widehat{\mu}_{t}^{0}(H_{s})$ is regressed on the historical information $H_{s-1}$ at time $s$. The resulting predicted value $\widehat{g}_{s+1}^{1}(H_{s-1})$ multiplied by the predicted response probability $\widehat{\pi}_{s}(1,H_{s-1})$ then severs as the outcome in the model $g_{s+1}^{1}(H_{s-2})$ to regress on the observed data at time $s-1$. Note that the estimated pattern mean will have good performance only if both the response probability and the outcome mean are well-approximated. 

\subsection{EIF and the EIF-based estimators }
Similar to cross-sectional studies, we derive the EIF for $\tau_{t}^{\text{\jtr}}$
to motivate a new estimator.

\begin{theorem}\label{thm: eif_longi}
Under Assumptions \ref{assump:1.longi}--\ref{assump:4.longi},
suppose that there exists $\varepsilon>0$, such that{} $\varepsilon<\big\{ e(H_{s-1}),\pi_{s}(a,H_{s-1})\big\}<1-\varepsilon$
for all $H_{s-1}$ and $a$ with $s=1,\cdots,t$, the EIF for $\tau_{t}^{\text{\jtr}}$
is
\begin{align*}
\varphi_{t}^{\text{\jtr}}(V;\mathbb{P}) & =\frac{A}{e(H_{0})}\big\{ R_{t}Y_{t}+\sum_{s=1}^{t}R_{s-1}(1-R_{s})\mu_{t}^{0}(H_{s-1})\big\}-\tau_{t}^{\text{\jtr}}\\
 & +\{1-\frac{A}{e(H_{0})}\}\Big[\pi_{1}(1,H_{0})\sum_{s=1}^{t}g_{s+1}^{1}(H_{0})+\big\{1-\pi_{1}(1,H_{0})\big\}\mu_{t}^{0}(H_{0})\Big]-\mu_{t}^{0}(H_{0})\\
 & +\frac{1-A}{1-e(H_{0})}\sum_{s=1}^{t}\Big[\sum_{k=1}^{s}\overline{\pi}_{k-1}(0,H_{k-2})\{1-\pi_{k}(1,H_{k-1})\}\delta(H_{k-1})-1\Big]\frac{R_{s}}{\overline{\pi}_{s}(0,H_{s-1})}\big\{\mu_{t}^{0}(H_{s})-\mu_{t}^{0}(H_{s-1})\big\}.
\end{align*}
\end{theorem} 

Solving $\mathbb{E}\{\varphi_{t}^{\text{\jtr}}(V;\mathbb{P})\}=0$
yields another identification formula of $\tau_{t}^{\text{\jtr}}$
and motivates the EIF-based estimator $\widehat{\tau}_{\text{mr}}$ by
plugging in the estimated nuisance functions as
\begin{align*}
\widehat{\tau}_{\text{mr}} & =\mathbb{\pr}_{n}\bigg(\frac{A}{\widehat{e}(H_{0})}\big\{ R_{t}Y_{t}+\sum_{s=1}^{t}R_{s-1}(1-R_{s})\widehat{\mu}_{t}^{0}(H_{s-1})\big\}+\left\{ 1-\frac{A}{\widehat{e}(H_{0})}\right\} \left[\widehat{\pi}_{1}(1,H_{0})\sum_{s=1}^{t}\widehat{g}_{s+1}^{1}(H_{0})+\big\{1-\widehat{\pi}_{1}(1,H_{0})\big\}\widehat{\mu}_{t}^{0}(H_{0})\right]\\
 & -\widehat{\mu}_{t}^{0}(H_{0})+\frac{1-A}{1-\widehat{e}(H_{0})}\sum_{s=1}^{t}\left[\sum_{k=1}^{s}\widehat{\overline{\pi}}_{k-1}(0,H_{k-2})\{1-\widehat{\pi}_{k}(1,H_{k-1})\}\widehat{\delta}(H_{k-1})-1\right]\frac{R_{s}}{\widehat{\overline{\pi}}_{s}(0,H_{s-1})}\big\{\widehat{\mu}_{t}^{0}(H_{s})-\widehat{\mu}_{t}^{0}(H_{s-1})\big\}\bigg).
\end{align*}
In addition, one can consider the normalized estimator $\widehat{\tau}_{\text{mr-N}}$
or the calibration-based estimator $\widehat{\tau}_{\text{mr-C}}$ to
mitigate the extreme weights, as elaborated in Web Appendix C.5. 

\subsection{Multiple robustness}
To simplify the notations, let $E_{0,l-1}(\cdot;H_{s}):=\E\big\{\cdots\E(\cdot\mid H_{l-1},R_{l}=1,A=0)\cdots\mid H_{s},R_{s+1}=1,A=0\big\}$
be the function of $(l-s)$ layers conditional expectations, with
the conditions beginning from $(H_{l-1},R_{l}=1,A=0)$ to $(H_{s},R_{s+1}=1,A=0)$,
and $E_{1,s-1}(\cdot;H_{0}):=\E\big\{\cdots\E\left(\cdot\mid H_{s-1},R_{s}=1,A=1\right)\allowbreak\cdots\mid H_{0},R_{1}=1,A=1\big\}$
be the function of $s$ layers conditional expectations, with the
conditions beginning from $(H_{s-1},R_{s}=1,A=1)$ to $(H_{0},R_{1}=1,A=1)$.
Denote $g_{\widehat{\mu},s+1}^{1}(H_{l-1})=\E\big\{\pi_{l+1}(1,H_{l})g_{\widehat{\mu},s+1}^{1}(H_{l})\mid H_{l-1},R_{l}=1,A=1\big\}$
for $l=1,\cdots,s-1$, and $g_{\widehat{\mu},s+1}^{1}(H_{s-1})=\E\Big[\big\{1-\pi_{s+1}(1,H_{s})\big\}\widehat{\mu}_{t}^{0}(H_{s})\mid H_{s-1},R_{s}=1,A=1\Big]$
for $s=1,\cdots,t$, i.e., we only estimate the outcome mean in the
pattern mean model $g_{s+1}^{1}(H_{l-1})$. The asymptotic properties
of $\widehat{\tau}_{\text{mr}}$ are presented in the following theorem. 
\begin{theorem}\label{thm:eif_longi2} 
Under Assumptions \ref{assump:1.longi}--\ref{assump:4.longi},
suppose that there exists $\varepsilon>0$, such that $\varepsilon<\big\{ e(H_{s-1}), \widehat{e}(H_{s-1}),\allowbreak \pi_{s}(a,H_{s-1}), \widehat{\pi}_{s}(a,H_{s-1})\big\}<1-\varepsilon$
for all $H_{s-1}$ and $a$ with $s=1,\cdots,t$, and the nuisance
functions and their estimators take values in Donsker classes. Assume
$\lVert\varphi_{t}^{\text{\jtr}}(V;\widehat{\mathbb{P}})-\varphi_{t}^{\text{\jtr}}(V;\mathbb{P})\rVert=o_{\mathbb{P}}(1)$.
Then, $\widehat{\tau}_{\text{mr}}=\tau_{t}^{\text{\jtr}}+n^{-1}\sum_{i=1}^{n}\varphi_{t}^{\text{\jtr}}(V_{i};\mathbb{P})+\text{Rem}(\widehat{\mathbb{P}},\mathbb{P})+o_{\mathbb{P}}(n^{-1/2}),$
where 
\begin{align*}
\text{Rem}(\widehat{\mathbb{P}},\mathbb{P}) & =\E\Bigg(\left\{ \frac{e(H_{0})}{\widehat{e}(H_{0})}-1\right\} \Big[\pi_{1}(1,H_{0})g_{t+1}^{1}(H_{0})-\widehat{\pi}_{1}(1,H_{0})\widehat{g}_{t+1}^{1}(H_{0})+\sum_{s=1}^{t-1}\big\{\pi_{1}(1,H_{0})g_{\widehat{\mu},s+1}^{1}(H_{0})\\
 & -\widehat{\pi}_{1}(1,H_{0})\widehat{g}_{s+1}^{1}(H_{0})\big\}\Big]+\sum_{s=1}^{t-1}\sum_{l=s+1}^{t}E_{1,s-1}\Bigg\{ E_{0,l-1}\bigg(\overline{\pi}_{s}(1,H_{s-1})\Big[\frac{1-e(H_{0})}{1-\widehat{e}(H_{0})}\{1-\widehat{\pi}_{s+1}(1,H_{s})\}\frac{\widehat{\delta}(H_{s-1})}{\delta(H_{s-1})}\\
 & \qquad\qquad\prod_{k=s+1}^{l}\frac{\pi_{k}(0,H_{k-1})}{\widehat{\pi}_{k}(0,H_{k-1})}-\{1-\pi_{s+1}(1,H_{s})\}\Big]\left\{ \widehat{\mu}_{t}^{0}(H_{l})-\widehat{\mu}_{t}^{0}(H_{l-1})\right\} ;H_{s}\bigg);H_{0}\Bigg\}\\
 & +\left\{ \widehat{\pi}_{1}(1,H_{0})-\pi_{1}(1,H_{0})\right\} \left\{ \frac{e(H_{0})}{\widehat{e}(H_{0})}\widehat{\mu}_{t}^{0}(H_{0})-\mu_{t}^{0}(H_{0})\right\} \\
 & +\widehat{\pi}_{1}(1,H_{0})\sum_{s=1}^{t}E_{0,s-1}\left[\left\{ 1-\frac{1-e(H_{0})}{1-\widehat{e}(H_{0})}\frac{\overline{\pi}_{s}(0,H_{s-1})}{\overline{\widehat{\pi}}_{s}(0,H_{s-1})}\right\} \left\{ \widehat{\mu}_{t}^{0}(H_{s})-\widehat{\mu}_{t}^{0}(H_{s-1})\right\} ;H_{0}\right]\Bigg).
\end{align*}

If $\text{Rem}(\widehat{\mathbb{P}},\mathbb{P})=o_{\mathbb{P}}(n^{-1/2})$,
then $n^{1/2}\left(\widehat{\tau}_{\text{mr}}-\tau_{t}^{\text{\jtr}}\right)\xrightarrow{d}\mathcal{N}\left(0,\mathbb{V}\left\{ \varphi_{t}^{\text{\jtr}}(V;\mathbb{P})\right\} \right)$,
where the asymptotic variance of $\widehat{\tau}_{\text{mr}}$ reaches
the semiparametric efficiency bound.

\end{theorem} 
The semiparametric efficiency bound prompts the EIF-based variance
estimator 
$\widehat{\mathbb{V}}(\widehat{\tau}_{\text{mr}})=n^{-2}\sum_{i=1}^{n}\allowbreak\big\{\varphi_{t}^{\text{\jtr}}(V_{i};\widehat{\mathbb{P}})-\widehat{\tau}_{\text{mr}}\big\}^{2}.$
In practice, the Wald-type confidence interval (CI) tends to have
narrower intervals which can be anti-conservative \citep{boos2013essential}.
Symmetric t bootstrap CI \citep{hall1988symmetric} is considered
to improve the coverage. In each bootstrap iteration from $b=1,\cdots,B$,
where $B$ is the total number of bootstrap replicates, we compute
$T^{*(b)}=(\widehat{\tau}^{(b)}-\widehat{\tau})/\widehat{\mathbb{V}}^{1/2}(\widehat{\tau}^{(b)})$
to get the estimated bootstrap distribution. The $95\%$ symmetric
t bootstrap CI of $\tau_{t}^{\text{\jtr}}$ is obtained by $(\widehat{\tau}-c^{*}\widehat{\mathbb{V}}^{1/2}(\widehat{\tau}),\widehat{\tau}+c^{*}\widehat{\mathbb{V}}^{1/2}(\widehat{\tau}))$,
where $c^{*}$ is the $95\%$ quantile of $\{|T^{*(b)}|:b=1,\cdots,B\}$.
Theorem \ref{thm:eif_longi2} motivates the following corollary, which
addresses the multiple robustness of $\widehat{\tau}_{\text{mr}}$ in
terms of the convergence rate under flexible modeling strategies. 

\begin{corollary}\label{cor:rate_coverge} 
Under the assumptions
in Theorem \ref{thm:eif_longi2}, suppose $\lVert\varphi_{t}^{\text{\jtr}}(V;\widehat{\mathbb{P}})-\varphi_{t}^{\text{\jtr}}(V;\mathbb{P})\rVert=o_{\mathbb{P}}(1)$
and there exists $0<M<\infty$, such that 
\[
\pr\left(\max\left\{ \Big|\frac{e(H_{0})}{\widehat{e}(H_{0})}\Big|,\Big|\mu_{t}^{0}(H_{0})\Big|,\Big|\widehat{g}_{s+1}^{1}(H_{0})\Big|,\Big|\frac{\{1-e(H_{0})\}\widehat{\delta}(H_{s-1})}{\{1-\widehat{e}(H_{0})\}\delta(H_{s-1})}\Big|\right\} \leq M\right)=1
\]
for $s=1,\cdots,t$, then $\widehat{\tau}_{\text{mr}}-\tau_{t}^{\text{\jtr}}=O_{\mathbb{P}}\left(n^{-1/2}+n^{-c}\right)$,
where $c=\min\big\{ c_{e}+c_{\mu},c_{e}+c_{\pi},c_{\mu}+c_{\pi},c_{e}+c_{g}\big\}$.
\end{corollary}
Similar to the cross-sectional setting, even if the nuisance functions
converge at a lower rate, we can still obtain a $n^{1/2}$-rate consistency.
An additional function $\big\{ g_{s+1}^{1}(H_{l-1}):l=1,\cdots,s\text{ and }s=1,\cdots,t\big\}$
is involved, whose convergence rate may
be harder to control as it incorporates the estimation of both the
outcome mean and response probability.
\section{Simulation study\label{sec:simu}}

\subsection{Cross-sectional setting}

We first conduct the simulation in a cross-sectional setting to evaluate
the finite-sample performance of the proposed estimators.
Set the sample size as 500. The covariates $X\in\mathbb{R}^{5}$ are
generated by $X_{j}\sim N(0.25,1)$ for $j=1,\cdots,4$ and $X_{5}\sim\text{Bernoulli}(0.5)$.
Consider a nonlinear transformation of the covariates and denote $Z_{j}=\{X_{j}^{2}+2\sin(X_{j})-1.5\}/\sqrt{2}$
for $j=1,\cdots,4$ and $Z_{5}=X_{5}$. We generate $A\mid X\sim\text{Bernoulli}\text{\{\ensuremath{e(X)}}\}$,
where $\text{logit}\{e(X)\}=0.1\sum_{j=1}^{4}Z_{j}$; $R_{1}\mid(X,A=a)\sim\text{Bernoulli}\{\pi_{1}(a,X)$\},
where $\text{logit\{\ensuremath{\pi_{1}(a,X)}\}}=(2a-1)\ensuremath{\sum_{j=1}^{5}Z_{j}/6}$;
and $Y_{1}\mid(X,A=a,R_{1}=1)\sim N\{\mu_{1}^{a}(X),1\}$, where $\mu_{1}^{a}(X)=(2+a)\sum_{j=1}^{5}Z_{j}/6$.
The true ATE $\tau_{1}^{\jtr}=0.0680$. To evaluate the robustness
of the estimators, we consider two model specifications
of the propensity score, response probability, and outcome mean. Specifically,
we fit the corresponding parametric models with the covariates $Z$
as the correctly specified models or with the covariates $X$ as the
misspecified models.

We compare the estimators from Example \ref{exmp: est_1time} and
their normalized versions with the three EIF-based estimators. The
first moment of the covariates $Z$ is incorporated in the calibration.
The estimators are assessed in terms of the point estimation, coverage
rates of the $95\%$ CI, and mean CI lengths under 8 scenarios, each
of which relies on whether the propensity score, response probability,
or outcome mean is correctly specified. We compute the variance
estimates $\widehat{\mathbb{V}}_{1}$ of the estimators by nonparametric
bootstrap with $B=100$ and use the $95\%$ Wald-type CI as $(\widehat{\tau}-1.96\widehat{\mathbb{V}}_{1}^{1/2},\widehat{\tau}+1.96\widehat{\mathbb{V}}_{1}^{1/2})$.
Figure \ref{fig:point_1time} shows the point estimation results based
on 1000 Monte Carlo simulations.
When three models are correctly specified, all the estimators are
unbiased. For the estimators without triple robustness, they are biased
when at least one of their required models is misspecified; while
the three EIF-based estimators verify triple robustness since they
are unbiased when any two of the three models are correct. Normalization
mitigates the impact of extreme weights and results in smaller variations.
Moreover, calibration produces a more steady estimator. The coverage
rates and mean CI lengths are presented in Table \ref{tab:sim_1time},
which match the observations we make from Figure \ref{fig:point_1time}.
All estimators have satisfactory coverage rates when their required
models are correct. Among the EIF-based estimators, the coverage rates are close to the empirical value when any
two of the three models are correct, with the smallest mean CI length
produced by $\widehat{\tau}_{\text{tr-C}}$.

\begin{figure}
\centering{}\caption{Performance of the estimators in the cross-sectional setting under
8 different model specifications, where ps, rp, and om are shorthands
for the propensity score, response probability, and outcome mean; ``yes'' denotes the
correct model with the nonlinear covariates $Z$, while ``no'' denotes
the wrong model with the linear covariates $X$.
In the x-axis, tr, tr-N, and tr-C denote the three EIF-based estimators
$\widehat{\tau}_{\text{tr}}$, $\widehat{\tau}_{\text{tr-N}}$, and $\widehat{\tau}_{\text{tr-C}}$;
psrp and psrp-N denote the estimators $\widehat{\tau}_{\text{ps-rp}}$
and $\widehat{\tau}_{\text{ps-rp-N}}$; psom and psom-N denote the estimators
$\widehat{\tau}_{\text{ps-om}}$ and $\widehat{\tau}_{\text{ps-om-N}}$; and
rpom denotes the estimator $\widehat{\tau}_{\text{rp-om}}$ in Example
\ref{exmp: est_1time}. 
\label{fig:point_1time}}
\includegraphics[scale=0.5]{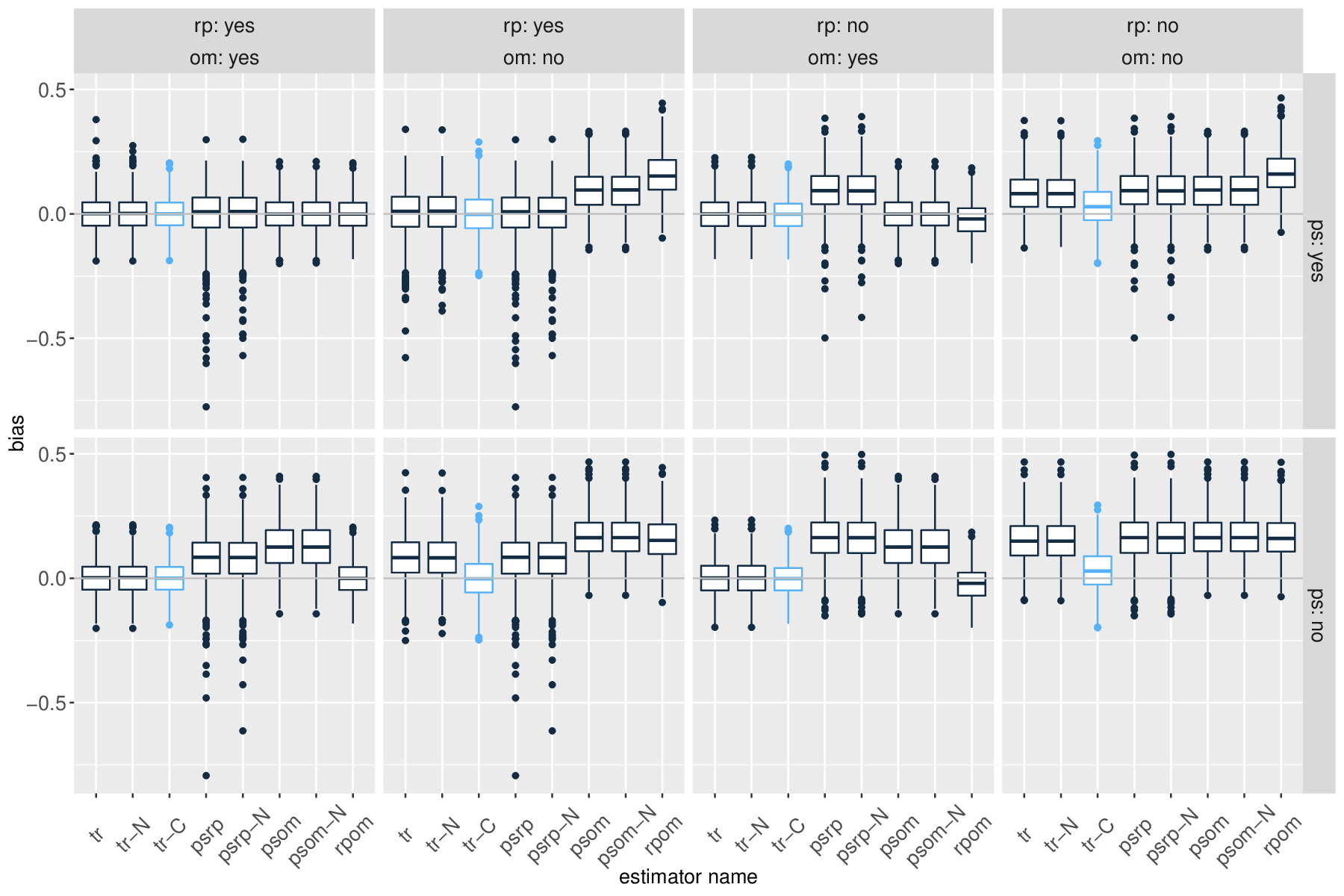} 
\end{figure}
\begin{table}[ht]
\global\long\def\arraystretch{0.75}%
\caption{Coverage rates and mean CI lengths in the cross-sectional setting
under 8 different model specifications, where PS, RP, and OM are shorthands
for the propensity score, response probability, and outcome mean; ``yes'' denotes the
correct model with the nonlinear covariates $Z$, while ``no'' denotes
the wrong model with the linear covariates $X$.\label{tab:sim_1time}}

\noindent \resizebox{\textwidth}{!}{%
\begin{tabular}{>{\centering}m{0.05\columnwidth}>{\centering}m{0.05\columnwidth}>{\centering}m{0.05\columnwidth}>{\centering}p{0.01\columnwidth}>{\centering}p{0.07\columnwidth}>{\centering}p{0.07\columnwidth}>{\centering}p{0.07\columnwidth}>{\centering}p{0.07\columnwidth}>{\centering}p{0.07\columnwidth}>{\centering}p{0.07\columnwidth}>{\centering}p{0.07\columnwidth}>{\centering}p{0.07\columnwidth}}
\hline 
\multicolumn{3}{c}{Model specification} &  &  &  & \multicolumn{4}{c}{Coverage rate (\%)} &  & \tabularnewline
 &  &  &  & \multicolumn{8}{c}{(Mean CI length, \%)}\tabularnewline
\hline 
PS  & RP  & OM  &  & $\widehat{\tau}_{\text{tr}}$  & $\widehat{\tau}_{\text{tr-N}}$  & $\widehat{\tau}_{\text{tr-C}}$  & $\widehat{\tau}_{\text{ps-rp}}$  & $\widehat{\tau}_{\text{ps-rp-N}}$  & $\widehat{\tau}_{\text{ps-om}}$  & $\widehat{\tau}_{\text{ps-om-N}}$  & $\widehat{\tau}_{\text{rp-om}}$\tabularnewline
\hline 
yes  & yes  & yes  &  & 94.7  & 94.7  & \textbf{94.4}  & 95.7  & 95.5  & 94.9  & 94.9  & \textbf{94.3}\tabularnewline
 &  &  &  & (30.9)  & (29.5)  & \textbf{(28.5)}  & (59.9)  & (41.8)  & (29.1)  & (29.0)  & \textbf{(28.2)}\tabularnewline
yes  & yes  & no  &  & 95.3  & 94.8  & \textbf{94.3}  & 95.7  & 95.5  & 80.6  & 80.6  & 57.6\tabularnewline
 &  &  &  & (41.8)  & (36.1)  & \textbf{(33.7)}  & (59.9)  & (41.8)  & (33.1)  & (33.1)  & (34.0)\tabularnewline
yes  & no  & yes  &  & 94.1  & 94.1  & \textbf{94.2}  & 79.7  & 80.0  & 94.9  & 94.9  & 93.5\tabularnewline
 &  &  &  & (28.8)  & (28.3)  & \textbf{(28.2)}  & (36.7)  & (35.3)  & (29.1)  & (29.0)  & (27.7)\tabularnewline
no  & yes  & yes  &  & 94.4  & 94.4  & \textbf{94.4}  & 85.8  & 86.0  & 72.8  & 72.9  & \textbf{94.3}\tabularnewline
 &  &  &  & (29.5)  & (29.1)  & \textbf{(28.5)}  & (45.7)  & (40.7)  & (37.9)  & (37.9)  & \textbf{(28.2)}\tabularnewline
yes  & no  & no  &  & 83.0  & 82.9  & 93.1  & 79.7  & 80.0  & 80.6  & 80.6  & 53.4\tabularnewline
 &  &  &  & (32.7)  & (32.3)  & (33.8)  & (36.7)  & (35.3)  & (33.1)  & (33.1)  & (34.1)\tabularnewline
no  & yes  & no  &  & 84.1  & 83.9  & 94.3  & 85.8  & 86.0  & 53.8  & 53.8  & 57.6\tabularnewline
 &  &  &  & (37.4)  & (35.9)  & (33.7)  & (45.7)  & (40.7)  & (34.6)  & (34.7)  & (34.0)\tabularnewline
no  & no  & yes  &  & 94.6  & 94.6  & 94.2  & 56.1  & 56.1  & 72.8  & 72.9  & 93.5\tabularnewline
 &  &  &  & (29.2)  & (29.2)  & (28.2)  & (38.0)  & (37.4)  & (37.9)  & (37.9)  & (27.7)\tabularnewline
no  & no  & no  &  & 61.3  & 61.3  & 93.1  & 56.1  & 56.1  & 53.8  & 53.8  & 53.4\tabularnewline
 &  &  &  & (35.1)  & (34.9)  & (33.8)  & (38.0)  & (37.4)  & (34.7)  & (34.7)  & (34.1)\tabularnewline
\hline 
\end{tabular}} 
\end{table}

\subsection{Longitudinal setting}
We further evaluate the performance of the proposed estimators in
longitudinal studies under J2R. Consider the data with two follow-up
time points. We choose the sample size as $n=1000$,
generate the same covariates $X\in\mathbb{R}^{5}$, and use the same
transformation on the covariates to construct $Z\in\mathbb{R}^{5}$
as the one in the cross-sectional setting. The treatments are generated
by $A\mid X\sim\text{Bernoulli}\text{\{\ensuremath{e(X)}}\}$, where
$\text{logit}\{e(X)\}=0.1\sum_{j=1}^{4}Z_{j}$. The observed indicators
and the longitudinal outcomes are generated in time order. Specifically,
at the first time point, we generate $R_{1}\mid(X,A=a)\sim\text{Bernoulli}\left\{ \pi_{1}(a,X)\right\} $,
where $\text{logit\{\ensuremath{\pi_{1}(a,X)}\}}=5(2a-1)\ensuremath{\sum_{j=1}^{4}Z_{j}/9}$,
and $Y_{1}\mid(X,R_{1}=1,A=a)\sim N\{\mu_{1}^{a}(X),1\}$, where $\mu_{1}^{a}(X)=(2+a)\big\{\sum_{j=1}^{4}\log(Z_{j}^{2})+\sum_{j=1}^{5}Z_{j}\big\}/6$;
at the second time point, we generate $R_{2}\mid(X,Y_{1},R_{1}=1,A=a)\sim\text{Bernoulli}\left\{ \pi_{2}(a,X,Y_{1})\right\} $,
where $\text{logit\{\ensuremath{\pi_{2}(a,X,Y_{1})}\}}=(2a-1)\big\{\sum_{j=1}^{4}\log(Z_{j}^{2})+Z_{5}+0.1Y_{1}\big\}/6$,
and $Y_{2}\mid(X,Y_{1},R_{2}=1,A=a)\sim N\{\mu_{2}^{a}(X,Y_{1}),1\}$,
where $\mu_{2}^{a}(X,Y_{1})=(2+a)\big(\sum_{j=1}^{5}Z_{j}+Y_{1}\big)/3$.
The true ATE $\tau_{2}^{\jtr}=0.3198$. Since the models are infeasible
to approximate parametrically, we apply GAM using smooth splines,
where we incorporate the original covariates $X$ in each nuisance
function and employ calibration.

We compare the performance of the point estimation, coverage rates
of the $95\%$ CI, and mean CI lengths for the proposed estimators.
For the EIF-based estimators, we compute the $95\%$ symmetric
t bootstrap CIs with a larger number of bootstrap
replicates as $B=500$.
For other estimators,
since multiple robustness is not guaranteed, we use nonparametric
bootstrap to obtain their bootstrap percentile intervals. 
Figure \ref{fig:sim_longi} shows the point estimation results
based on 1000 Monte Carlo simulations. All the EIF-based estimators
are unbiased, and the one involving calibration has the smallest variation,
alleviating the impact of extreme values. Other estimators suffer
from different levels of bias. Table \ref{tab:sim_2time} supports
the superiority of the EIF-based estimators in terms of coverage rates
and mean CI lengths. 

\begin{figure}
\centering{}\caption{Performance of the estimators in the longitudinal setting. In the
x-axis, mr, mr-N, and mr-C denote the three EIF-based estimators $\widehat{\tau}_{\text{mr}}$,
$\widehat{\tau}_{\text{mr-N}}$, and $\widehat{\tau}_{\text{mr-C}}$; psrp
and psrp-N denote the estimators $\widehat{\tau}_{\text{ps-rp}}$ and
$\widehat{\tau}_{\text{ps-rp-N}}$; psom and psom-N denote the estimators
$\widehat{\tau}_{\text{ps-om}}$ and $\widehat{\tau}_{\text{ps-om-N}}$; and
rppm denotes the estimator $\widehat{\tau}_{\text{rp-pm}}$ in Example
\ref{exmp: est_longi}. 
\label{fig:sim_longi}}
\includegraphics[scale=0.3]{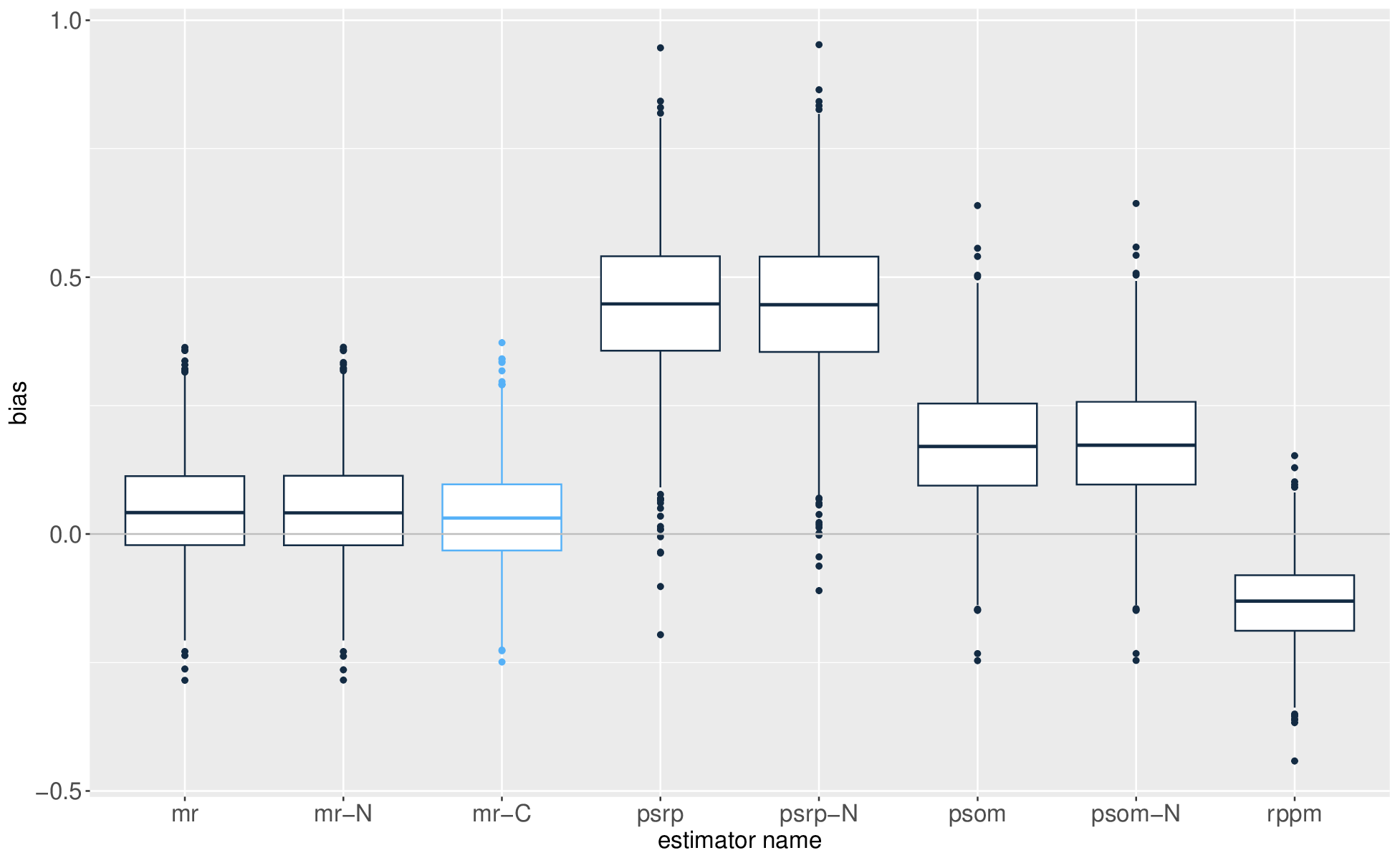} 
\end{figure}

\begin{table}[ht]
\centering{} 
\global\long\def\arraystretch{0.75}%
\centering{}\caption{Coverage rates and mean CI lengths in the longitudinal setting.\label{tab:sim_2time}}
\begin{tabular}{>{\centering}p{0.14\columnwidth}>{\centering}p{0.14\columnwidth}>{\centering}p{0.14\columnwidth}}
\toprule 
Estimator   & Coverage rate ($\%$)   & Mean CI length ($\%$)\tabularnewline
\midrule 
$\widehat{\tau}_{\text{mr}}$   & 95.4  & 43.8\tabularnewline
$\widehat{\tau}_{\text{mr-N}}$   & 95.2  & 43.7\tabularnewline
$\widehat{\tau}_{\text{mr-C}}$   & \textbf{96.6 } & \textbf{42.8}\tabularnewline
$\widehat{\tau}_{\text{ps-rp}}$   & 26.7  & 72.0\tabularnewline
$\widehat{\tau}_{\text{ps-rp-N}}$   & 27.1  & 72.1\tabularnewline
$\widehat{\tau}_{\text{ps-om}}$   & 93.1  & 51.8\tabularnewline
$\widehat{\tau}_{\text{ps-om-N}}$   & 92.5  & 52.0\tabularnewline
$\widehat{\tau}_{\text{rp-pm}}$   & 77.2  & 39.5\tabularnewline
\bottomrule
\end{tabular}
\end{table}
\section{Application\label{sec:app}}
We apply our proposed estimators to analyze the data from the antidepressant
clinical trial introduced in Section \ref{sec:intro} under J2R. Apart
from the partially observed HAMD-17 scores, a categorical variable
indicating the investigation sites is observed for all individuals.
For the nuisance functions involved in the proposed estimators, we
fit GAM sequentially.
To handle the extreme weights, calibration is applied,
where we include the first two moments of the history.
We compute the $95\%$ symmetric t bootstrap CIs for the three EIF-based
estimators and the $95\%$ bootstrap percentile intervals for other
estimators, with $B=500$.

Table \ref{tab:real} presents the analysis results. All the estimators
have similar point estimates. However, we detect a relatively obvious
difference in the values between $\widehat{\tau}_{\text{ps-rp}}$ and
$\widehat{\tau}_{\text{ps-rp-N}}$, indicating the existence of extreme
weights. The weight distributions in Web Appendix H validate the presence of outliers at weeks 4, 6, and 8 in
the control group. Calibration stabilizes the estimation results and
leads to a smaller CI compared to the other two EIF-based estimators.
Although $\widehat{\tau}_{\text{ps-om}}$ and $\widehat{\tau}_{\text{rp-pm}}$
have similar point estimates and narrower CIs compared to the EIF-based
estimators, they rely on a good approximation of their corresponding
two models, which may not be guaranteed in practice due to the lack
of consistency under slow convergences of the estimated nuisance functions.
The EIF-based estimators are preferred with a trade-off between bias
and precision since they have a guaranteed multiple robustness in
terms of rate convergence. All the resulting $95\%$ CIs indicate
a statistically significant treatment effect.
\begin{table}
\caption{Analysis of the HAMD-17 data for the ATE under J2R. \label{tab:real}}
\renewcommand{\arraystretch}{0.75}
\centering{}%
\begin{tabular}{cccc}
\toprule 
Estimator & Point estimate & $95\%$ CI & CI length\tabularnewline
\midrule
$\widehat{\tau}_{\text{mr}}$ & -1.93 & (-3.63, -0.24) & 3.39\tabularnewline
$\widehat{\tau}_{\text{mr-N}}$ & -1.93 & (-3.62, -0.25) & 3.37\tabularnewline
$\widehat{\tau}_{\text{mr-C}}$ & \textbf{-1.71} & \textbf{(-3.25, -0.16)} & \textbf{3.09}\tabularnewline
$\widehat{\tau}_{\text{ps-rp}}$ & -2.05 & (-4.08, -0.50) & 3.57\tabularnewline
$\widehat{\tau}_{\text{ps-rp-N}}$ & -1.61 & (-3.74, -0.07) & 3.67\tabularnewline
$\widehat{\tau}_{\text{ps-om}}$ & -1.74 & (-3.20, -0.25) & 2.95\tabularnewline
$\widehat{\tau}_{\text{ps-om-N}}$ & -1.75 & (-3.18, -0.22) & 2.96\tabularnewline
$\widehat{\tau}_{\text{rp-pm}}$ & -1.78 & (-3.18, -0.25) & 2.93\tabularnewline
\bottomrule
\end{tabular}
\end{table}

\section{Conclusion \label{sec:conclusion}}
Evaluating the treatment effect under an assumed MNAR assumption has been receiving growing interest in both primary and sensitivity analyses in longitudinal studies. We propose a potential outcomes framework to describe the missing data scenario pre-specified as J2R to identify the ATE. The new estimator is constructed with the help of the EIF, combining the propensity score, response probability, outcome mean, and pattern mean. It allows flexible modeling strategies such as semiparametric or machine learning models, with the good property of multiple robustness in that it achieves $n^{1/2}$-consistency and asymptotic normality even when the models converge at a slower rate such as $n^{-1/4}$. The proposed estimators can be applied in a wide range of clinical studies including randomized trials and observational studies, and are extendable to other MNAR-related scenarios.

The model assumptions are relaxed in the established semiparametric framework. However, standard untestable assumptions about the missing components are imposed to identify the ATE. The assumed outcome mean for the dropouts under J2R prevents introducing external parameters 
and reveals its credibility 
for the drug with a short-term effect. Meanwhile, it may produce a conservative treatment effect evaluation 
if the active treatment is supposed to be superior \citep{liu2016analysis}. Its wide applicability 
appeals to regulatory agencies. 

Our framework relies on a monotone missingness pattern for the longitudinal data, which however may not always be the case in reality. \citet{sun2018inverse} provide an inverse probability weighting approach to deal with the MAR data with non-monotone missingness patterns. It is possible to extend our method to handle intermittent missing data using their proposed approaches. We leave it as a future research direction.

The construction of the multiply robust estimators is based on continuous longitudinal outcomes. Possibilities exist in the extension of the proposed framework to broader types of outcomes. 
For example, \citet{yang2020smim} consider the $\delta$-adjusted and control-based models to evaluate the treatment effect on the survival outcomes; \citet{tang2018controlled} extends CBI to binary and ordinal longitudinal outcomes using sequential generalized linear models. These extensions shed light on establishing new multiply
robust estimators with the use of our idea.

\section*{Acknowledgements}
Yang is partially supported by the NSF SES 2242776, NIH 1R01AG066883 and 1R01ES031651.

\section*{Supplementary Materials}

Supplementary materials contain technical details in Sections 2--5. The R package to implement the method is available at \url{https://github.com/siyi48/mrJ2R}.

\section*{Data Availability}
The data that support the findings in this paper are openly available
in the Drug Information Association Missing Data at \url{https://www.lshtm.ac.uk/research/centres-projects-groups/missing-data#dia-missing-data}
collected by \citet{mallinckrodt2014recent}.

\bibliographystyle{Chicago}
\bibliography{template/biomsample_bib}       


\newpage
\begin{center}
\textbf{Supplementary Materials for "Multiply robust estimators in
longitudinal studies with missing data under control-based imputation" by Liu et al.}
\par\end{center}

\begin{center}
Siyi Liu, Shu Yang, Yilong Zhang, Guanghan (Frank) Liu
\par\end{center}

{}\pagenumbering{arabic} 
\renewcommand*{\thepage}{S\arabic{page}}

{}\setcounter{lemma}{0}  
\global\long\def\thelemma{\textup{S}\arabic{lemma}}%
{}\setcounter{equation}{0}  
\global\long\def\theequation{S\arabic{equation}}%
{}\setcounter{section}{0}  
\global\long\def\thesection{Web Appendix \Alph{section}}%
{} 
\global\long\def\thesubsection{Web Appendix \Alph{section}.\arabic{subsection}}%
{}\setcounter{table}{0}  
\global\long\def\tablename{Web Table}%
{}  
\global\long\def\thetable{\arabic{table}}%
{}\setcounter{figure}{0}  
\global\long\def\figurename{Web Figure}%
{}  
\global\long\def\thefigure{\arabic{figure}}%
{}\setcounter{theorem}{0}  
\global\long\def\thethm{\textup{S}\arabic{thm}}%
{}\setcounter{corollary}{0}  
\global\long\def\thecorollary{\textup{S}\arabic{corollary}}%
The supplementary material contains technical details, additional
simulation, and real-data application results. \ref{sec:supp_iden}
provides proof for the identification formulas provided in Theorems
1 and 5. \ref{sec:supp_eif} presents detailed derivations of the
EIFs in Theorems 2 and 6. \ref{sec:supp_est} gives additional estimators
and the detailed estimation steps. \ref{sec:supp_tr} consists of
the proofs regarding multiple robustness. 
\ref{sec:supp_aipw} connects the proposed multiply robust estimators with the existing results in the literature. 
\ref{sec:supp_sen} gives a sensitivity analysis framework to test the robustness of results against the partial ignorability of missingness assumption.
\ref{sec:supp_simu} contains
additional simulation results. \ref{sec:supp_app} shows additional
notes on the real-data application.

\section{Proof of the identification formulas \label{sec:supp_iden}}

\subsection{Proof of Theorem 1 \label{subsec:supp_iden_1time}}

{}We first prove the equivalence of the three identification formulas,
then prove the validity of the identification formula (a) in Theorem
1.

{}Denote 
\begin{align*}
E_{1,1} & =\mathbb{E}\big[\pi_{1}(1,X)\{\mu_{1}^{1}(X)-\mu_{1}^{0}(X)\}\big];\\
E_{2,1} & =\mathbb{E}\big[\frac{A}{e(X)}\{R_{1}Y_{1}+(1-R_{1})\mu_{1}^{0}(X)\}-\frac{1-A}{1-e(X)}\{R_{1}Y_{1}+(1-R_{1})\mu_{1}^{0}(X)\}\big];\\
E_{3,1} & =\mathbb{E}\big[\frac{A}{e(X)}R_{1}Y_{1}-\frac{1-A}{1-e(X)}\frac{\pi_{1}(1,X)}{\pi_{1}(0,X)}R_{1}Y_{1}\big].
\end{align*}
\begingroup\makeatletter\def\f@size{11}\check@mathfonts
{}Note that $E_{1,1}=E_{2,1}$ holds since 
\begin{align*}
\mathbb{E}\left[\frac{A}{e(X)}\{R_{1}Y_{1}+(1-R_{1})\mu_{1}^{0}(X)\}\right] & =\mathbb{E}\left[\frac{\E\left(A\mid X\right)}{e(X)}\E\left\{ R_{1}Y_{1}+(1-R_{1})\mu_{1}^{0}(X)\mid X,A=1\right\} \right]\\
 & =\mathbb{E}\left(\frac{\E\left(A\mid X\right)}{e(X)}\left[\E(R_{1}\mid X,A=1)\mu_{1}^{1}(X)+\{1-\E(R_{1}\mid X,A=1)\}\mu_{1}^{0}(X)\right]\right)\\
 & =\mathbb{E}\left(\frac{\E\left(A\mid X\right)}{e(X)}\left[\pi_{1}(1,X)\mu_{1}^{1}(X)+\{1-\pi_{1}(1,X)\}\mu_{1}^{0}(X)\right]\right)\\
 & =\mathbb{E}\left[\pi_{1}(1,X)\mu_{1}^{1}(X)+\{1-\pi_{1}(1,X)\}\mu_{1}^{0}(X)\right],
\end{align*}
And similarly, 
\begin{align*}
\E\left[\frac{1-A}{1-e(X)}\{R_{1}Y_{1}+(1-R_{1})\mu_{1}^{0}(X)\}\right] & =\E\left[\frac{1-A}{1-e(X)}\{R_{1}\mu_{1}^{0}(X)+(1-R_{1})\mu_{1}^{0}(X)\}\right]\\
 & =\E\left\{ \frac{1-A}{1-e(X)}\mu_{1}^{0}(X)\right\} \\
 & =\E\{\mu_{1}^{0}(X)\}.
\end{align*}
Then, we have $E_{2,1}=\mathbb{E}\left[\pi_{1}(1,X)\mu_{1}^{1}(X)+\{1-\pi_{1}(1,X)\}\mu_{1}^{0}(X)\right]-\E\{\mu_{1}^{0}(X)\}=E_{1,1}$.

{}Also note that $E_{1,1}=E_{3,1}$ holds since 
\begin{align*}
E_{3,1} & =\E\left\{ \frac{\E\left(A\mid X\right)}{e(X)}E(R_{1}\mid X,A=1)\E(Y_{1}\mid X,R_{1}=1,A=1)\right\} \\
 & \quad-\E\left\{ \frac{\E\left(1-A\mid X\right)}{1-e(X)}\frac{\pi_{1}(1,X)}{\pi_{1}(0,X)}E(R_{1}\mid X,A=0)\E(Y_{1}\mid X,R_{1}=1,A=0)\right\} \\
 & =\E\left\{ \frac{\E(A\mid X)}{e(X)}\pi_{1}(1,X)\mu_{1}^{1}(X)\right\} -\E\left\{ \frac{\E\left(1-A\mid X\right)}{1-e(X)}\pi_{1}(1,X)\mu_{1}^{0}(X)\right\} \\
 & =\E\left\{ \pi_{1}(1,X)\mu_{1}^{1}(X)-\pi_{1}(1,X)\mu_{1}^{0}(X)\right\} =E_{1,1}.
\end{align*}

{}We proceed to prove the validity of the identification formula
(a) in Theorem 1. Denote $\tau_{1,1}=\E[Y_{1}^{\text{}}\{1,R(1)\}${]}
and $\tau_{0,1}=\E[Y_{1}^{\text{}}\{0,R(0)\}]$. Note that 
\begin{align*}
\tau_{1,1} & =\E\left[R_{1}(1)Y_{1}^{\text{}}(1,1)+\{1-R_{1}(1)\}Y_{1}(1,0)\right]\\
 & =\E\left[\E\left\{ R_{1}(1)\mid X\right\} \E\left\{ Y_{1}(1,1)\mid X,R_{1}(1)=1\right\} +\E\left\{ 1-R_{1}(1)\mid X\right\} \E\left\{ Y_{1}(1,0)\mid X,R_{1}(1)=0\right\} \right]\\
 & =\E\left[\E\left(R_{1}\mid X,A=1\right)\E\left\{ Y_{1}(1,1)\mid X,R_{1}(1)=1,A=1\right\} +\E\left(1-R_{1}\mid X,A=0\right)\E\left\{ Y_{1}(1,0)\mid X,R_{1}(1)=0\right\} \right]\\
 & \quad\text{(By A1, A3)}\\
 & =\E\left\{ \pi_{1}(1,X)\E\left(Y_{1}\mid X,R_{1}=1,A=1\right)+\left\{ 1-\pi_{1}(1,X)\right\} \E\left(Y_{1}\mid X,A=0\right)\right\} \text{(By A3, A4)}\\
 & =\E\left\{ \pi_{1}(1,X)\mu_{1}^{1}(X)+\{1-\pi_{1}(1,X)\}\mu_{1}^{0}(X)\right\} \text{(By A2)}.
\end{align*}
and 
\begin{align*}
\tau_{0,1} & =\E\left[R_{1}(0)Y_{1}^{\text{}}(0,1)+\{1-R_{1}(0)\}Y_{1}(0,0)\right]\\
 & =\E\left[\E\left\{ R_{1}(0)\mid X\right\} \E\left\{ Y_{1}(0,1)\mid X,R_{1}(0)=1\right\} +\E\left\{ 1-R_{1}(0)\mid X\right\} \E\left\{ Y_{1}(0,0)\mid X,R_{1}(0)=0\right\} \right]\\
 & =\E\Big[\E\left(R_{1}\mid X,A=0\right)\E\left\{ Y_{1}(0,1)\mid X,R_{1}(0)=1,A=0\right\} \\
 & \qquad+\E\left(1-R_{1}\mid X,A=0\right)\E\left\{ Y_{1}(0,0)\mid X,R_{1}(0)=0,A=0\right\} \Big]\text{(By A1, A3)}\\
 & =\E\left[\pi_{1}(0,X)\E\left(Y_{1}\mid A=0,R_{1}=1,X\right)+\left\{ 1-\pi_{1}(0,X)\right\} \mu_{1}^{0}(X)\right]\text{(By A3, A4)}\\
 & =\E\left[\pi_{1}(0,X)\mu_{1}^{0}(X)+\left\{ 1-\pi_{1}(0,X)\right\} \mu_{1}^{0}(X)\right]\text{(By A2)}\\
 & =\E\left\{ \mu_{1}^{0}(X)\right\} .
\end{align*}

{}Combine the two parts, we have 
\[
\tau_{1}^{\text{\jtr}}=\tau_{1,1}-\tau_{0,1}=\E\left\{ \pi_{1}(1,X)\mu_{1}^{1}(X)-\pi_{1}(1,X)\mu_{1}^{0}(X)\right\} =E_{1,1}.
\]

\subsection{Proof of Theorem 5 \label{subsec:supp_iden_longi}}

{}We first prove the equivalence of the three identification formulas,
then prove the validity of the identification formula (a) in Theorem
5.

{}Denote 
\begin{align*}
E_{1,t} & =\mathbb{E}\left[\pi_{1}(1,H_{0})\left\{ \sum_{s=1}^{t}g_{s+1}^{1}(H_{0})-\mu_{t}^{0}(H_{0})\right\} \right];\\
E_{2,t} & =\mathbb{E}\left[\frac{2A-1}{e(H_{0})^{A}\left\{ 1-e(H_{0})\right\} ^{1-A}}\left\{ R_{t}Y_{t}+\sum_{s=1}^{t}R_{s-1}(1-R_{s})\mu_{t}^{0}(H_{s-1})\right\} \right];\\
E_{3,t} & =\mathbb{E}\left(\frac{A}{e(H_{0})}R_{t}Y_{t}+\frac{1-A}{1-e(H_{0})}\left[\sum_{s=1}^{t}\bar{\pi}_{s-1}(0,H_{s-2})\left\{ 1-\pi_{s}(1,H_{s-1})\right\} \delta(H_{s-1})-1\right]\frac{R_{t}Y_{t}}{\bar{\pi}_{t}(0,H_{t-1})}\right).
\end{align*}
\endgroup

\begingroup\makeatletter\def\f@size{10}\check@mathfonts
{}To simplify the proof, we first introduce relevant lemmas.

{}\begin{lemma}\label{lemma:iden_mar}\par Under MAR, the group
mean can be identified using the sequential outcome means, i.e., $\E\left[Y_{t}^{\text{}}\{0,D(0)\}\right]=\E\left\{ \mu_{t}^{0}(H_{0})\right\} .$\par \end{lemma}

{}\begin{proof} Similar to the notations in the main text, we define
the pattern mean in the control group as $g_{s+1}^{0}(H_{l-1})=\E\big\{\pi_{l+1}(0,H_{l})g_{s+1}^{0}(H_{l})\mid H_{l-1},R_{l}=1,A=0\big\}$
for $l=1,\cdots,s-1$ with $g_{s+1}^{0}(H_{s-1})=\E\left[\left\{ 1-\pi_{s+1}(0,H_{s})\right\} \mu_{t}^{0}(H_{s})\mid H_{s-1},R_{s}=1,A=0\right]$
if we let $\pi_{t+1}(0,H_{t})=0$. Based on the pattern-mixture model
(PMM; \citealp{little1993pattern}) framework, we express the potential
outcome $Y_{t}^{\text{}}\{0,D(0)\}$ based on its potential dropout
pattern as $Y_{t}^{\text{}}\{0,D(0)\}=\sum_{s=1}^{t+1}\mathbb{I}\left\{ D(0)=s\right\} Y_{t}(0,s)$
and compute the expectation. For any $s\in\left\{ 2,\cdots,t+1\right\} $,
$\E\left[\mathbb{I}\left\{ D(0)=s\right\} Y_{t}(0,s)\right]$ is calculated
as 
\begin{align*}
 & =\E\left[R_{1}(0)\cdots R_{s-1}(0)\left\{ 1-R_{s}(0)\right\} Y_{t}(0,s)\right]\text{ (By the definition of \ensuremath{D})}\\
 & =\E\left(\E\left(R_{1}(0)\mid H_{0}\right)\E\left[R_{2}(0)\cdots R_{s-1}(0)\left\{ 1-R_{s}(0)\right\} Y_{t}(0,s)\mid H_{0},R_{1}(0)=1\right]\right)\\
 & =\E\Bigg\{\E\left(R_{1}(0)\mid H_{0}\right)\E\bigg(\E\left\{ R_{2}(0)\mid H_{1},R_{1}(0)=1\right\} \E\left[R_{3}(0)\cdots R_{s-1}(0)\left\{ 1-R_{s}(0)\right\} Y_{t}(0,s)\mid H_{1},R_{2}(0)=1\right]\\
 & \qquad\qquad\qquad\qquad\qquad\mid H_{0},R_{1}(0)=1\bigg)\Bigg\}\\
 & =\cdots\text{ (keep using the iterated expectation until the condition is \ensuremath{\left(H_{s-1},R_{s-1}(0)=1\right)})}\\
 & =\E\Bigg[\E\left(R_{1}(0)\mid H_{0}\right)\E\Bigg\{\cdots\E\left(\E\left(R_{s-1}(0)\mid H_{s-2},R_{s-2}(0)=1\right)\E\left[\left\{ 1-R_{s}(0)\right\} Y_{t}(0,s)\mid H_{s-1},R_{s-1}(0)=1\right]\right)\\
 & \qquad\qquad\qquad\qquad\qquad\qquad\mid H_{s-2},R_{s-2}(0)=1\mid\cdots\mid H_{0},R_{1}(0)=1\bigg)\Bigg\}\Bigg]\\
 & =\E\Bigg[\E\left(R_{1}(0)\mid H_{0}\right)\E\Bigg\{\E\left\{ R_{2}(0)\mid H_{1},R_{1}(0)=1\right\} \cdots\E\bigg(\E\left\{ R_{s-1}(0)\mid H_{s-2},R_{s-2}(0)=1\right\} \\
 & \qquad\qquad\qquad\E\Big[\E\left\{ 1-R_{s}(0)\mid H_{s-1},R_{s-1}(0)=1\right\} \E\left\{ Y_{t}(0,s)\mid H_{s-1},R_{s-1}(0)=1\right\} \mid H_{s-2},R_{s-1}(0)=1\Big]\\
 & \qquad\qquad\qquad\mid H_{s-3},R_{s-2}(0)=1\bigg)\cdots\mid H_{0},R_{1}(0)=1\Bigg\}\Bigg]\text{\ (By A5, \ensuremath{R_{s}(0)\indep Y_{t}(0,s)\mid\left(H_{s-1},R_{s-1}(0)=1\right)})}\\
 & =\E\Bigg\{\E\left(R_{1}\mid H_{0},A=0\right)\E\bigg(\cdots\E\Big[\E\left\{ R_{s-1}\mid H_{s-2},R_{s-2}=1,A=0\right\} \E\big\{\E\left(1-R_{s}\mid H_{s-1},R_{s-1}=1,A=0\right)\\
 & \qquad\E\left(Y_{t}\mid H_{s-1},R_{s-1}=1,A=0\right)\mid H_{s-2},R_{s-1}=1,A=0\big\}\mid H_{s-3},R_{s-2}=1,A=0\Big]\cdots\mid H_{0},R_{1}=1,A=0\bigg)\Bigg\}\text{ (By A6)}\\
 & =\E\Bigg[\pi_{1}(0,H_{0})\E\bigg\{\pi_{2}(0,H_{1})\cdots\E\bigg(\pi_{s-1}(0,H_{s-2})\E\left[\left\{ 1-\pi_{s}(0,H_{s-1})\right\} \mu_{t}^{0}(H_{s-1})\mid H_{s-2},R_{s-1}=1,A=0\right]\\
 & \qquad\qquad\qquad\qquad\qquad\qquad\mid H_{s-3},R_{s-2}=1,A=0\bigg)\cdots\mid H_{0},R_{1}=1,A=0\bigg)\Bigg]\\
 & =\E\left(\pi_{1}(0,H_{0})\E\left[\cdots\E\left\{ \pi_{s-1}(0,H_{s-2})g_{s}^{0}(H_{s-2})\mid H_{s-3},R_{s-2}=1,A=0\right\} \cdots\mid H_{0},R_{1}=1,A=0\right]\right)\\
 & =\E\left[\pi_{1}(0,H_{0})g_{s}^{0}(H_{0})\right]\text{ (By the definition of the pattern mean)}.
\end{align*}
\par When $s=1$, using the same calculation technique, we have $\E\left[\mathbb{I}\left\{ D(0)=s\right\} Y_{t}(0,s)\right]=\E\left[\left\{ 1-\pi_{1}(0,H_{0})\right\} \mu_{t}^{0}(H_{0})\right]$.
Note that under MAR, $\sum_{s=1}^{t}\pi_{1}(0,H_{0})g_{s+1}^{0}(H_{0})+\big\{1-\pi_{1}(0,H_{0})\big\}\mu_{t}^{0}(H_{0})=\mu_{t}^{0}(H_{0})$,
which completes the proof. \end{proof}

{Lemma \ref{lemma:iden_mar} validates the equivalence
in identifying the potential outcome mean $\E\left[Y_{t}^{\text{}}\{0,D(0)\}\right]$
between our proposed framework under J2R and the existing methods
under MAR, in the sense that one can use the sequential regression
model $\mu_{t}^{0}(H_{0})$ to estimate the control group mean. }

{}\begin{lemma}\label{lemma:delta_ratio}\par The propensity score
ratio $\delta(H_{s})$ have the following expression: 
\[
\delta(H_{s})=\frac{\bar{\pi}_{s}(1,H_{s-1})}{\bar{\pi}_{s}(0,H_{s-1})}\prod_{j=1}^{s}\frac{f(Y_{j}\mid H_{j-1},R_{j}=1,A=1)}{f(Y_{j}\mid H_{j-1},R_{j}=1,A=0)}.
\]
\par \end{lemma}

{}\begin{proof} For any $j\in\{1,\cdots,s\}$, we have 
\begin{align*}
\frac{f(Y_{j}\mid H_{j-1},R_{j}=1,A=1)}{f(Y_{j}\mid H_{j-1},R_{j}=1,A=0)} & =\frac{f(Y_{j},A=1\mid H_{j-1},R_{j}=1)/f(A=1\mid H_{j-1},R_{j}=1)}{f(Y_{j},A=0\mid H_{j-1},R_{j}=1)/f(A=0\mid H_{j-1},R_{j}=1)}\text{ (conditional probability)}\\
 & =\frac{f(A=1\mid H_{j},R_{j}=1)f(Y_{j}\mid H_{j-1},R_{j}=1)/f(A=1\mid H_{j-1},R_{j}=1)}{f(A=0\mid H_{j},R_{j}=1)f(Y_{j}\mid H_{j-1},R_{j}=1)/f(A=0\mid H_{j-1},R_{j}=1)}\\
 & =\frac{f(A=1\mid H_{j},R_{j}=1)/f(A=1\mid H_{j-1},R_{j}=1)}{f(A=0\mid H_{j},R_{j}=1)/f(A=0\mid H_{j-1},R_{j}=1)}\\
 & =\frac{f(A=1\mid H_{j},R_{j}=1)/f(A=1\mid H_{j-1},R_{j-1}=1)}{f(A=0\mid H_{j},R_{j}=1)/f(A=0\mid H_{j-1},R_{j-1}=1)}\frac{\pi_{j}(0,H_{j-1})}{\pi_{j}(1,H_{j-1})}.
\end{align*}
The last equality holds since 
\begin{align*}
\frac{f(A=0\mid H_{j-1},R_{j}=1)}{f(A=1\mid H_{j-1},R_{j}=1)} & =\frac{f(A=0,R_{j}=1\mid H_{j-1},R_{j-1}=1)/f(R_{j}=1\mid H_{j-1},R_{j-1}=1)}{f(A=1,R_{j}=1\mid H_{j-1},R_{j-1}=1)/f(R_{j}=1\mid H_{j-1},R_{j-1}=1)}\\
 & =\frac{f(A=0,R_{j}=1\mid H_{j-1},R_{j-1}=1)}{f(A=1,R_{j}=1\mid H_{j-1},R_{j-1}=1)}\\
 & =\frac{f(R_{j}=1\mid H_{j-1},R_{j-1}=1,A=0)f(A=0\mid H_{j-1},R_{j-1}=1)}{f(R_{j}=1\mid H_{j-1},R_{j-1}=1,A=1)f(A=1\mid H_{j-1},R_{j-1}=1)}\\
 & =\frac{\pi_{j}(0,H_{j-1})}{\pi_{j}(1,H_{j-1})}\frac{f(A=0\mid H_{j-1},R_{j-1}=1)}{f(A=1\mid H_{j-1},R_{j-1}=1)}.
\end{align*}
\par Taking the cumulative product for $j$ from $1$ to $s$, we
have 
\begin{align*}
\prod_{j=1}^{s}\frac{f(Y_{j}\mid H_{j-1},R_{j}=1,A=1)}{f(Y_{j}\mid H_{j-1},R_{j}=1,A=0)} & =\prod_{j=1}^{s}\frac{f(A=1\mid H_{j},R_{j}=1)/f(A=1\mid H_{j-1},R_{j-1}=1)}{f(A=0\mid H_{j},R_{j}=1)/f(A=0\mid H_{j-1},R_{j-1}=1)}\frac{\pi_{j}(0,H_{j-1})}{\pi_{j}(1,H_{j-1})}\\
 & =\frac{\bar{\pi}_{s}(0,H_{s-1})}{\bar{\pi}_{s}(1,H_{s-1})}\frac{f(A=1\mid H_{s},R_{s}=1)/f(A=1\mid H_{0})}{f(A=0\mid H_{s},R_{s}=1)/f(A=0\mid H_{0})}\\
 & =\frac{\bar{\pi}_{s}(0,H_{s-1})}{\bar{\pi}_{s}(1,H_{s-1})}\delta(H_{s}),
\end{align*}
which completes the proof. \end{proof}

{}We proceed to prove for the equivalence of the three identification
formulas. Note that $E_{1,t}=E_{2,t}$ holds since $\E\left[A\left\{ R_{t}Y_{t}+\sum_{s=1}^{t}R_{s-1}(1-R_{s})\mu_{t}^{0}(H_{s-1})\right\} \big/e(H_{0})\right]$
\begin{align*}
 & =\E\left[\frac{A}{e(H_{0})}\E\left\{ R_{t}Y_{t}+\sum_{s=1}^{t}R_{s-1}(1-R_{s})\mu_{t}^{0}(H_{s-1})\mid A=1\right\} \right]\\
 & =\E\left(\frac{A}{e(H_{0})}\left[\sum_{s=1}^{t}\pi_{1}(1,H_{0})g_{s+1}^{1}(H_{0})+\left\{ 1-\pi_{1}(1,H_{0})\right\} \mu_{t}^{0}(H_{0})\right]\right)\text{ (follow the proof in Lemma \ref{lemma:iden_mar})}\\
 & =\E\left[\sum_{s=1}^{t}\pi_{1}(1,H_{0})g_{s+1}^{1}(H_{0})+\left\{ 1-\pi_{1}(1,H_{0})\right\} \mu_{t}^{0}(H_{0})\right].
\end{align*}
Similarly, follow the proof in Lemma \ref{lemma:iden_mar}, 
\[
\E\left[\frac{1-A}{1-e(H_{0})}\big\{ R_{t}Y_{t}+\sum_{s=1}^{t}R_{s-1}(1-R_{s})\mu_{t}^{0}(H_{s-1})\big\}\right]=\E\left\{ \frac{1-A}{1-e(H_{0})}\mu_{t}^{0}(H_{0})\right\} =\E\left\{ \mu_{t}^{0}(H_{0})\right\} .
\]
Then, we have $E_{2,t}=\E\left[\pi_{1}(1,H_{0})\sum_{s=1}^{t}g_{s+1}^{1}(H_{0})+\left\{ 1-\pi_{1}(1,H_{0})\right\} \mu_{t}^{0}(H_{0})-\mu_{t}^{0}(H_{0})\right]=E_{1,t}$.

{}Also note that $E_{1,t}=E_{3,t}$ holds since for the first term
in $E_{3,t}$, $\E\left\{ AR_{t}Y_{t}/e(H_{0})\right\} =\E\left\{ \pi_{1}(1,H_{0})g_{t+1}^{1}(H_{0})\right\} $.
We focus on the second term and consider separate it into two components:
\begin{align*}
= & \E\left(\frac{1-A}{1-e(H_{0})}\left[\sum_{s=1}^{t}\bar{\pi}_{s-1}(0,H_{s-2})\{1-\pi_{s}(1,H_{s-1})\}\delta(H_{s-1})\right]\frac{R_{t}Y_{t}}{\bar{\pi}_{t}(0,H_{t-1})}\right)-\E\left\{ \frac{1-A}{1-e(H_{0})}\frac{R_{t}Y_{t}}{\bar{\pi}_{t}(0,H_{t-1})}\right\} .
\end{align*}

{}The second component can be easily obtained using the similar strategy
in Lemma \ref{lemma:iden_mar}, which results in $\E\big\{\mu_{t}^{0}(H_{0})\big\}$.
For the first components, apply Lemma \ref{lemma:delta_ratio}, for
$s\in\{2,\cdots,t\}$, we have 
\begin{align*}
 & \E\left[\frac{1-A}{1-e(H_{0})}\bar{\pi}_{s-1}(0,H_{s-2})\{1-\pi_{s}(1,H_{s-1})\}\delta(H_{s-1})\frac{R_{t}Y_{t}}{\bar{\pi}_{t}(0,H_{t-1})}\right]\\
= & \E\left[\frac{1-A}{1-e(H_{0})}\bar{\pi}_{s-1}(0,H_{s-2})\{1-\pi_{s}(1,H_{s-1})\}\delta(H_{s-1})\frac{R_{t}}{\bar{\pi}_{t}(0,H_{t-1})}\E\left(Y_{t}\mid H_{t-1},R_{t}=1,A=0\right)\right]\\
= & \cdots\text{ (keep using the iterated expectation, conditional on \ensuremath{H_{t-2},\cdots,H_{s-1}} in backward order)}\\
= & \E\left[\frac{1-A}{1-e(H_{0})}\bar{\pi}_{s-1}(0,H_{s-2})\{1-\pi_{s}(1,H_{s-1})\}\delta(H_{s-1})\frac{R_{s-1}}{\bar{\pi}_{s-1}(0,H_{s-2})}\mu_{t}^{0}(H_{s-1})\right]\\
= & \E\bigg(\frac{1-A}{1-e(H_{0})}\bar{\pi}_{s-1}(1,H_{s-2})\frac{R_{s-1}}{\bar{\pi}_{s-1}(0,H_{s-2})}\prod_{j=1}^{s-2}\frac{f(Y_{j}\mid H_{j-1},R_{j}=1,A=1)}{f(Y_{j}\mid H_{j-1},R_{j}=1,A=0)}\\
 & \quad\E\left[\{1-\pi_{s}(1,H_{s-1})\}\mu_{t}^{0}(H_{s-1})\frac{f(Y_{s-1}\mid H_{s-2},R_{s-1}=1,A=1)}{f(Y_{s-1}\mid H_{s-2},R_{s-1}=1,A=0)}\mid H_{s-2},R_{s-1}=1,A=0\right]\bigg)\text{ (Lemma \ref{lemma:delta_ratio})}\\
= & \E\left[\frac{1-A}{1-e(H_{0})}\bar{\pi}_{s-1}(1,H_{s-2})\frac{R_{s-1}}{\bar{\pi}_{s-1}(0,H_{s-2})}\prod_{j=1}^{s-2}\frac{f(Y_{j}\mid H_{j-1},R_{j}=1,A=1)}{f(Y_{j}\mid H_{j-1},R_{j}=1,A=0)}g_{s}^{1}(H_{s-2})\right]\\
= & \E\bigg[\frac{1-A}{1-e(H_{0})}\bar{\pi}_{s-2}(1,H_{s-3})\frac{R_{s-2}}{\bar{\pi}_{s-2}(0,H_{s-3})}\prod_{j=1}^{s-3}\frac{f(Y_{j}\mid H_{j-1},R_{j}=1,A=1)}{f(Y_{j}\mid H_{j-1},R_{j}=1,A=0)}\\
 & \quad\E\left\{ \pi_{s-1}(1,H_{s-2})g_{s}^{1}(H_{s-2})\frac{f(Y_{s-2}\mid H_{s-3},R_{s-2}=1,A=1)}{f(Y_{s-2}\mid H_{s-3},R_{s-2}=1,A=0)}\mid H_{s-3},R_{s-2}=1,A=0\right\} \bigg]\\
= & \E\left\{ \frac{1-A}{1-e(H_{0})}\bar{\pi}_{s-2}(1,H_{s-3})\frac{R_{s-2}}{\bar{\pi}_{s-2}(0,H_{s-3})}\prod_{j=1}^{s-3}\frac{f(Y_{j}\mid H_{j-1},R_{j}=1,A=1)}{f(Y_{j}\mid H_{j-1},R_{j}=1,A=0)}g_{s}^{1}(H_{s-3})\right\} \\
= & \cdots\text{ (keep using the iterated expectation, conditional on \ensuremath{H_{s-4},\cdots,H_{0}} in backward order)}\\
= & \E\left\{ \pi_{1}(1,H_{0})g_{s}^{1}(H_{0})\right\} .
\end{align*}
For $s=1$, use the same technique and can get $\E\left[\left\{ 1-\pi_{1}(1,H_{0})\right\} \mu_{t}^{0}(H_{0})\right]$.
Combine those components, we have 
\[
E_{3,t}=\E\left[\pi_{1}(1,H_{0})g_{t+1}^{1}(H_{0})+\sum_{s=1}^{t-1}\pi_{1}(1,H_{0})g_{s+1}^{1}(H_{0})+\left\{ 1-\pi_{1}(1,H_{0})\right\} \mu_{t}^{0}(H_{0})-\mu_{t}^{0}(H_{0})\right]=E_{1,t}.
\]

{}We proceed to prove the validity of the identification formula
(a) in Theorem 5. Denote $\tau_{1,t}=\E[Y_{t}\{1,D(1)\}${]} and $\tau_{0,t}=\E[Y_{t}\{0,D(0)\}]$.
For the first part of the identification formula (a) in Theorem 1,
\begin{align*}
\tau_{1,t} & =\E\left[\sum_{s=1}^{t+1}\mathbb{I}\left\{ D(1)=s\right\} Y_{t}(1,s)\right]\\
 & =\E\left[\sum_{s=1}^{t}\mathbb{I}\left\{ D(1)=s\right\} Y_{t}(1,s)+\mathbb{I}\left\{ D(1)=t+1\right\} Y_{t}(1,t+1)\right]\\
 & =\sum_{s=1}^{t}\E\left[R_{1}(1)\cdots R_{s-1}(1)\left\{ 1-R_{s}(1)\right\} Y_{t}(1,s)\right]+\E\left[R_{1}(1)\cdots R_{t}(1)Y_{t}(1,t+1)\right]\text{ (By the definition of \ensuremath{D})}.
\end{align*}

{}For $s\in\{2,\cdots,t\}$, $\E\left[R_{1}(1)\cdots R_{s-1}(1)\left\{ 1-R_{s}(1)\right\} Y_{t}(1,s)\right]$
can be computed by iterative expectations as 
\begin{align*}
= & \E\left(\E\left\{ R_{1}(1)\mid H_{0}\right\} \E\left[R_{2}(1)\cdots R_{s-1}(1)\left\{ 1-R_{s}(1)\right\} Y_{t}(1,s)\mid H_{0},R_{1}(1)=1\right]\right)\\
= & \E\Bigg\{\E\left\{ R_{1}(1)\mid H_{0}\right\} \E\bigg(\E\left\{ R_{2}(1)\mid H_{1},R_{1}(1)=1\right\} \\
 & \qquad\qquad\qquad\qquad\qquad\E\left[R_{3}(1)\cdots R_{s-1}(1)\left\{ 1-R_{s}(1)\right\} Y_{t}(1,s)\mid H_{1},R_{1}(1)=1\right]\mid H_{0},R_{1}(1)=1\bigg)\Bigg\}\\
= & \cdots\text{ (keep using the iterated expectation, use similar steps in the proof of Lemma \ref{lemma:iden_mar})}\\
= & \E\Bigg\{\E\left\{ R_{1}(1)\mid H_{0}\right\} \E\bigg(\E\left\{ R_{2}(1)\mid H_{1},R_{1}(1)=1\right\} \cdots\E\Big[\E\left\{ 1-R_{s}(1)\mid H_{s-1},R_{s-1}(1)=1\right\} \\
 & \qquad\qquad\qquad\qquad\E\left\{ Y_{t}(1,s)\mid H_{s-1},D(1)=s\right\} \mid H_{s-2},R_{s-1}(1)=1\Big]\cdots\mid H_{0},R_{1}(1)=1\bigg)\Bigg\}\\
= & \E\Bigg\{\E\left\{ R_{1}(1)\mid H_{0}\right\} \E\bigg(\E\left\{ R_{2}(1)\mid H_{1},R_{1}(1)=1\right\} \cdots\E\Big[\E\left\{ 1-R_{s}(1)\mid H_{s-1},R_{s-1}(1)=1\right\} \\
 & \qquad\qquad\qquad\qquad\qquad\mu_{t}^{0}(H_{s-1})\mid H_{s-2},R_{s-1}(1)=1\Big]\cdots\mid H_{0},R_{1}(1)=1\bigg)\Bigg\}\text{ (By A8)}\\
= & \E\Bigg\{\E\left\{ R_{1}(1)\mid H_{0},A=1\right\} \E\bigg(\E\left\{ R_{2}(1)\mid H_{1},R_{1}(1)=1,A=1\right\} \cdots\E\Big[\E\left\{ 1-R_{s}(1)\mid H_{s-1},R_{s-1}(1)=1,A=1\right\} \\
 & \qquad\qquad\qquad\qquad\qquad\mu_{t}^{0}(H_{s-1})\mid H_{s-2},R_{s-1}(1)=1,A=1\Big]\cdots\mid H_{0},R_{1}(1)=1,A=1\bigg)\Bigg\}\text{ (By A5)}\\
= & \E\left\{ \pi_{1}(1,H_{0})\E\left(\pi_{2}(1,H_{0})\cdots\E\left[\left\{ 1-\pi_{s}(1,H_{s-1})\right\} \mu_{t}^{0}(H_{s-1})\mid H_{s-2},R_{s-1}=1,A=1\right]\cdots\mid H_{0},R_{1}=1,A=1\right)\right\} \\
 & \text{ (By A7)}\\
= & \E\left\{ \pi_{1}(1,H_{0})g_{s}^{1}(H_{0})\right\} \text{ (By the definition of the pattern mean)}.
\end{align*}
Similarly, $\E\left\{ R_{1}(1)\cdots R_{t}(1)Y_{t}(1,t+1)\right\} $
\begin{align*}
= & \E\left[\E\left\{ R_{1}(1)\mid H_{0}\right\} \E\left\{ R_{2}(1)\cdots R_{t}(1)Y_{t}(1,t+1)\mid H_{0},R_{1}(1)=1\right\} \right]\\
= & \E\bigg(\E\left\{ R_{1}(1)\mid H_{0}\right\} \E\Big[\E\left\{ R_{2}(1)\mid H_{1},R_{1}(1)=1\right\} \\
 & \quad\E\left\{ R_{3}(1)\cdots R_{t}(1)Y_{t}(1,t+1)\mid H_{1},R_{2}(1)=1\right\} \mid H_{0},R_{1}(1)=1\Big]\bigg)\\
= & \cdots\text{ (keep using the iterated expectation, use similar steps in the proof of Lemma \ref{lemma:iden_mar})}\\
= & \E\Bigg\{\E\left\{ R_{1}(1)\mid H_{0}\right\} \E\bigg(\E\left\{ R_{2}(1)\mid H_{1},R_{1}(1)=1\right\} \cdots\E\Big[\E\left\{ R_{t}(1)\mid H_{t-1},R_{t-1}(1)=1\right\} \\
 & \qquad\qquad\qquad\qquad\E\left\{ Y_{t}(1,t+1)\mid H_{t-1},D(1)=t+1\right\} \mid H_{t-2},R_{t-1}(1)=1\Big]\cdots\mid H_{0},R_{1}(1)=1\bigg)\Bigg\}\\
= & \E\Bigg\{\E\left\{ R_{1}(1)\mid H_{0},A=1\right\} \E\bigg(\E\left\{ R_{2}(1)\mid H_{1},R_{1}(1)=1,A=1\right\} \cdots\E\Big[\E\left\{ R_{t}(1)\mid H_{t-1},R_{t-1}(1)=1,A=1\right\} \\
 & \quad\E\left\{ Y_{t}(1,t+1)\mid H_{t-1},D(1)=t+1,A=1\right\} \mid H_{t-2},R_{t-1}(1)=1,A=1\Big]\cdots\mid H_{0},R_{1}(1)=1,A=1\bigg)\Bigg\}\\
= & \E\bigg(\E\left(R_{1}\mid H_{0},A=1\right)\E\Big[\E\left(R_{2}\mid H_{1},R_{1}=1,A=1\right)\cdots\E\big\{\E\left(R_{t}\mid H_{t-1},R_{t-1}=1,A=1\right)\\
 & \qquad\qquad\qquad\E\left(Y_{t}\mid H_{t-1},R_{t}=1,A=1\right)\mid H_{t-2},R_{t-1}=1,A=1\big\}\cdots\mid H_{0},R_{1}(1)=1,A=1\Big]\bigg)\text{ (By A7)}\\
= & \E\left(\pi_{1}(1,H_{0})\E\left[\pi_{2}(1,H_{1})\cdots\E\left\{ \pi_{t}(1,H_{t-1})\mu_{t}^{1}(H_{t-1})\mid H_{t-2},R_{t-1}=1,A=1\right\} \cdots\mid H_{0},R_{1}(1)=1,A=1\right]\right)\\
= & \E\left\{ \pi_{1}(1,H_{0})g_{t+1}^{1}(H_{0})\right\} .
\end{align*}
When $s=1$, we have $\E\left[\left\{ 1-R_{s}(1)\right\} Y_{t}(1,1)\right]$
as 
\begin{align*}
\E\left[\E\left\{ 1-R_{1}(1)\mid H_{0}\right\} \E\left\{ Y_{t}(1,1)\mid H_{0},D(1)=0\right\} \right] & =\E\left\{ \E\left(1-R_{1}\mid H_{0},A=1\right)\mu_{t}^{0}(H_{0})\right\} \\
 & =\E\left[\left\{ 1-\pi_{1}(1,H_{0})\right\} \mu_{t}^{0}(H_{0})\right].
\end{align*}
Therefore, $\tau_{1,t}=\E\left[\pi_{1}(1,H_{0})\sum_{s=1}^{t}g_{s+1}^{1}(H_{0})+\left\{ 1-\pi_{1}(1,H_{0})\right\} \mu_{t}^{0}(H_{0})\right]$.
For the second part, by Lemma \ref{lemma:iden_mar} we know that $\tau_{0,t}=\E\left\{ \mu_{t}^{0}(H_{0})\right\} $.
Combine the two parts, we have 
\[
\tau_{t}^{\text{\jtr}}=\tau_{1,t}-\tau_{0,t}=\E\left[\pi_{1}(1,H_{0})\sum_{s=1}^{t}g_{s+1}^{1}(H_{0})+\left\{ 1-\pi_{1}(1,H_{0})\right\} \mu_{t}^{0}(H_{0})-\mu_{t}^{0}(H_{0})\right]=E_{1,t}.
\]

\subsection{Interpretations of Theorem 5\label{subsec:supp_idenint_longi} }

{}We give some intuition of the identification formulas in the longitudinal
setting. Theorem 5 (a) describes the treatment effect in terms of
the response probability and pattern mean. Under J2R, if the individual
in the treatment group is not fully observed, we would expect its
missing outcome will follow the same outcome model as the control
group with the same missing pattern given the observed data. The treatment
group mean is then expressed as the weighted sum over the missing
patterns as $\mathbb{E}\left[\pi_{1}(1,H_{0})\sum_{s=1}^{t}g_{s+1}^{1}(H_{0})+\left\{ 1-\pi_{1}(1,H_{0})\right\} \mu_{t}^{0}(H_{0})\right]$
under the PMM framework. For the control group, the group mean is
$\E\left\{ \mu_{t}^{0}(H_{0})\right\} $ under MAR.

{}Theorem 5 (b) describes the treatment effect as the difference
in means between the treatment and control groups over the missing
patterns, in terms of the propensity score and outcome mean. Similar
to the cross-sectional setting, after adjusting for the covariate
balance with the use of propensity score weights, the outcomes at
the last time point are combinations of the observed outcomes and
the conditional outcome means given the observed data, distinguished
by distinct dropout patterns.

{}Theorem 5 (c) describes the treatment effect over the missing patterns
in terms of the propensity score and response probability. The first
term $AR_{t}Y_{t}/e(H_{0})$ characterizes the participants who stay
in the assigned treatment throughout the entire study period identified
by $R_{t}$ after the adjustment for the group difference by $A/e(H_{0})$,
which is parallel to $\E\left\{ \pi_{1}(1,H_{0})g_{t+1}^{1}(H_{0})\right\} $
in Theorem 5 (a). The transformed outcome $(1-A)R_{t}Y_{t}\big/\left[\left\{ 1-e(H_{0})\right\} \bar{\pi}_{t}(0,H_{t-1})\right]$
measures the outcome mean $\mu_{t}^{0}(H_{0})$ given the baseline
covariates, for the participants who complete the trial in the control
group. Notice that 
\[
\delta(H_{s-1})=\frac{\bar{\pi}_{s-1}(1,H_{s-2})}{\bar{\pi}_{s-1}(0,H_{s-2})}\prod_{l=1}^{s-1}\frac{f(Y_{l}\mid H_{l-1},R_{l}=1,A=1)}{f(Y_{l}\mid H_{l-1},R_{l}=1,A=0)}
\]
is the cumulative product of the density ratios of the current outcome
given the observed historical information, multiplied by a ratio of
the cumulative response probability in the treatment and control group.
Therefore, with the transformed outcome involved, the term $\bar{\pi}_{s-1}(0,H_{s-2})\{1-\pi_{s}(1,H_{s-1})\}\delta(H_{s-1})$
implicitly shifts the participants with the same observed information,
who drop out at time $s$ in the treatment group, to the control group,
which matches $\E\left\{ \pi_{1}(1,H_{0})g_{s}^{1}(H_{0})\right\} $
when $s=2,\cdots,t$ and $\E\left[\left\{ 1-\pi_{1}(1,H_{0})\right\} \mu_{t}^{0}(H_{0})\right]$
when $s=1$ after marginalizing the history. Therefore, the second
term in the identification formula is equivalent to $\E\left[\pi_{1}(1,H_{0})\left\{ \sum_{s=1}^{t-1}g_{s+1}^{1}(H_{0})-\mu_{t}^{0}(H_{0})\right\} \right]$
in Theorem 5 (a).
\section{Proof of the EIFs \label{sec:supp_eif}}

{}Let $V=(X,A,R_{1}Y_{1},R_{1},\cdots,R_{t}Y_{t},R_{t})$ with $R_{0}=1$
be the vector of all observed variables with the likelihood factorized
as 
\begin{equation}
f(V)=f(X)f(A\mid X)\prod_{s=1}^{t}\left\{ f(Y_{s}\mid H_{s-1},R_{s}=1,A)f(R_{s}\mid H_{s-1},R_{s-1}=1,A)\right\} \label{eq:CR-lik-1}
\end{equation}
We will use the semiparametric theory in \citet{bickel1993efficient}
to derive the EIF of $\tau_{t}^{\text{\jtr}}$. To derive the EIFs,
we consider a one-dimensional parametric submodel, $f_{\theta}(V)$,
which contains the true model $f(V)$ at $\theta=0$, i.e., $f_{\theta}(V)\vert_{\theta=0}=f(V)$,
where $\theta$ consists of the nuisance model parameters. We use
$\theta$ in the subscript to denote the quantity evaluated with respect
to the submodel, e.g., $\mu_{t,\theta}^{a}$ is the value of $\mu_{t}^{a}$
with respect to the submodel. We use dot to denote the partial derivative
with respect to $\theta$, e.g., $\dot{\mu}_{t,\theta}^{a}=\partial\mu_{t}^{a}/\partial\theta$,
and use $s(\cdot)$ to denote the score function. From formula \eqref{eq:CR-lik-1},
the score function of the observed data can be decomposed as 
\[
s_{\theta}(V)=s_{\theta}(X)+s_{\theta}(A\mid X)+\sum_{s=1}^{t}\left\{ s_{\theta}(Y_{s}\mid H_{s-1},R_{s}=1,A)+s_{\theta}(R_{s}\mid H_{s-1},R_{s-1}=1,A)\right\} ,
\]
where $s_{\theta}(X)=\partial\log f_{\theta}(X)/\partial\theta$,
$s_{\theta}(A\mid X)=\partial\log\pr_{\theta}(A\mid X)/\partial\theta$,
$s_{\theta}(Y_{s}\mid H_{s-1},R_{s}=1,A)=\partial\log f_{\theta}(Y_{s}\mid H_{s-1},R_{s}=1,A)/\partial\theta$,
and $s_{\theta}(R_{s}\mid H_{s-1},R_{s-1}=1,A)=\partial\log\pr_{\theta}(R_{s}\mid H_{s-1},R_{s-1}=1,A)/\partial\theta$
are the score functions corresponding to the $(2t+2)$ components
of the likelihood. Because $f_{\theta}(V)\vert_{\theta=0}=f(V)$,
we can simplify $s_{\theta}(\cdot)\rvert_{\theta=0}$ as $s(\cdot)$.

{}From the semiparametric theory, the tangent space 
\[
\Lambda=B_{1}\oplus B_{2}\oplus B_{3,1}\oplus B_{4,1}\oplus\cdots\oplus B_{3,t}\oplus B_{4,t}
\]
is the direct sum of 
\begin{eqnarray*}
B_{1} & = & \{u(X):\E\{u(X)\}=0\},\\
B_{2} & = & \{u(A,X):\E\{u(A,X)\mid X\}=0\},\\
B_{3,s} & = & \{u(H_{s},A):\E\{u(H_{s},A)\mid A,H_{s-1}\}=0\},\\
B_{4,s} & = & \{u(R_{s},A,H_{s-1}):\E\{u(R_{s},A,H_{s-1})\mid A,H_{s-1}\}=0\},
\end{eqnarray*}
for $s=1,\cdots,t$, where $B_{1}$, $B_{2}$, $B_{3,s}$ and $B_{4,s}$
are orthogonal to each other, and $u(\cdot)$ is some functions. The
EIF for $\tau_{t}^{\jtr}$, denoted by $\varphi_{t}^{\text{\jtr}}(V;\mathbb{P})\in\Lambda$,
must satisfy 
\[
\left.\dot{\tau}_{t,\theta}^{\text{\jtr}}\right\rvert _{\theta=0}=\E\{\varphi_{t}^{\text{\jtr}}(V;\mathbb{P})s(V)\}.
\]

{}We will derive the EIFs in both cross-sectional and longitudinal
settings. To simplify the proof, we first provide some lemmas with
their proofs.

{}\begin{lemma}\label{lemma:basic_score}\par For any function
$u(V)$ that does not depend on $\theta$, $\partial\E_{\theta}\left\{ u(V)\right\} /\partial\theta\deriv=\E\left\{ u(V)s(V)\right\} $.\par \end{lemma}

{}\begin{proof} By the definition 
\begin{align*}
\frac{\partial\E_{\theta}\left\{ u(V)\right\} }{\partial\theta} & \deriv=\frac{\partial}{\partial\theta}\int u(V)f_{\theta}(V)d\nu(V)\deriv\\
 & =\int u(V)\frac{\partial}{\partial\theta}\log f_{\theta}(V)\deriv f(V)d\nu(V)\\
 & =\E\left\{ u(V)s(V)\right\} .
\end{align*}
\end{proof}

{}\begin{lemma}\label{lemma:pi_deriv}\par For $s=1,\cdots,t$,
we have 
\begin{align*}
\dot{\pi}_{s,\theta}(1,H_{s-1})\deriv & =\E\left[\frac{A}{e(X)}\frac{R_{s-1}}{\bar{\pi}_{s-1}(1,H_{s-2})}\left\{ R_{s}-\pi_{s}(1,H_{s-1})\right\} s(V)\mid H_{s-1}\right],\\
\dot{\pi}_{s,\theta}(0,H_{s-1})\deriv & =\E\left[\frac{1-A}{1-e(X)}\frac{R_{s-1}}{\bar{\pi}_{s-1}(0,H_{s-2})}\left\{ R_{s}-\pi_{s}(0,H_{s-1})\right\} s(V)\mid H_{s-1}\right].
\end{align*}
\par \end{lemma}

{}\begin{proof} Note that 
\begin{align*}
\dot{\pi}_{s,\theta}(1,H_{s-1})\deriv & =\frac{\partial}{\partial\theta}\E_{\theta}\left(R_{s}\mid H_{s-1},R_{s-1}=1,A=1\right)\deriv\\
 & =\E\left\{ R_{s}s(R_{s}\mid H_{s-1},R_{s-1}=1,A=1)\mid H_{s-1},R_{s-1}=1,A=1\right\} \text{ (by Lemma \ref{lemma:basic_score})}\\
 & =\E\left[\left\{ R_{s}-\pi_{s}(1,H_{s-1})\right\} s(R_{s}\mid H_{s-1},R_{s-1}=1,A=1)\mid H_{s-1},R_{s-1}=1,A=1\right]\\
 & =\E\left[\frac{A}{e(X)}\frac{R_{s-1}}{\bar{\pi}_{s-1}(1,H_{s-2})}\left\{ R_{s}-\pi_{s}(1,H_{s-1})\right\} s(R_{s}\mid H_{s-1},R_{s-1},A)\mid H_{s-1}\right]\text{ (by Bayes' rule)}\\
 & =\E\left[\frac{A}{e(X)}\frac{R_{s-1}}{\bar{\pi}_{s-1}(1,H_{s-2})}\left\{ R_{s}-\pi_{s}(1,H_{s-1})\right\} s(V)\mid H_{s-1}\right],
\end{align*}
where the last equality holds since $B_{3,s'}$, $B_{4,s'}$ for $s'>s$
are orthogonal to the spaces $B_{1},B_{2},$ $B_{3,1}$, $B_{4,1},\cdots,B_{3,s},B_{4,s}$.
Similarly, we can prove the result for $\dot{\pi}_{s,\theta}(0,H_{s-1})\deriv$.\end{proof}

{}\begin{lemma}\label{lemma:mu_deriv}\par For $s=1,\cdots,t$,
we have 
\begin{align}
\dot{\mu}_{t,\theta}^{1}(H_{t-1})\deriv & =\E\left[\frac{A}{e(X)}\frac{R_{t}}{\bar{\pi}_{t}(1,H_{t-1})}\left\{ Y_{t}-\mu_{t}^{1}(H_{t-1})\right\} s(V)\mid H_{t-1}\right],\nonumber \\
\dot{\mu}_{t,\theta}^{0}(H_{s-1})\deriv & =\E\left[\frac{1-A}{1-e(X)}\text{\ensuremath{\sum_{k=s}^{t}}}\frac{R_{k}}{\bar{\pi}_{k}(0,H_{k-1})}\left\{ \mu_{t}^{0}(H_{k})-\mu_{t}^{0}(H_{k-1})\right\} s(V)\mid H_{s-1}\right].\label{eq:mut0}
\end{align}
\par \end{lemma}

{}\begin{proof} Note that 
\begin{align*}
\dot{\mu}_{t,\theta}^{1}(H_{t-1})\deriv & =\frac{\partial}{\partial\theta}\E_{\theta}(Y_{t}\mid H_{t-1},R_{t}=1,A=1)\deriv\\
 & =\E\left\{ Y_{t}s(Y_{t}\mid H_{t-1},R_{t}=1,A=1)\mid H_{t-1},R_{t}=1,A=1\right\} \text{ (by Lemma \ref{lemma:basic_score})}\\
 & =\E\left\{ \frac{A}{e(X)}\frac{R_{t}}{\bar{\pi}_{t}(1,H_{t-1})}Y_{t}s(Y_{t}\mid H_{t-1},R_{t}=1,A=1)\mid H_{t-1}\right\} \text{ (by Bayes' rule)}\\
 & =\E\left[\frac{A}{e(X)}\frac{R_{t}}{\bar{\pi}_{t}(1,H_{t-1})}\left\{ Y_{t}-\mu_{t}^{1}(H_{t-1})\right\} s(Y_{t}\mid H_{t-1},R_{t},A)\mid H_{t-1}\right]\\
 & =\E\left[\frac{A}{e(X)}\frac{R_{t}}{\bar{\pi}_{t}(1,H_{t-1})}\left\{ Y_{t}-\mu_{t}^{1}(H_{t-1})\right\} s(V)\mid H_{t-1}\right].
\end{align*}
The last equality holds by the orthogonality of the spaces.\par For
the condition involves $A=0$, we prove it by induction in backward
order since it involves iteratively taking the derivative with respect
to $\theta$.\par For $s=t$, we can obtain $\dot{\mu}_{t,\theta}^{0}(H_{t-1})\deriv$
using the similar procedure as the one involves $A=1$, and get 
\[
\dot{\mu}_{t,\theta}^{0}(H_{t-1})\deriv=\E\left[\frac{1-A}{1-e(X)}\frac{R_{t}}{\bar{\pi}_{t}(0,H_{t-1})}\left\{ Y_{t}-\mu_{t}^{0}(H_{t-1})\right\} s(V)\mid H_{t-1}\right],
\]
which matches the right hand side of Equation \eqref{eq:mut0} when
$s=t$.\par Suppose Equation \eqref{eq:mut0} holds at time $(s+1)$
when $s<t$, i.e., 
\[
\dot{\mu}_{t,\theta}^{0}(H_{s})\deriv=\E\left[\frac{1-A}{1-e(X)}\text{\ensuremath{\sum_{k=s+1}^{t}}}\frac{R_{k}}{\bar{\pi}_{k}(0,H_{k-1})}\left\{ \mu_{t}^{0}(H_{k})-\mu_{t}^{0}(H_{k-1})\right\} s(V)\mid H_{s}\right].
\]
Then for the time point $s$, based on the sequential expression of
$\mu_{t}^{0}(H_{s-1})=\E\big\{\mu_{t}^{0}(H_{s})\mid H_{s-1},R_{s}=1,A=0\big\}$,
\begin{align*}
\dot{\mu}_{t,\theta}^{0}(H_{s})\deriv & =\frac{\partial}{\partial\theta}\E_{\theta}\left\{ \mu_{t}^{0}(H_{s})\mid H_{s-1},R_{s}=1,A=0\right\} \deriv\\
 & =\E\left\{ \dot{\mu}_{t,\theta}^{0}(H_{s})\deriv\mid H_{s-1},R_{s}=1,A=1\right\} \\
 & \;+\E\left\{ \mu_{t}^{0}(H_{s})s(Y_{t}\mid H_{s-1},R_{s}=1,A=1)\mid H_{s-1},R_{s}=1,A=1\right\} \text{ (by chain rule)}\\
 & =\E\left(\E\left[\frac{1-A}{1-e(X)}\text{\ensuremath{\sum_{k=s+1}^{t}}}\frac{R_{k}}{\bar{\pi}_{k}(0,H_{k-1})}\left\{ \mu_{t}^{0}(H_{k})-\mu_{t}^{0}(H_{k-1})\right\} s(V)\mid H_{s}\right]\mid H_{s-1}\right)\\
 & \;+\E\left(\frac{1-A}{1-e(X)}\frac{R_{s}}{\bar{\pi}_{s}(0,H_{s-1})}\left\{ \mu_{t}^{0}(H_{s})-\mu_{t}^{0}(H_{s-1})\right\} s(Y_{s}\mid H_{s-1},R_{s},A)\mid H_{s-1}\right)\text{ (by Bayes' rule)}\\
 & =\E\left[\frac{1-A}{1-e(X)}\text{\ensuremath{\sum_{k=s+1}^{t}}}\frac{R_{k}}{\bar{\pi}_{k}(0,H_{k-1})}\left\{ \mu_{t}^{0}(H_{k})-\mu_{t}^{0}(H_{k-1})\right\} s(V)\mid H_{s-1}\right]\text{ (by double expectation)}\\
 & \;+\E\left(\frac{1-A}{1-e(X)}\frac{R_{s}}{\bar{\pi}_{s}(0,H_{s-1})}\left\{ \mu_{t}^{0}(H_{s})-\mu_{t}^{0}(H_{s-1})\right\} s(V)\mid H_{s-1}\right)\text{\text{ (by orthogonality)}}\\
 & =\E\left[\frac{1-A}{1-e(X)}\text{\ensuremath{\sum_{k=s}^{t}}}\frac{R_{k}}{\bar{\pi}_{k}(0,H_{k-1})}\left\{ \mu_{t}^{0}(H_{k})-\mu_{t}^{0}(H_{k-1})\right\} s(V)\mid H_{s-1}\right],
\end{align*}
which completes the proof. \end{proof}

{}Denote the marginal mean for the longitudinal outcomes at the last
time point in the control group as $\tau_{0,t}$, i.e., $\tau_{0,t}=\E[Y_{t}^{\text{\jtr}}\{0,D(0)\}]$.
Under J2R, the missing values in the control group is MAR. The following
lemma provides the EIF for the control group mean $\tau_{0,t}$ under
MAR.

{}\begin{lemma}\label{lemma:eif_mar}\par Under MAR, the EIF for
$\tau_{0,t}$ is 
\[
\varphi_{0,t}(V;\mathbb{P})=\frac{1-A}{1-e(X)}\text{\ensuremath{\sum_{s=1}^{t}}}\frac{R_{s}}{\bar{\pi}_{s}(0,H_{s-1})}\left\{ \mu_{t}^{0}(H_{s})-\mu_{t}^{0}(H_{s-1})\right\} +\mu_{t}^{0}(H_{0})-\tau_{0,t}.
\]
\par \end{lemma}

{}\begin{proof} From the proof of Theorem 5, $\tau_{0,t}=\E\left\{ \mu_{t}^{0}(H_{0})\right\} $.
Then 
\begin{align*}
\dot{\tau}_{0,t,\theta}\deriv & =\frac{\partial}{\partial\theta}\E_{\theta}\left\{ \mu_{t}^{0}(H_{0})\right\} \deriv\\
 & =\E\left\{ \dot{\mu}_{t,\theta}^{0}(H_{0})\deriv\right\} +\E\left\{ \mu_{t}^{0}(H_{0})s(H_{0})\right\} \text{ (by chain rule)}\\
 & =\E\left(\E\left[\frac{1-A}{1-e(X)}\text{\ensuremath{\sum_{k=1}^{t}}}\frac{R_{k}}{\bar{\pi}_{k}(0,H_{k-1})}\left\{ \mu_{t}^{0}(H_{k})-\mu_{t}^{0}(H_{k-1})\right\} s(V)\mid H_{0}\right]\right)\\
 & \quad+\E\left[\left\{ \mu_{t}^{0}(H_{0})-\tau_{0,t}\right\} s(V)\right]\text{\text{ (by Lemma \ref{lemma:mu_deriv} and orthogonality)}}\\
 & =\E\left(\left[\frac{1-A}{1-e(X)}\text{\ensuremath{\sum_{k=1}^{t}}}\frac{R_{k}}{\bar{\pi}_{k}(0,H_{k-1})}\left\{ \mu_{t}^{0}(H_{k})-\mu_{t}^{0}(H_{k-1})\right\} +\mu_{t}^{0}(H_{0})-\tau_{0,t}\right]s(V)\right).
\end{align*}
Therefore, the proof is completed by the definition of the EIF as
$\left.\dot{\tau}_{0,t,\theta}\right\rvert _{\theta=0}=\E\{\varphi_{0,t}(V;\pr)s(V)\}.$\end{proof}

{}To proceed the proof in the longitudinal setting, we give the following
lemma for $\dot{g}_{s+1,\theta}^{1}(H_{l-1})\deriv$ when $l=1,\cdots,s-1$
and $s=1,\cdots,t$.

{}\begin{lemma}\label{lemma:g_seq}\par For any $s\in\{1,\cdots,t\}$,
when $l=1,\cdots,s-1$, we have 
\begin{align*}
\dot{g}_{s+1,\theta}^{1}(H_{l-1})\deriv & =\E\Bigg\{\bigg(\frac{A}{e(X)}\frac{R_{l}}{\bar{\pi}_{l}(1,H_{l-1})}\left[R_{s}\left\{ 1-R_{s+1}\right\} \mu_{t}^{0}(H_{s})-g_{s+1}^{1}(H_{l-1})\right]\\
 & \quad+\frac{1-A}{1-e(X)}\prod_{j=l+1}^{s}\pi_{j}(0,H_{j-1})\frac{f(Y_{j-1}\mid H_{j-2},R_{j-1}=1,A=1)}{f(Y_{j-1}\mid H_{j-2},R_{j-1}=1,A=0)}D_{s+1}\bigg)s(V)\mid H_{l-1}\Bigg\},\\
\dot{g}_{s+1,\theta}^{1}(H_{s-1})\deriv & =\E\left(\left[\frac{A}{e(X)}\frac{R_{s}}{\bar{\pi}_{s}(1,H_{s-1})}\left\{ \left(1-R_{s+1}\right)\mu_{t}^{0}(H_{s})-g_{s+1}^{1}(H_{s-1})\right\} +\frac{1-A}{1-e(X)}D_{s+1}\right]s(V)\mid H_{s-1}\right),
\end{align*}
where for the simplicity of notations, we denote 
\[
D_{s+1}:=\left\{ 1-\pi_{s+1}(0,H_{s})\right\} \frac{f(Y_{s}\mid H_{s-1},R_{s}=1,A=1)}{f(Y_{s}\mid H_{s-1},R_{s}=1,A=0)}\left[\text{\ensuremath{\sum_{k=s+1}^{t}}}\frac{R_{k}}{\bar{\pi}_{k}(0,H_{k-1})}\left\{ \mu_{t}^{0}(H_{k})-\mu_{t}^{0}(H_{k-1})\right\} \right],
\]
and let $D_{t+1}=0$.\par \end{lemma} 

{}\begin{proof} We first compute $\dot{g}_{s+1,\theta}^{1}(H_{s-1})\deriv$,
and use the iterated relationship $g_{s+1}^{1}(H_{l-1})=\E\big\{\pi_{l+1}(1,H_{l})g_{s+1}^{1}(H_{l})\mid H_{l-1},R_{l}=1,A=1\big\}$
for $l=1,\cdots,s-1$ and proceed by induction in backward order beginning
from $l=s-1$ to get $\dot{g}_{s+1,\theta}^{1}(H_{l-1})\deriv$. For
$\dot{g}_{s+1,\theta}^{1}(H_{s-1})\deriv$, 
\begin{align*}
\dot{g}_{s+1,\theta}^{1}(H_{s-1})\deriv & =\frac{\partial}{\partial\theta}\E_{\theta}\left[\left\{ 1-\pi_{s+1}(1,H_{s})\right\} \mu_{t}^{0}(H_{s})\mid H_{s-1},R_{s}=1,A=1\right]\deriv\\
 & =\E\left[-\dot{\pi}_{s+1,\theta}(1,H_{s})\deriv\mu_{t}^{0}(H_{s})\mid H_{s-1},R_{s}=1,A=1\right]\\
 & \;+\E\left[\left\{ 1-\pi_{s+1}(1,H_{s})\right\} \dot{\mu}_{t,\theta}^{0}(H_{s})\deriv\mid H_{s-1},R_{s}=1,A=1\right]\\
 & \;+\E\left[\left\{ 1-\pi_{s+1}(1,H_{s})\right\} \mu_{t}^{0}(H_{s})s(Y_{s}\mid H_{s-1},R_{s}=1,A=1)\mid H_{s-1},R_{s}=1,A=1\right]\\
 & =\E\left(-\E\left[\frac{A}{e(X)}\frac{R_{s}}{\bar{\pi}_{s}(1,H_{s-1})}\left\{ R_{s+1}-\pi_{s+1}(1,H_{s})\right\} \mu_{t}^{0}(H_{s})s(V)\mid H_{s}\right]\mid H_{s-1}\right)\text{\text{ (Lemma \ref{lemma:pi_deriv})}}\\
 & \;+\E\bigg(\E\Big[\frac{1-A}{1-e(X)}\left\{ 1-\pi_{s+1}(1,H_{s})\right\} \text{\ensuremath{\sum_{k=s+1}^{t}}}\frac{R_{k}}{\bar{\pi}_{k}(0,H_{k-1})}\left\{ \mu_{t}^{0}(H_{k})-\mu_{t}^{0}(H_{k-1})\right\} \\
 & \qquad\quad\quad s(V)\mid H_{s}\Big]\mid H_{s-1},R_{s}=1,A=1\bigg)\text{ (Lemma \ref{lemma:mu_deriv})}\\
 & \;+\E\left[\frac{A}{e(X)}\frac{R_{s}}{\bar{\pi}_{s}(1,H_{s-1})}\left\{ 1-\pi_{s+1}(1,H_{s})\right\} \mu_{t}^{0}(H_{s})s(Y_{s}\mid H_{s-1},R_{s},A)\mid H_{s-1}\right]\\
 & =\E\left[-\frac{A}{e(X)}\frac{R_{s}}{\bar{\pi}_{s}(1,H_{s-1})}\left\{ R_{s+1}-\pi_{s+1}(1,H_{s})\right\} \mu_{t}^{0}(H_{s})s(V)\mid H_{s-1}\right]\text{ (by double expectation)}\\
 & \;+\E\bigg(\E\Big[\frac{1-A}{1-e(X)}\left\{ 1-\pi_{s+1}(1,H_{s})\right\} \text{\ensuremath{\sum_{k=s+1}^{t}}}\frac{R_{k}}{\bar{\pi}_{k}(0,H_{k-1})}\left\{ \mu_{t}^{0}(H_{k})-\mu_{t}^{0}(H_{k-1})\right\} \\
 & \qquad\quad\quad s(V)\mid H_{s}\Big]\frac{f(Y_{s}\mid H_{s-1},R_{s}=1,A=1)}{f(Y_{s}\mid H_{s-1},R_{s}=1,A=0)}\mid H_{s-1}\bigg)\text{\text{ (by defintion of expectation)}}\\
 & \;+\E\left[\frac{A}{e(X)}\frac{R_{s}}{\bar{\pi}_{s}(1,H_{s-1})}\left\{ 1-\pi_{s+1}(1,H_{s})\right\} \left\{ \mu_{t}^{0}(H_{s})-g_{s+1}^{1}(H_{s-1})\right\} s(V)\mid H_{s-1}\right]\\
 & =\E\left[-\frac{A}{e(X)}\frac{R_{s}}{\bar{\pi}_{s}(1,H_{s-1})}\left\{ R_{s+1}-\pi_{s+1}(1,H_{s})\right\} \mu_{t}^{0}(H_{s})s(V)\mid H_{s-1}\right]\\
 & \;+\E\left\{ \frac{1-A}{1-e(X)}D_{s+1}s(V)\mid H_{s-1}\right\} \text{ (by double expectation)}\\
 & \;+\E\left[\frac{A}{e(X)}\frac{R_{s}}{\bar{\pi}_{s}(1,H_{s-1})}\left\{ 1-\pi_{s+1}(1,H_{s})\right\} \left\{ \mu_{t}^{0}(H_{s})-g_{s+1}^{1}(H_{s-1})\right\} s(V)\mid H_{s-1}\right],
\end{align*}
which completes the proof of the first part regarding $\dot{g}_{s+1,\theta}^{1}(H_{s-1})\deriv$.\par For
the second part of the proof, we derive it by induction backward starting
from $l=s-1$. For $l=s-1$, 
\begin{align*}
\dot{g}_{s+1,\theta}^{1}(H_{s-2})\deriv & =\frac{\partial}{\partial\theta}\E_{\theta}\left[\pi_{s}(1,H_{s-1})g_{s+1}^{1}(H_{s-1})\mid H_{s-2},R_{s-1}=1,A=1\right]\deriv\\
 & =\E\left[\dot{\pi}_{s,\theta}(1,H_{s-1})\deriv g_{s+1}^{1}(H_{s-1})\mid H_{s-1},R_{s}=1,A=1\right]\\
 & \;+\E\left[\pi_{s}(1,H_{s-1})\dot{g}_{s+1,\theta}^{1}(H_{s-1})\deriv\mid H_{s-2},R_{s-1}=1,A=1\right]\\
 & \;+\E\left[\pi_{s}(1,H_{s-1})g_{s+1}^{1}(H_{s-1})s(Y_{s-1}\mid H_{s-2},R_{s-1}=1,A=1)\mid H_{s-2},R_{s-1}=1,A=1\right]\\
 & =\E\left(\E\left[\frac{A}{e(X)}\frac{R_{s-1}}{\bar{\pi}_{s-1}(1,H_{s-2})}\left\{ R_{s}-\pi_{s}(1,H_{s-1})\right\} g_{s+1}^{1}(H_{s-1})s(V)\mid H_{s-1}\right]\mid H_{s-2}\right)\text{\text{ (Lemma \ref{lemma:pi_deriv})}}\\
 & \;+\E\left(\E\left[\frac{A}{e(X)}\frac{R_{s}}{\bar{\pi}_{s}(1,H_{s-1})}\pi_{s}(1,H_{s-1})\left\{ \left(1-R_{s+1}\right)\mu_{t}^{0}(H_{s})-g_{s+1}^{1}(H_{s-1})\right\} s(V)\mid H_{s-1}\right]\mid H_{s-2}\right)\\
 & \;+\E\left[\E\left\{ \frac{1-A}{1-e(X)}\pi_{s}(1,H_{s-1})D_{s+1}s(V)\mid H_{s-1}\right\} \mid H_{s-2},R_{s-1}=1,A=1\right]\\
 & \;+\E\left[\frac{A}{e(X)}\frac{R_{s-1}}{\bar{\pi}_{s-1}(1,H_{s-2})}\pi_{s}(1,H_{s-1})g_{s+1}^{1}(H_{s-1})s(Y_{s-1}\mid H_{s-2},R_{s-1},A)\mid H_{s-2}\right]\\
 & =\E\left[\frac{A}{e(X)}\frac{R_{s-1}}{\bar{\pi}_{s-1}(1,H_{s-2})}\left\{ R_{s}-\pi_{s}(1,H_{s-1})\right\} g_{s+1}^{1}(H_{s-1})s(V)\mid H_{s-2}\right]\\
 & \;+\E\left[\frac{A}{e(X)}\frac{R_{s}}{\bar{\pi}_{s}(1,H_{s-1})}\pi_{s}(1,H_{s-1})\left\{ \left(1-R_{s+1}\right)\mu_{t}^{0}(H_{s})-g_{s+1}^{1}(H_{s-1})\right\} s(V)\mid H_{s-2}\right]\\
 & \;+\E\bigg(\frac{1-A}{1-e(X)}\pi_{s}(1,H_{s-1})\frac{f(Y_{s-1}\mid H_{s-2},R_{s-1}=1,A=1)}{f(Y_{s-1}\mid H_{s-2},R_{s-1}=1,A=0)}D_{s+1}s(V)\mid H_{s-2}\bigg)\\
 & \;+\E\left[\frac{A}{e(X)}\frac{R_{s-1}}{\bar{\pi}_{s-1}(1,H_{s-2})}\left\{ \pi_{s}(1,H_{s-1})g_{s+1}^{1}(H_{s-1})-g_{s+1}^{1}(H_{s-2})\right\} s(V)\mid H_{s-2}\right]\\
 & =\E\bigg(\Big[\frac{A}{e(X)}\frac{R_{s-1}}{\bar{\pi}_{s-1}(1,H_{s-2})}\left\{ R_{s}\left(1-R_{s+1}\right)\mu_{t}^{0}(H_{s})-g_{s+1}^{1}(H_{s-2})\right\} \\
 & \qquad+\frac{1-A}{1-e(X)}\pi_{s}(1,H_{s-1})\frac{f(Y_{s-1}\mid H_{s-2},R_{s-1}=1,A=1)}{f(Y_{s-1}\mid H_{s-2},R_{s-1}=1,A=0)}D_{s+1}\Big]s(V)\mid H_{s-2}\bigg),
\end{align*}
matches the right hand side when $l=s-1$. Suppose the equality holds
for $(l+1)$ when $l<s-2$, i.e., 
\begin{align*}
\dot{g}_{s+1,\theta}^{1}(H_{l})\deriv & =\E\bigg(\Big[\frac{A}{e(X)}\frac{R_{l+1}}{\bar{\pi}_{l+1}(1,H_{l})}\left\{ R_{s}\left(1-R_{s+1}\right)\mu_{t}^{0}(H_{s})-g_{s+1}^{1}(H_{l})\right\} \\
 & \quad+\frac{1-A}{1-e(X)}\prod_{j=l+2}^{s}\pi_{j}(0,H_{j-1})\frac{f(Y_{j-1}\mid H_{j-2},R_{j-1}=1,A=1)}{f(Y_{j-1}\mid H_{j-2},R_{j-1}=1,A=0)}D_{s+1}s(V)\mid H_{l}\Big]\bigg).
\end{align*}
Then for $l$, we apply chain rule on the iterated formula: 
\begin{align*}
\dot{g}_{s+1,\theta}^{1}(H_{l-1})\deriv & =\frac{\partial}{\partial\theta}\E_{\theta}\big\{\pi_{l+1}(1,H_{l})g_{s+1}^{1}(H_{l})\mid H_{l-1},R_{l}=1,A=1\big\}\deriv\\
 & =\E\big\{\dot{\pi}_{l+1,\theta}(1,H_{l})\deriv g_{s+1}^{1}(H_{l})\mid H_{l-1},R_{l}=1,A=1\big\}\\
 & \;+\E\big\{\pi_{l+1}(1,H_{l})\dot{g}_{s+1,\theta}^{1}(H_{l})\deriv\mid H_{l-1},R_{l}=1,A=1\big\}\\
 & \;+\E\big\{\pi_{l+1}(1,H_{l})g_{s+1}^{1}(H_{l})s(Y_{l}\mid H_{l-1},R_{l}=1,A=1)\mid H_{l-1},R_{l}=1,A=1\big\}\\
 & =\E\left[\frac{A}{e(X)}\frac{R_{l}}{\bar{\pi}_{l}(1,H_{l-1})}\left\{ R_{l+1}-\pi_{l+1}(1,H_{l})\right\} g_{s+1}^{1}(H_{l})s(V)\mid H_{l-1}\right]\\
 & \;+\E\left[\frac{A}{e(X)}\frac{R_{l}}{\bar{\pi}_{l}(1,H_{l-1})}\left\{ R_{s}\left(1-R_{s+1}\right)\mu_{t}^{0}(H_{s})-g_{s+1}^{1}(H_{l})\right\} s(V)\mid H_{l-1}\right]\\
 & \;+\E\Big[\E\big\{\frac{1-A}{1-e(X)}\pi_{l+1}(1,H_{l})\prod_{j=l+2}^{s}\pi_{j}(0,H_{j-1})\frac{f(Y_{j-1}\mid H_{j-2},R_{j-1}=1,A=1)}{f(Y_{j-1}\mid H_{j-2},R_{j-1}=1,A=0)}\\
 & \qquad\qquad D_{s+1}s(V)\mid H_{l}\big\}\mid H_{l-1},R_{l}=1,A=1\Big]\\
 & \;+\E\left[\frac{A}{e(X)}\frac{R_{l}}{\bar{\pi}_{l}(1,H_{s-1})}\left\{ \pi_{l+1}(1,H_{l})g_{s+1}^{1}(H_{l})-g_{s+1}^{1}(H_{l-1})\right\} s(V)\mid H_{l-1}\right]\\
 & =\E\left[\frac{A}{e(X)}\frac{R_{l}}{\bar{\pi}_{l}(1,H_{l-1})}\left\{ R_{s}\left(1-R_{s+1}\right)\mu_{t}^{0}(H_{s})-g_{s+1}^{1}(H_{l-1})\right\} s(V)\mid H_{l-1}\right]\\
 & \;+\E\Big[\E\big\{\frac{1-A}{1-e(X)}\pi_{l+1}(1,H_{l})\prod_{j=l+2}^{s}\pi_{j}(0,H_{j-1})\frac{f(Y_{j-1}\mid H_{j-2},R_{j-1}=1,A=1)}{f(Y_{j-1}\mid H_{j-2},R_{j-1}=1,A=0)}\\
 & \qquad\qquad D_{s+1}s(V)\mid H_{l}\big\}\frac{f(Y_{l-1}\mid H_{l-2},R_{l-1}=1,A=1)}{f(Y_{l-1}\mid H_{l-2},R_{l-1}=1,A=0)}\mid H_{l-1}\Big]\\
 & =\E\left[\frac{A}{e(X)}\frac{R_{l}}{\bar{\pi}_{l}(1,H_{l-1})}\left\{ R_{s}\left(1-R_{s+1}\right)\mu_{t}^{0}(H_{s})-g_{s+1}^{1}(H_{l-1})\right\} s(V)\mid H_{l-1}\right]\\
 & \;+\E\left\{ \frac{1-A}{1-e(X)}\prod_{j=l+1}^{s}\pi_{j}(0,H_{j-1})\frac{f(Y_{j-1}\mid H_{j-2},R_{j-1}=1,A=1)}{f(Y_{j-1}\mid H_{j-2},R_{j-1}=1,A=0)}D_{s+1}s(V)\mid H_{l-1}\right\} ,
\end{align*}
completes the proof.\end{proof}

{}From Lemma \ref{lemma:g_seq}, we proceed to obtain $\partial\E_{\theta}\left\{ \pi_{1}(1,H_{0})g_{s+1}^{1}(H_{0})\right\} /\partial\theta\deriv$
in the following lemma.

{}\begin{lemma}\label{lemma:pi*g}\par For any $s\in\{1,\cdots,t\}$,
we have 
\begin{align*}
 & \frac{\partial\E_{\theta}\left\{ \pi_{1}(1,H_{0})g_{s+1}^{1}(H_{0})\right\} }{\partial\theta}\Bigg|_{\theta=0}=\E\bigg(\Big[\frac{A}{e(X)}\left\{ R_{s}\left(1-R_{s+1}\right)\mu_{t}^{0}(H_{s})-\pi_{1}(1,H_{0})g_{s+1}^{1}(H_{0})\right\} \\
 & +\frac{1-A}{1-e(X)}\bar{\pi}_{s}(1,H_{s-1})\left\{ 1-\pi_{s+1}(0,H_{s})\right\} \delta(H_{s})W_{s+1}+\pi_{1}(1,H_{0})g_{s+1}^{1}(H_{0})\Big]s(V)\bigg).
\end{align*}
where $W_{s+1}=\text{\ensuremath{\sum_{k=s+1}^{t}}}R_{k}\left\{ \mu_{t}^{0}(H_{k})-\mu_{t}^{0}(H_{k-1})\right\} /\bar{\pi}_{k}(0,H_{k-1})$
and $W_{t+1}=0$.\par \end{lemma}

{}\begin{proof} Lemma \ref{lemma:g_seq} implies that when $l=1$,
\begin{align*}
\dot{g}_{s+1,\theta}^{1}(H_{0})\deriv & =\E\bigg(\Big[\frac{A}{e(X)}\frac{R_{1}}{\pi_{1}(1,H_{0})}\left\{ R_{s}\left(1-R_{s+1}\right)\mu_{t}^{0}(H_{s})-g_{s+1}^{1}(H_{0})\right\} \\
 & \quad+\frac{1-A}{1-e(X)}\prod_{j=2}^{s}\pi_{j}(0,H_{j-1})\frac{f(Y_{j-1}\mid H_{j-2},R_{j-1}=1,A=1)}{f(Y_{j-1}\mid H_{j-2},R_{j-1}=1,A=0)}D_{s+1}\Big]s(V)\mid H_{0}\bigg).
\end{align*}
Then we have $\partial\E_{\theta}\left\{ \pi_{1}(1,H_{0})g_{s+1}^{1}(H_{0})\right\} \text{\ensuremath{\big/}}\partial\theta\deriv$
\begin{align*}
= & \E\left\{ \dot{\pi}_{1,\theta}(1,H_{0})\deriv g_{s+1}^{1}(H_{0})\right\} +\E\left\{ \pi_{1}(1,H_{0})\dot{g}_{s+1,\theta}^{1}(H_{0})\deriv\right\} \\
 & +\E\left\{ \pi_{1}(1,H_{0})g_{s+1}^{1}(H_{0})s(H_{0})\right\} \\
= & \E\bigg(\Big[\frac{A}{e(X)}\left\{ R_{s}\left(1-R_{s+1}\right)\mu_{t}^{0}(H_{s})-\pi_{1}(1,H_{0})g_{s+1}^{1}(H_{0})\right\} \\
 & +\frac{1-A}{1-e(X)}\bar{\pi}_{s}(0,H_{s-1})\left\{ 1-\pi_{s+1}(0,H_{s})\right\} W_{s+1}\prod_{j=1}^{s}\frac{f(Y_{j}\mid H_{j-1},R_{j}=1,A=1)}{f(Y_{j}\mid H_{j-1},R_{j}=1,A=0)}\Big]s(V)\bigg)\\
 & +\E\left\{ \pi_{1}(1,H_{0})g_{s+1}^{1}(H_{0})s(V)\right\} .
\end{align*}
Note that $\delta(H_{s})=\prod_{j=1}^{s}\big\{ f(Y_{j}\mid H_{j-1},R_{j}=1,A=1)/f(Y_{j}\mid H_{j-1},R_{j}=1,A=0)\big\}\bar{\pi}_{s}(1,H_{s-1})/\bar{\pi}_{s}(0,H_{s-1})$
by Lemma \ref{lemma:delta_ratio}, which completes the proof. \end{proof}

\subsection{Proof of Theorem 2}

{}We compute the EIF by rewriting the identification formula in Theorem
1 (a) as $\tau_{1}^{\text{\jtr}}=\tau_{1,1}-\tau_{0,1}$, where $\tau_{1,1}=\E\left[\pi_{1}(1,X)\mu_{1}^{1}(X)+\{1-\pi_{1}(1,X)\}\mu_{1}^{0}(X)\right]$
and $\tau_{0,1}=\E\left\{ \mu_{1}^{0}(X)\right\} $ based on the proof
in \ref{subsec:supp_iden_1time}. By Lemma \ref{lemma:eif_mar}, 
\[
\varphi_{0,1}(V;\mathbb{P})=\frac{1-A}{1-e(X)}\frac{R_{1}}{\pi_{1}(0,X)}\left\{ Y_{1}-\mu_{1}^{0}(X)\right\} +\mu_{1}^{0}(X)-\tau_{0,1}.
\]
We proceed to compute $\dot{\tau}_{1,1,\theta}\deriv$. Note that,
\begin{align*}
\dot{\tau}_{1,1,\theta}\deriv & =\frac{\partial}{\partial\theta}\E_{\theta}\left[\pi_{1}(1,X)\mu_{1}^{1}(X)+\{1-\pi_{1}(1,X)\}\mu_{1}^{0}(X)\right]\deriv\\
 & =\E\left\{ \dot{\pi}_{1,\theta}(1,X)\deriv\mu_{1}^{1}(X)+\pi_{1}(1,X)\dot{\mu}_{1,\theta}^{1}(X)\deriv+\pi_{1}(1,X)\mu_{1}^{1}(X)s(X)\right\} \\
 & \quad+\E\left[-\dot{\pi}_{1,\theta}(1,X)\deriv\mu_{1}^{0}(X)+\{1-\pi_{1}(1,X)\}\dot{\mu}_{1,\theta}^{0}(X)\deriv+\{1-\pi_{1}(1,X)\}\mu_{1}^{0}(X)s(X)\right]\\
 & =\E\left(\E\left[\frac{A}{e(X)}\left\{ R_{1}-\pi_{1}(1,H_{0})\right\} \left\{ \mu_{1}^{1}(X)-\mu_{1}^{0}(X)\right\} s(V)\mid X\right]\right)\text{\text{ (by Lemma \ref{lemma:pi_deriv})}}\\
 & \quad+\E\left(\E\left[\frac{A}{e(X)}\frac{R_{1}}{\pi_{1}(1,X)}\pi_{1}(1,X)\left\{ Y_{1}-\mu_{1}^{1}(X)\right\} s(V)\mid X\right]\right)\text{\text{ (by Lemma \ref{lemma:mu_deriv})}}\\
 & \quad+\E\left(\E\left[\frac{1-A}{1-e(X)}\frac{R_{1}}{\pi_{1}(0,X)}\{1-\pi_{1}(1,X)\}\left\{ Y_{1}-\mu_{1}^{0}(X)\right\} s(V)\mid X\right]\right)\text{\text{ (by Lemma \ref{lemma:mu_deriv})}}\\
 & \quad+\E\left(\left[\pi_{1}(1,X)\mu_{1}^{1}(X)+\left\{ 1-\pi_{1}(1,X)\right\} \mu_{1}^{0}(X)-\tau_{1,1}\right]s(V)\right)\text{\text{ (by orthogonality)}}\\
 & =\E\left(\left[\frac{A}{e(X)}R_{1}+\frac{1-A}{1-e(X)}\frac{R_{1}}{\pi_{1}(0,X)}\left\{ 1-\pi_{1}(1,X)\right\} \right]\left\{ Y_{1}-\mu_{1}^{0}(X)\right\} s(V)\right)\\
 & \quad+\E\left(\left[\left\{ \pi_{1}(1,X)\mu_{1}^{1}(X)+\{1-\pi_{1}(1,X)\}\mu_{1}^{0}(X)\right\} -\frac{A}{e(X)}\pi_{1}(1,H_{0})\left\{ \mu_{1}^{1}(X)-\mu_{1}^{0}(X)\right\} -\tau_{1,1}\right]s(V)\right).
\end{align*}
Then we can get $\varphi_{1,1}(V;\pr)$ based on the definition of
the EIF as $\left.\dot{\tau}_{1,1,\theta}\right\rvert _{\theta=0}=\E\{\varphi_{1,1}(V;\mathbb{P})s(V)\}.$

{}The EIF of $\tau_{1}^{\text{\jtr}}$ can then be obtained: 
\begin{align*}
\varphi_{1}^{\text{\jtr}}(V;\mathbb{P}) & =\varphi_{1,1}(V;\mathbb{P})-\varphi_{0,1}(V;\mathbb{P})\\
 & =\left[\frac{A}{e(X)}R_{1}+\frac{1-A}{1-e(X)}\frac{R_{1}}{\pi_{1}(0,X)}\left\{ 1-\pi_{1}(1,X)\right\} \right]\left\{ Y_{1}-\mu_{1}^{0}(X)\right\} \\
 & \quad+\left\{ \pi_{1}(1,X)\mu_{1}^{1}(X)+\{1-\pi_{1}(1,X)\}\mu_{1}^{0}(X)\right\} -\frac{A}{e(X)}\pi_{1}(1,H_{0})\left\{ \mu_{1}^{1}(X)-\mu_{1}^{0}(X)\right\} -\tau_{1,1}\\
 & \quad-\frac{1-A}{1-e(X)}\frac{R_{1}}{\pi_{1}(0,X)}\left\{ Y_{1}-\mu_{1}^{0}(X)\right\} -\mu_{1}^{0}(X)+\tau_{0,1}\\
 & =\left\{ \frac{A}{e(X)}-\frac{1-A}{1-e(X)}\frac{\pi_{1}(1,X)}{\pi_{1}(0,X)}\right\} R_{1}\left\{ Y_{1}-\mu_{1}^{0}(X)\right\} +\left\{ 1-\frac{A}{e(X)}\right\} \pi_{1}(1,X)\left\{ \mu_{1}^{1}(X)-\mu_{1}^{0}(X)\right\} -\tau_{1}^{\text{\jtr}}.
\end{align*}

\subsection{Proof of Theorem 6}

{}We compute the EIF based on the identification formula in Theorem
5 (a) as $\tau_{t}^{\jtr}=\tau_{1,t}-\tau_{0,t}$, where $\tau_{1,t}=\E\left[\pi_{1}(1,H_{0})\sum_{s=1}^{t}g_{s+1}^{1}(H_{0})+\left\{ 1-\pi_{1}(1,H_{0})\right\} \mu_{t}^{0}(H_{0})\right]$
and $\tau_{0,t}=\E\left\{ \mu_{t}^{0}(H_{0})\right\} $ based on the
proof in \ref{subsec:supp_iden_longi}. By Lemma \ref{lemma:eif_mar},
we can obtain $\varphi_{0,t}(V;\pr)$. We only need to calculate the
EIF for $\tau_{1,t}$. Note that for any $s\in\{1,\cdots,t\},$ $\partial\E_{\theta}\left\{ \pi_{1}(1,H_{0})g_{s+1}^{1}(H_{0})\right\} /\partial\theta\deriv$
is obtained by Lemma \ref{lemma:pi*g}. The part $\partial\E_{\theta}\left[\left\{ 1-\pi_{1}(1,H_{0})\right\} \mu_{t}^{0}(H_{0})\right]/\partial\theta\deriv$
can be derived using chain rule and Lemmas \ref{lemma:pi_deriv} and
\ref{lemma:mu_deriv} as 
\begin{align*}
\frac{\E_{\theta}\left[\left\{ 1-\pi_{1}(1,H_{0})\right\} \mu_{t}^{0}(H_{0})\right]}{\partial\theta}\Bigg|_{\theta=0} & =\E\left\{ -\dot{\pi}_{1,\theta}(1,H_{0})\deriv\mu_{t}^{0}(H_{0})\right\} +\E\left[\left\{ 1-\pi_{1}(1,H_{0})\right\} \dot{\mu}_{t,\theta}^{0}(H_{0})\deriv\right]\\
 & \quad+\E\left[\left\{ 1-\pi_{1}(1,H_{0})\right\} \mu_{t}^{0}(H_{0})s(H_{0})\right]\\
 & =\E\left(\left[-\frac{A}{e(X)}\left\{ R_{1}-\pi_{1}(1,H_{0})\right\} \mu_{t}^{0}(H_{0})+\left\{ 1-\pi_{1}(1,H_{0})\right\} \mu_{t}^{0}(H_{0})\right]s(V)\right)\\
 & \quad+\E\left[\left\{ 1-\pi_{1}(1,H_{0})\right\} \varphi_{0,t}(V;\mathbb{P})s(V)\right].
\end{align*}
Combine all terms together and by the definition of the EIF, we have
\begin{align*}
\varphi_{1,t}(V;\mathbb{P}) & =\frac{A}{e(X)}\sum_{s=0}^{t}R_{s}\left(1-R_{s+1}\right)\mu_{t}^{0}(H_{s})-\frac{A}{e(X)}\sum_{s=1}^{t}\pi_{1}(1,H_{0})g_{s+1}^{1}(H_{0})\\
 & \;+\frac{1-A}{1-e(X)}\sum_{s=1}^{t-1}\bar{\pi}_{s}(1,H_{s-1})\left\{ 1-\pi_{s+1}(0,H_{s})\right\} \delta(H_{s})W_{s+1}\text{ (since \ensuremath{D_{t+1}=0})}\\
 & \;+\pi_{1}(1,H_{0})\sum_{s=1}^{t}g_{s+1}^{1}(H_{0})+\left\{ 1-\pi_{1}(1,H_{0})\right\} \varphi_{0,t}(V;\mathbb{P})+\mu_{t}^{0}(H_{0})-\tau_{1,t}.
\end{align*}

{}Apply Lemma \ref{lemma:eif_mar}, the EIF $\varphi_{t}^{\text{\jtr}}(V;\mathbb{P})$
of $\tau_{t}^{\text{\jtr}}$ is 
\begin{align*}
\varphi_{t}^{\text{\jtr}}(V;\mathbb{P}) & =\varphi_{1,t}(V;\pr)-\varphi_{0,t}(V;\mathbb{P})\\
 & =\frac{A}{e(X)}\sum_{s=0}^{t}R_{s}\left(1-R_{s+1}\right)\mu_{t}^{0}(H_{s})-\frac{A}{e(X)}\sum_{s=1}^{t}\pi_{1}(1,H_{0})g_{s+1}^{1}(H_{0})\\
 & \;+\frac{1-A}{1-e(X)}\sum_{s=1}^{t-1}\bar{\pi}_{s}(1,H_{s-1})\left\{ 1-\pi_{s+1}(0,H_{s})\right\} \delta(H_{s})W_{s+1}\\
 & \;+\pi_{1}(1,H_{0})\sum_{s=1}^{t}g_{s+1}^{1}(H_{0})-\pi_{1}(1,H_{0})\varphi_{0,t}(V;P)+\mu_{t}^{0}(H_{0})-\tau_{t}^{\text{\jtr}}\\
 & =\frac{A}{e(X)}\left\{ R_{t}Y_{t}+\sum_{s=1}^{t}R_{s-1}\left(1-R_{s}\right)\mu_{t}^{0}(H_{s-1})\right\} -\tau_{t}^{\text{\jtr}}\\
 & +\left\{ 1-\frac{A}{e(X)}\right\} \left[\pi_{1}(1,H_{0})\sum_{s=1}^{t}g_{s+1}^{1}(H_{0})+\left\{ 1-\pi_{1}(1,H_{0})\right\} \mu_{t}^{0}(H_{0})\right]-\mu_{t}^{0}(H_{0})\\
 & +\frac{1-A}{1-e(X)}\left(\sum_{s=1}^{t-1}\bar{\pi}_{s}(1,H_{s-1})\left\{ 1-\pi_{s+1}(0,H_{s})\right\} \delta(H_{s})W_{s+1}+\left\{ 1-\pi_{1}(1,H_{0})\right\} W_{1}-W_{1}\right).
\end{align*}
Simplify the last term, note that $\sum_{s=1}^{t}\bar{\pi}_{s}(1,H_{s-1})\left\{ 1-\pi_{s+1}(0,H_{s})\right\} \delta(H_{s})W_{s+1}+\left\{ 1-\pi_{1}(1,H_{0})\right\} W_{1}$
\begin{align*}
 & =\sum_{s=0}^{t-1}\bar{\pi}_{s}(1,H_{s-1})\left\{ 1-\pi_{s+1}(0,H_{s})\right\} \delta(H_{s})\text{\ensuremath{\sum_{k=s+1}^{t}}}\frac{R_{k}}{\bar{\pi}_{k}(0,H_{k-1})}\left\{ \mu_{t}^{0}(H_{k})-\mu_{t}^{0}(H_{k-1})\right\} \\
 & =\sum_{k=1}^{t}\sum_{s=0}^{k-1}\bar{\pi}_{s}(1,H_{s-1})\left\{ 1-\pi_{s+1}(0,H_{s})\right\} \delta(H_{s})\frac{R_{k}}{\bar{\pi}_{k}(0,H_{k-1})}\left\{ \mu_{t}^{0}(H_{k})-\mu_{t}^{0}(H_{k-1})\right\} \text{ (change the order of \ensuremath{k} and \ensuremath{s})}\\
 & =\sum_{k=1}^{t}\sum_{s=1}^{k}\bar{\pi}_{s-1}(1,H_{s-2})\left\{ 1-\pi_{s}(0,H_{s-1})\right\} \delta(H_{s-1})\frac{R_{k}}{\bar{\pi}_{k}(0,H_{k-1})}\left\{ \mu_{t}^{0}(H_{k})-\mu_{t}^{0}(H_{k-1})\right\} \text{ (change \ensuremath{s} to \ensuremath{s+1})}\\
 & =\sum_{s=1}^{t}\left[\sum_{k=1}^{s}\bar{\pi}_{k-1}(1,H_{k-2})\left\{ 1-\pi_{k}(0,H_{k-1})\right\} \delta(H_{k-1})\right]\frac{R_{s}}{\bar{\pi}_{s}(0,H_{s-1})}\left\{ \mu_{t}^{0}(H_{s})-\mu_{t}^{0}(H_{s-1})\right\} \text{ (interchange \ensuremath{k} and \ensuremath{s})}.
\end{align*}
Then the last term becomes 
\begin{align*}
 & \frac{1-A}{1-e(X)}\left(\sum_{s=1}^{t}\left[\sum_{k=1}^{s}\bar{\pi}_{k-1}(1,H_{k-2})\left\{ 1-\pi_{k}(0,H_{k-1})\right\} \delta(H_{k-1})\right]\frac{R_{s}}{\bar{\pi}_{s}(0,H_{s-1})}\left\{ \mu_{t}^{0}(H_{s})-\mu_{t}^{0}(H_{s-1})\right\} -W_{1}\right)\\
= & \frac{1-A}{1-e(X)}\left(\sum_{s=1}^{t}\left[\sum_{k=1}^{s}\bar{\pi}_{k-1}(1,H_{k-2})\left\{ 1-\pi_{k}(0,H_{k-1})\right\} \delta(H_{k-1})-1\right]\frac{R_{s}}{\bar{\pi}_{s}(0,H_{s-1})}\left\{ \mu_{t}^{0}(H_{s})-\mu_{t}^{0}(H_{s-1})\right\} \right).
\end{align*}

{}Therefore, the EIF $\varphi_{t}^{\text{\jtr}}(V;\mathbb{P})$ of
$\tau_{t}^{\text{\jtr}}$ 
\begin{align*}
 & =\frac{A}{e(X)}\left\{ R_{t}Y_{t}+\sum_{s=1}^{t}R_{s-1}\left(1-R_{s}\right)\mu_{t}^{0}(H_{s-1})\right\} -\tau_{t}^{\text{\jtr}}\\
 & \quad+\left\{ 1-\frac{A}{e(X)}\right\} \left[\pi_{1}(1,H_{0})\sum_{s=1}^{t}g_{s+1}^{1}(H_{0})+\left\{ 1-\pi_{1}(1,H_{0})\right\} \mu_{t}^{0}(H_{0})\right]-\mu_{t}^{0}(H_{0})\\
 & \quad+\frac{1-A}{1-e(X)}\left(\sum_{s=1}^{t}\left[\sum_{k=1}^{s}\bar{\pi}_{k-1}(1,H_{k-2})\left\{ 1-\pi_{k}(0,H_{k-1})\right\} \delta(H_{k-1})-1\right]\frac{R_{s}}{\bar{\pi}_{s}(0,H_{s-1})}\left\{ \mu_{t}^{0}(H_{s})-\mu_{t}^{0}(H_{s-1})\right\} \right),
\end{align*}
which matches the expression given in Theorem 6.

\section{Estimation\label{sec:supp_est} }

\subsection{Normalized estimators motivated from Theorem 1\label{subsec:supp_norm_1time}}

{}We give the normalized version of the ps-om and ps-rp estimators
in cross-sectional studies below.

\begin{example}\label{exmp: estn_1time}The normalized version of
the ps-om and ps-rp estimators are as follows: 
\begin{enumerate}
\item The normalized ps-om estimator: 
\begin{align*}
\hat{\tau}_{\text{ps-om-N}} & =\pr_{n}\left[\frac{A}{e(X;\hat{\alpha})}\left\{ R_{1}Y_{1}+(1-R_{1})\mu_{1}^{0}(X;\hat{\beta})\right\} \right]\big/\pr_{n}\left\{ \frac{A}{e(X;\hat{\alpha})}\right\} \\
 & -\pr_{n}\left[\frac{1-A}{1-e(X;\hat{\alpha})}\left\{ R_{1}Y_{1}+(1-R_{1})\mu_{1}^{0}(X;\hat{\beta})\right\} \right]\big/\pr_{n}\left\{ \frac{1-A}{1-e(X;\hat{\alpha})}\right\} .
\end{align*}
The normalized estimator is consistent under $\mathcal{M}_{\text{ps+om}}$. 
\item The normalized ps-rp estimator: 
\begin{align*}
\hat{\tau}_{\text{ps-rp-N}} & =\pr_{n}\left\{ \frac{A}{e(X;\hat{\alpha})}R_{1}Y_{1}\right\} \Big/\pr_{n}\left\{ \frac{A}{e(X;\hat{\alpha})}\right\} -\pr_{n}\left\{ \frac{1-A}{1-e(X;\hat{\alpha})}\frac{\pi_{1}(1,X;\hat{\gamma})}{\pi_{1}(0,X;\hat{\gamma})}R_{1}Y_{1}\right\} \Big/\pr_{n}\left\{ \frac{1-A}{1-e(X;\hat{\alpha})}\frac{R_{1}}{\pi_{1}(0,X;\hat{\gamma})}\right\} .
\end{align*}
The normalized estimator is consistent under $\mathcal{M}_{\text{ps+rp}}$. 
\end{enumerate}
\end{example}

\subsection{EIF-based estimators motivated from Theorem 2 }

{}We provide the EIF-based estimator $\hat{\tau}_{\text{tr}}$ and
its normalized estimator $\hat{\tau}_{\text{tr-N}}$ in the cross-sectional
studies as follows. 
\begin{eqnarray*}
{\widehat{\tau}_{\text{tr}}} & = & \mathbb{P}_{n}\Bigg[\left\{ \frac{A}{e(X;\hat{\alpha})}-\frac{1-A}{1-e(X;\hat{\alpha})}\frac{\pi_{1}(1,X;\hat{\gamma})}{\pi_{1}(0,X;\hat{\gamma})}\right\} R_{1}\left\{ Y_{1}-\mu_{1}^{0}(X;\hat{\beta})\right\} \\
 &  & -\frac{A-e(X;\hat{\alpha})}{e(X;\hat{\alpha})}\pi_{1}(1,X;\hat{\gamma})\left\{ \mu_{1}^{1}(X;\hat{\beta})-\mu_{1}^{0}(X;\hat{\beta})\right\} \Bigg].
\end{eqnarray*}
\begin{eqnarray*}
\hat{\tau}_{\text{tr-N}} & = & \mathbb{P}_{n}\left(\frac{A}{e(X;\hat{\alpha})}\left[R_{1}\big\{ Y_{1}-\mu_{1}^{0}(X;\hat{\beta})\big\}-\pi_{1}(1,X;\hat{\gamma})\big\{\mu_{1}^{1}(X;\hat{\beta})-\mu_{1}^{0}(X;\hat{\beta})\big\}\right]\right)/\pr_{n}\{\frac{A}{e(X;\hat{\alpha})}\}\\
 &  & -\pr_{n}\left[\frac{1-A}{1-e(X;\hat{\alpha})}\frac{\pi_{1}(1,X;\hat{\gamma})}{\pi_{1}(0,X;\hat{\gamma})}R_{1}\big\{ Y_{1}-\mu_{1}^{0}(X;\hat{\beta})\big\}\right]/\pr_{n}\{\frac{1-A}{1-e(X;\hat{\alpha})}\frac{R_{1}}{\pi_{1}(0,X;\hat{\gamma})}\}\\
 &  & +\pr_{n}\left[\pi_{1}(1,X;\hat{\gamma})\big\{\mu_{1}^{1}(X;\hat{\beta})-\mu_{1}^{0}(X;\hat{\beta})\big\}\right].
\end{eqnarray*}

{}One can conduct calibration to further reduce the impact of the
outliers as introduced in the main text. The calibration-based estimator
$\hat{\tau}_{\text{tr-C}}$ is as follows. 
\begin{align*}
{\color{black}\hat{\tau}_{\text{tr-C}}}{\color{black}{\color{black}{\color{black}{\color{black}=}}}} & \pr_{n}\left(Aw_{a_{1}}\left[R_{1}\big\{ Y_{1}-\mu_{1}^{0}(X;\hat{\beta})\big\}-\pi_{1}(1,X;\hat{\gamma})\big\{\mu_{1}^{1}(X;\hat{\beta})-\mu_{1}^{0}(X;\hat{\beta})\big\}\right]\right)\Big/\pr_{n}\left(Aw_{a_{1}}\right)\\
 & -\pr_{n}\left[(1-A)R_{1}w_{a_{0}}w_{r_{1}}\pi_{1}(1,X;\hat{\gamma})\big\{ Y_{1}-\mu_{1}^{0}(X;\hat{\beta})\big\}\right]/\pr_{n}\left\{ (1-A)R_{1}w_{a_{0}}w_{r_{1}}\right\} \\
 & -\pr_{n}\left[\pi_{1}(1,X_{i};\hat{\gamma})\big\{\mu_{1}^{1}(X;\hat{\beta})-\mu_{1}^{0}(X;\hat{\beta})\big\}\right].
\end{align*}

\subsection{Normalized estimators motivated from Theorem 5\label{subsec:supp_norm_longi}}

{}We give the normalized version of the ps-om and ps-rp estimators
in the longitudinal setting below.

\begin{example}\label{exmp: estn_longi}The normalized version of
the ps-om and ps-rp estimators are as follows: 
\begin{enumerate}
\item The normalized ps-om estimator: 
\begin{align*}
\hat{\tau}_{\text{ps-om-N}} & =\pr_{n}\left[\frac{A}{\hat{e}(H_{0})}\left\{ R_{t}Y_{t}+\sum_{s=1}^{t}R_{s-1}(1-R_{s})\hat{\mu}_{t}^{0}(H_{s-1})\right\} \right]\big/\pr_{n}\left\{ \frac{A}{\hat{e}(H_{0})}\right\} \\
 & -\pr_{n}\left[\frac{1-A}{1-\hat{e}(H_{0})}\left\{ R_{t}Y_{t}+\sum_{s=1}^{t}R_{s-1}(1-R_{s})\hat{\mu}_{t}^{0}(H_{s-1})\right\} \right]\big/\pr_{n}\left\{ \frac{1-A}{1-\hat{e}(H_{0})}\right\} .
\end{align*}
\item The normalized ps-rp estimator: 
\begin{align*}
\hat{\tau}_{\text{ps-rp-N}} & =\pr_{n}\left\{ \frac{A}{\hat{e}(H_{0})}R_{t}Y_{t}\right\} \Big/\pr_{n}\left\{ \frac{A}{\hat{e}(H_{0})}\right\} \\
 & \quad-\pr_{n}\bigg(\frac{1-A}{1-\hat{e}(H_{0})}\Big[\sum_{s=1}^{t}\hat{\bar{\pi}}_{s-1}(0,H_{s-2})\left\{ 1-\hat{\pi}_{s}(1,H_{s-1})\right\} \hat{\delta}(H_{s-1})-1\Big]\frac{R_{t}Y_{t}}{\hat{\bar{\pi}}_{t}(0,H_{t-1})}\bigg)\\
 & \qquad\Big/\pr_{n}\left\{ \frac{1-A}{1-\hat{e}(H_{0})}\frac{R_{t}}{\hat{\bar{\pi}}_{t}(0,H_{t-1})}\right\} .
\end{align*}
\end{enumerate}
\end{example}

\subsection{Estimation procedure in the longitudinal setting\label{subsec:eststeps}}

{}We consider the case when $t=2$, and give detailed steps to estimate
$\hat{\tau}_{\text{rp-pm}}$, $\hat{\tau}_{\text{ps-om}}$ and $\hat{\tau}_{\text{ps-rp}}$
as an example for a straightforward illustration. Extend the estimation
procedure to the setting when $t>2$ is straightforward. Based on
Example 2 (a) in the main text, 
\begin{align*}
\hat{\tau}_{\text{rp-pm}} & =\mathbb{\pr}_{n}\bigg\{\hat{\pi}_{1}(1,H_{0})\Big(\hat{\E}\left\{ \hat{\pi}_{2}(1,H_{1})\hat{\mu}_{2}^{1}(H_{1})\mid H_{0},R_{1}=1,A=1\right\} \\
 & +\hat{\E}\left[\left\{ 1-\hat{\pi}_{2}(1,H_{1})\right\} \hat{\mu}_{2}^{0}(H_{1})\mid H_{0},R_{1}=1,A=1\right]-\hat{\mu}_{2}^{0}(H_{0})\Big)\bigg\}.
\end{align*}
The steps of estimating the rp-pm estimator when $t=2$ are summarized
as follows:

{}\setlength{\itemindent}{1.5em}

\begin{enumerate}
\item[\textbf{\large{}{}Step 1}{\large{}{}.}] {}For subjects with $R_{2}=1$, obtain the fitted outcome mean $\hat{\mu}_{2}^{a}(H_{1})$
for $a=0,1$.

\item[\textbf{\large{}{}Step 2}{\large{}{}.}] {}For subjects with $R_{1}=1$, obtain the following estimated nuisance
functions:

\begin{enumerate}
\item {}The estimated pattern mean $\hat{g}_{2}^{1}(H_{0}),\hat{g}_{3}^{1}(H_{0})$:
Fit $g_{2}^{1}(H_{0})=\E\Big[\left\{ 1-\pi_{2}(1,H_{1})\right\} \mu_{2}^{0}(H_{1})\mid H_{0},R_{1}=1,A=1\Big]$
and $g_{3}^{1}(H_{0})=\E\left\{ \pi_{2}(1,H_{1})\mu_{2}^{1}(H_{1})\mid H_{0},R_{1}=1,A=1\right\} $
using the predicted values $\left\{ 1-\hat{\pi}_{2}(1,H_{1})\right\} \hat{\mu}_{2}^{0}(H_{1})$
and $\hat{\pi}_{2}(1,H_{1})\hat{\mu}_{2}^{1}(H_{1})$ against $H_{0}$
in the group with $R_{1}=1$ and $A=1$, respectively.

\item {}The estimated response probability $\hat{\pi}_{2}(a,H_{1})$.

\item {}The estimated outcome mean $\hat{\mu}_{2}^{0}(H_{0})$: Fit $\mu_{2}^{0}(H_{0})=\E\left\{ \mu_{2}^{0}(H_{1})\mid H_{0},R_{1}=1,A=0\right\} $
using the predicted values $\hat{\mu}_{2}^{0}(H_{1})$ against $H_{0}$
in the group with $R_{1}=1$ and $A=0$.

\end{enumerate}
\item[\textbf{\large{}{}Step 3}{\large{}{}.}] {}For all the subjects, obtain the estimated response probability
$\hat{\pi}_{1}(a,H_{1})$.

\item[\textbf{\large{}{}Step 4}{\large{}{}.}] {}Get $\hat{\tau}_{\text{rp-pm}}$ by the empirical average.

\end{enumerate}
{}Based on Example 2 (b) in the main text, 
\begin{align*}
\hat{\tau}_{\text{ps-om}} & =\pr_{n}\left[\frac{2A-1}{\hat{e}(X)^{A}\left\{ 1-\hat{e}(X)\right\} ^{1-A}}\left\{ R_{2}Y_{2}+R_{1}(1-R_{2})\hat{\mu}_{2}^{0}(H_{1})+(1-R_{1})\hat{\mu}_{2}^{0}(H_{0})\right\} \right].
\end{align*}
The steps of estimating the ps-om estimator are as follows:

{}\setlength{\itemindent}{1.5em}

\begin{enumerate}
\item[\textbf{\large{}{}Step 1}{\large{}{}.}] {}For subjects with $R_{2}=1$, obtain the fitted outcome mean model
$\hat{\mu}_{2}^{0}(H_{1})$ .

\item[\textbf{\large{}{}Step 2}{\large{}{}.}] {}For subjects with $R_{1}=1$, obtain the fitted outcome mean model
$\hat{\mu}_{2}^{0}(H_{0})$, by fitting $\mu_{2}^{0}(H_{0})=\E\left\{ \mu_{2}^{0}(H_{1})\mid H_{0},R_{1}=1,A=0\right\} $
using the predicted values $\hat{\mu}_{2}^{0}(H_{1})$ against $H_{0}$
in the group with $R_{1}=1$ and $A=0$.

\item[\textbf{\large{}{}Step 3}{\large{}{}.}] {}For all the subjects, obtain the fitted propensity score model
$\hat{e}(X)$.

\item[\textbf{\large{}{}Step 4}{\large{}{}.}] {}Get $\hat{\tau}_{\text{ps-om}}$ by the empirical average.

\end{enumerate}
{}Based on Example 2 (c) in the main text, 
\begin{align*}
\hat{\tau}_{\text{ps-rp}} & =\pr_{n}\bigg(\frac{A}{\hat{e}(X)}R_{2}Y_{2}+\frac{1-A}{1-\hat{e}(X)}\left[\hat{\pi}_{1}(0,H_{0})\left\{ 1-\hat{\pi}_{2}(1,H_{1})\right\} \hat{\delta}(H_{1})-\hat{\pi}_{1}(1,H_{0})\right]\frac{R_{2}Y_{2}}{\hat{\bar{\pi}}_{2}(0,H_{1})}\bigg).
\end{align*}
The steps of estimating the ps-rp estimator are as follows:

{}\setlength{\itemindent}{1.5em}

\begin{enumerate}
\item[\textbf{\large{}{}Step 1}{\large{}{}.}] {}For subjects with $R_{1}=1$, obtain the following models:

\begin{enumerate}
\item {}The fitted propensity score model $\hat{e}(H_{1})$.

\item {}The fitted response probability model $\hat{\pi}_{2}(a,H_{1})$.

\end{enumerate}
\item[\textbf{\large{}{}Step 2}{\large{}{}.}] {}For all the subjects, obtain the following models:

\begin{enumerate}
\item {}The fitted propensity score model $\hat{e}(X)$.

\item {}The fitted response probability model $\hat{\pi}_{1}(a,H_{0})$.

\end{enumerate}
\item[\textbf{\large{}{}Step 4}{\large{}{}.}] {}Obtain $\hat{\delta}(H_{1})=\left\{ \hat{e}(H_{1})/\hat{e}(H_{0})\right\} \Big/\left[\left\{ 1-\hat{e}(H_{1})\right\} /\left\{ 1-\hat{e}(H_{0})\right\} \right]$
for the subjects with $R_{1}=1$, and get $\hat{\tau}_{\text{ps-rp}}$
by the empirical average.

\end{enumerate}

\subsection{Multiply robust estimators motivated from Theorem 6\label{subsec:supp_trest_longi}}

{}From the EIF, one can motivated new estimators of $\tau_{t}^{\text{\jtr}}$.
We present the expression of $\hat{\tau}_{\text{mr}}$ below. 
\begin{align*}
\hat{\tau}_{\text{mr}} & =\mathbb{\pr}_{n}\bigg(\frac{A}{\hat{e}(H_{0})}\big\{ R_{t}Y_{t}+\sum_{s=1}^{t}R_{s-1}(1-R_{s})\hat{\mu}_{t}^{0}(H_{s-1})\big\}\\
 & +\left\{ 1-\frac{A}{\hat{e}(H_{0})}\right\} \left[\hat{\pi}_{1}(1,H_{0})\sum_{s=1}^{t}\hat{g}_{s+1}^{1}(H_{0})+\big\{1-\hat{\pi}_{1}(1,H_{0})\big\}\hat{\mu}_{t}^{0}(H_{0})\right]-\hat{\mu}_{t}^{0}(H_{0})\\
 & +\frac{1-A}{1-\hat{e}(H_{0})}\sum_{s=1}^{t}\left[\sum_{k=1}^{s}\hat{\bar{\pi}}_{k-1}(0,H_{k-2})\{1-\hat{\pi}_{k}(1,H_{k-1})\}\hat{\delta}(H_{k-1})-1\right]\frac{R_{s}}{\hat{\bar{\pi}}_{s}(0,H_{s-1})}\big\{\hat{\mu}_{t}^{0}(H_{s})-\hat{\mu}_{t}^{0}(H_{s-1})\big\}\bigg).
\end{align*}
Now, we provide the normalized version of $\hat{\tau}_{\text{mr}}$
as follows. The normalized estimator is less influenced by the extreme
weights compared to $\hat{\tau}_{\text{mr}}$. 
\begin{align*}
\hat{\tau}_{\text{mr-N}} & =\mathbb{\pr}_{n}\bigg(\frac{A}{\hat{e}(H_{0})}\Big[R_{t}Y_{t}+\sum_{s=1}^{t}R_{s-1}(1-R_{s})\hat{\mu}_{t}^{0}(H_{s-1})-\hat{\pi}_{1}(1,H_{0})\sum_{s=1}^{t}\hat{g}_{s+1}^{1}(H_{0})\\
 & -\big\{1-\hat{\pi}_{1}(1,H_{0})\big\}\hat{\mu}_{t}^{0}(H_{0})\Big]\bigg)\Big/\pr_{n}\left\{ \frac{A}{\hat{e}(H_{0})}\right\} \\
 & +\pr_{n}\left\{ \hat{\pi}_{1}(1,H_{0})\sum_{s=1}^{t}\hat{g}_{s+1}^{1}(H_{0})-\hat{\pi}_{1}(1,H_{0})\hat{\mu}_{t}^{0}(H_{0})\right\} \\
 & +\sum_{s=1}^{t}\pr_{n}\Bigg\{\frac{1-A}{1-\hat{e}(H_{0})}\bigg(\left[\sum_{k=1}^{s}\hat{\bar{\pi}}_{k-1}(0,H_{k-2})\left\{ 1-\hat{\pi}_{k}(1,H_{k-1})\right\} \hat{\delta}(H_{k-1})-1\right]\\
 & \quad\quad\quad\quad\frac{R_{s}}{\hat{\bar{\pi}}_{s}(0,H_{s-1})}\big\{\hat{\mu}_{t}^{0}(H_{s})-\hat{\mu}_{t}^{0}(H_{s-1})\big\}\bigg)\Bigg\}\Big/\pr_{n}\left\{ \frac{1-A}{1-\hat{e}(H_{0})}\frac{R_{s}}{\hat{\bar{\pi}}_{s}(0,H_{s-1})}\right\} .
\end{align*}

{}In addition, one can conduct calibration to further reduce the
impact of the outliers. The calibration-based estimator expresses
as follows. 
\begin{align*}
\hat{\tau}_{\text{mr-C}} & =\pr_{n}\bigg(Aw_{a_{1}}\Big[R_{t}Y_{t}+\sum_{s=1}^{t}R_{s-1}(1-R_{s})\hat{\mu}_{t}^{0}(H_{s-1})-\hat{\pi}_{1}(1,H_{0})\sum_{s=1}^{t}\hat{g}_{s+1}^{1}(H_{0})\\
 & \qquad\quad-\big\{1-\hat{\pi}_{1}(1,H_{0})\big\}\hat{\mu}_{t}^{0}(H_{0})\Big]\bigg)\bigg/\pr_{n}\left(Aw_{a_{1}}\right)\\
 & +\pr_{n}\left\{ \pi_{1}(1,H_{0};\hat{\gamma})\sum_{s=1}^{t}\hat{g}_{s+1}^{1}(H_{0})-\pi_{1}(1,H_{0};\hat{\gamma})\hat{\mu}_{t}^{0}(H_{0})\right\} \\
 & +\pr_{n}\Bigg\{(1-A)R_{s}w_{a_{0}}w_{r_{1}}\cdots w_{r_{s}}\bigg(\left[\sum_{k=1}^{s}\hat{\bar{\pi}}_{k-1}(0,H_{k-2})\left\{ 1-\hat{\pi}_{k}(1,H_{k-1})\right\} \hat{\delta}(H_{k-1})-1\right]\\
 & \quad\quad\quad\quad\big\{\hat{\mu}_{t}^{0}(H_{s})-\hat{\mu}_{t}^{0}(H_{s-1})\big\}\bigg)\Bigg\}\Big/\pr_{n}\left\{ (1-A)R_{s}w_{a_{0}}w_{r_{1}}\cdots w_{r_{s}}\right\} .
\end{align*}

{}We present the detailed estimation steps for the calibration-based
estimator $\hat{\tau}_{\text{mr-C}}$ when $t=2$ below for illustration.

{}\setlength{\itemindent}{1.5em}

\begin{enumerate}
\item[\textbf{\large{}{}Step 1}{\large{}{}.}] {}For subjects with $R_{2}=1$, obtain the fitted outcome mean models
$\hat{\mu}_{2}^{a}(H_{1})$ for $a=0,1$.

\item[\textbf{\large{}{}Step 2}{\large{}{}.}] {}For subjects with $R_{1}=1$, obtain the following quantities:

\begin{enumerate}
\item {}The fitted propensity score model $\hat{e}(H_{1})$.

\item {}The fitted response probability model $\hat{\pi}_{2}(a,H_{1})$.

\item {}The fitted outcome mean model $\hat{\mu}_{2}^{0}(H_{0})$, by fitting
$\mu_{2}^{0}(H_{0})=\E\left\{ \mu_{2}^{0}(H_{1})\mid H_{0},R_{1}=1,A=0\right\} $
using the predicted values $\hat{\mu}_{2}^{0}(H_{0})$ against $H_{0}$
in the group with $R_{1}=1$ and $A=0$.

\item {}The fitted models $\hat{g}_{2}^{1}(H_{0}),\hat{g}_{3}^{1}(H_{0})$:
Fit $g_{2}^{1}(H_{0})=\E\left\{ \pi_{2}(1,H_{1})\mu_{2}^{1}(H_{1})\mid H_{0},R_{1}=1,A=1\right\} $
and $g_{3}^{1}(H_{0})=\E\left[\left\{ 1-\pi_{2}(1,H_{1})\right\} \mu_{2}^{0}(H_{1})\mid H_{0},R_{1}=1,A=1\right]$
using the predicted values $\hat{\pi}_{2}(1,H_{1})\hat{\mu}_{2}^{1}(H_{1})$
and $\left\{ 1-\hat{\pi}_{2}(1,H_{1})\right\} \hat{\mu}_{2}^{0}(H_{1})$
against $H_{0}$ in the group with $R_{1}=1$ and $A=1$, respectively.

\item {}The calibration weights $w_{r_{2}}$ associated with the response
indicator $R_{2}$: Solve the optimization problem (1) subject to
$\sum_{i:R_{2,i}=1}w_{r_{2},i}h(X_{i})=\sum_{i:R_{1,i}=1}h(X_{i})/\left(\sum_{i=1}^{n}R_{1,i}\right)$.

\end{enumerate}
\item[\textbf{\large{}{}Step 3}{\large{}{}.}] {}For all the subjects, obtain the following models:

\begin{enumerate}
\item {}The fitted propensity score model $\hat{e}(H_{0})$ and the ratio
$\hat{\delta}(H_{1})$ for the subjects with $R_{1}=1$.

\item {}The fitted response probability model $\hat{\pi}_{1}(a,H_{0})$.

\item {}The calibration weights $w_{r_{1}}$ associated with the response
indicator $R_{1}$: Solve the optimization problem (1) subject to
$\sum_{i:R_{1,i}=1}w_{r_{1},i}h(X_{i})=n^{-1}\sum_{i=1}^{n}h(X_{i})$.

\item {}The calibration weights $w_{a_{1}},w_{a_{0}}$ associated with
the treatment: Solve the optimization problem (1) subject to $\sum_{i:A_{i}=1}w_{a_{1},i}h(X_{i})=n^{-1}\sum_{i=1}^{n}h(X_{i})$
to get $w_{a_{1}}$; subject to $\sum_{i:A_{i}=0}w_{a_{0},i}h(X_{i})=n^{-1}\sum_{i=1}^{n}h(X_{i})$
to get $w_{a_{0}}$.

\end{enumerate}
\item[\textbf{\large{}{}Step 4}{\large{}{}.}] {}Get the calibration-based estimator as 
\begin{align*}
\hat{\tau}_{\text{mr-C}} & =\pr_{n}\Bigg[Aw_{a_{1}}\bigg(R_{2}Y_{2}+R_{1}(1-R_{2})\hat{\mu}_{2}^{0}(H_{1})+(1-R_{1})\hat{\mu}_{2}^{0}(H_{0})\\
 & \;-\hat{\pi}_{1}(1,H_{0})\hat{\E}\left\{ \hat{\pi}_{2}(1,H_{1})\hat{\mu}_{2}^{1}(H_{1})\mid H_{0},R_{1}=1,A=1\right\} \\
 & \;-\hat{\pi}_{1}(1,H_{0})\hat{\E}\left[\left\{ 1-\hat{\pi}_{2}(1,H_{1})\right\} \hat{\mu}_{2}^{0}(H_{1})\mid H_{0},R_{1}=1,A=1\right]\\
 & \;-\big\{1-\hat{\pi}_{1}(1,H_{0})\big\}\hat{\mu}_{2}^{0}(H_{0})\bigg)\Bigg]\bigg/\pr_{n}\left(Aw_{a_{1}}\right)\\
 & \;+\pr_{n}\bigg(\hat{\pi}_{1}(1,H_{0})\hat{\E}\left\{ \hat{\pi}_{2}(1,H_{1})\hat{\mu}_{2}^{1}(H_{1})\mid H_{0},R_{1}=1,A=1\right\} \\
 & \;+\hat{\pi}_{1}(1,H_{0})\hat{\E}\left[\left\{ 1-\hat{\pi}_{2}(1,H_{1})\right\} \hat{\mu}_{2}^{0}(H_{1})\mid H_{0},R_{1}=1,A=1\right]-\big\{1-\hat{\pi}_{1}(1,H_{0})\big\}\hat{\mu}_{2}^{0}(H_{0})\bigg)\\
 & +\pr_{n}\bigg[\left(1-A\right)R_{2}w_{a_{0}}w_{r_{1}}w_{r_{2}}\Big[\hat{\pi}_{1}(1,H_{0})\left\{ 1-\hat{\pi}_{2}(1,H_{1})\right\} \hat{\delta}(H_{1})\\
 & \qquad-\hat{\pi}_{1}(1,H_{0})\Big]\big\{ Y_{2}-\hat{\mu}_{2}^{0}(H_{1})\big\}\bigg]\bigg/\left\{ \pr_{n}\left(1-A\right)R_{2}w_{a_{0}}w_{r_{1}}w_{r_{2}}\right\} \\
 & -\pr_{n}\left[\left(1-A\right)R_{1}w_{a_{0}}w_{r_{1}}\hat{\pi}_{1}(1,H_{0})\big\{\hat{\mu}_{2}^{0}(H_{1})-\hat{\mu}_{2}^{0}(H_{0})\big\}\right]\bigg/\pr_{n}\left\{ \left(1-A\right)R_{1}w_{a_{0}}w_{r_{1}}\right\} .
\end{align*}

\end{enumerate}

\section{Proof of the multiple robustness \label{sec:supp_tr}}

{}We prove the multiple robustness and semiparametric efficiency
of the EIF-based estimators. For the cross-sectional data, we prove
the triple robustness in two aspects: consistency when using parametric
models and rate convergence when using flexible models. For the longitudinal
outcomes, we focus on the multiple robustness in terms of the rate
convergence. Throughout the section, we use the estimators motivated
by Theorems 2 and 6 for illustration, which is asymptotically equivalent
to the corresponding normalized and calibration-based estimators.

\subsection{Proof of Theorem 3 \label{subsec:supp_tr_1time}}

\paragraph{Proof of the triple robustness: }

{}Suppose the model estimators $\hat{\theta}=(\hat{\alpha},\hat{\beta},\hat{\gamma})^{\text{T}}$
converges to $\theta^{*}=(\alpha^{*},\beta^{*},\gamma^{*})^{\text{T}}$
in the sense that $\lVert\hat{\theta}-\theta^{*}\rVert=o_{p}(1)$,
where at least one component of $\hat{\theta}$ needs to converge
to the true value. As the sample size $n\rightarrow\infty$, we would
expect $\hat{\tau}_{\text{tr}}$ converges to 
\begin{align}
 & \E\left[\big\{\frac{A}{e(X;\alpha^{*})}-\frac{1-A}{1-e(X;\alpha^{*})}\frac{\pi_{1}(1,X;\gamma^{*})}{\pi_{1}(0,X;\gamma^{*})}\big\} R_{1}\big\{ Y_{1}-\mu_{1}^{0}(X;\beta^{*})\big\}\right]\label{eq:trpart1_1time}\\
- & \E\left[\frac{A-e(X;\alpha^{*})}{e(X;\alpha^{*})}\pi_{1}(1,X;\gamma^{*})\big\{\mu_{1}^{1}(X;\beta^{*})-\mu_{1}^{0}(X;\beta^{*})\big\}\right]\label{eq:trpart2_1time}
\end{align}

{}Rearrange \eqref{eq:trpart1_1time}, we have 
\begin{align*}
 & \E\left[\big\{\frac{A}{e(X;\alpha^{*})}-\frac{1-A}{1-e(X;\alpha^{*})}\frac{\pi_{1}(1,X;\gamma^{*})}{\pi_{1}(0,X;\gamma^{*})}\big\} R_{1}\big\{ Y_{1}-\mu_{1}^{0}(X;\beta^{*})\big\}\right]\\
= & \E\left[\frac{\E(A\mid X)}{e(X;\alpha^{*})}\E(R_{1}\mid X,A=1)\big\{\E(Y_{1}\mid X,R_{1}=1,A=1)-\mu_{1}^{0}(X;\beta^{*})\big\}\right]\\
 & -\E\left[\frac{\E(1-A\mid X)}{1-e(X;\alpha^{*})}\frac{\pi_{1}(1,X;\gamma^{*})}{\pi_{1}(0,X;\gamma^{*})}\E(R_{1}\mid X,A=0)\big\{\E(Y_{1}\mid X,R_{1}=1,A=0)-\mu_{1}^{0}(X;\beta^{*})\big\}\right]\\
= & \E\left[\frac{e(X)}{e(X;\alpha^{*})}\pi_{1}(1,X)\big\{\mu_{1}^{1}(X)-\mu_{1}^{0}(X;\beta^{*})\big\}-\frac{1-e(X)}{1-e(X;\alpha^{*})}\frac{\pi_{1}(0,X)\pi_{1}(1,X;\gamma^{*})}{\pi_{1}(0,X;\gamma^{*})}\big\{\mu_{1}^{0}(X)-\mu_{1}^{0}(X;\beta^{*})\big\}\right]\\
= & \E\left[\pi_{1}(1,X)\big\{\mu_{1}^{1}(X)-\mu_{1}^{0}(X)\big\}+\pi_{1}(1,X)\mu_{1}^{1}(X)\left\{ \frac{e(X)}{e(X;\alpha^{*})}-1\right\} +\pi_{1}(1,X)\left\{ \mu_{1}^{0}(X)-\frac{e(X)}{e(X;\alpha^{*})}\mu_{1}^{0}(X;\beta^{*})\right\} \right]\\
 & -\E\left[\frac{1-e(X)}{1-e(X;\alpha^{*})}\frac{\pi_{1}(0,X)\pi_{1}(1,X;\gamma^{*})}{\pi_{1}(0,X;\gamma^{*})}\big\{\mu_{1}^{0}(X)-\mu_{1}^{0}(X;\beta^{*})\big\}\right]\\
= & \tau_{1}^{\text{CR}}+\E\left[\pi_{1}(1,X)\mu_{1}^{1}(X)\left\{ \frac{e(X)}{e(X;\alpha^{*})}-1\right\} +\pi_{1}(1,X)\left\{ \mu_{1}^{0}(X)-\frac{e(X)}{e(X;\alpha^{*})}\mu_{1}^{0}(X;\beta^{*})\right\} \right]\\
 & -\E\left[\left\{ \frac{1-e(X)}{1-e(X;\alpha^{*})}\frac{\pi_{1}(0,X)\pi_{1}(1,X;\gamma^{*})}{\pi_{1}(0,X;\gamma^{*})}-\frac{e(X)}{e(X;\alpha^{*})}\pi_{1}(1,X)\right\} \mu_{1}^{0}(X;\beta^{*})\right].
\end{align*}
Rearrange \eqref{eq:trpart2_1time}, we have $\E\left[\frac{A-e(X;\alpha^{*})}{e(X;\alpha^{*})}\pi_{1}(1,X;\gamma^{*})\big\{\mu_{1}^{1}(X;\beta^{*})-\mu_{1}^{0}(X;\beta^{*})\big\}\right]$
\[
=\E\left[\frac{e(X)-e(X;\alpha^{*})}{e(X;\alpha^{*})}\pi_{1}(1,X;\gamma^{*})\big\{\mu_{1}^{1}(X;\beta^{*})-\mu_{1}^{0}(X;\beta^{*})\big\}\right].
\]
Combine the two parts together, \eqref{eq:trpart1_1time}+\eqref{eq:trpart2_1time}
\begin{align*}
= & \tau_{1}^{\text{CR}}+\E\left[\pi_{1}(1,X)\left\{ \frac{e(X)}{e(X;\alpha^{*})}-1\right\} \mu_{1}^{1}(X)+\left\{ \pi_{1}(1,X)-\frac{1-e(X)}{1-e(X;\alpha^{*})}\frac{\pi_{1}(0,X)\pi_{1}(1,X;\gamma^{*})}{\pi_{1}(0,X;\gamma^{*})}\right\} \mu_{1}^{0}(X)\right]\\
 & +\E\left[\left\{ \frac{e(X)}{e(X;\alpha^{*})}-1\right\} \pi_{1}(1,X;\gamma^{*})\mu_{1}^{1}(X;\beta^{*})\right]\\
 & -\E\left[\left\{ \frac{1-e(X)}{1-e(X;\alpha^{*})}\frac{\pi_{1}(0,X)\pi_{1}(1,X;\gamma^{*})}{\pi_{1}(0,X;\gamma^{*})}-\frac{e(X)}{e(X;\alpha^{*})}\pi_{1}(1,X)+\frac{e(X)-e(X;\alpha^{*})}{e(X;\alpha^{*})}\pi_{1}(1,X;\gamma^{*})\right\} \mu_{1}^{0}(X;\beta^{*})\right]\\
= & \tau_{1}^{\text{CR}}+\E\left[\left\{ \frac{e(X)}{e(X;\alpha^{*})}-1\right\} \left\{ \pi_{1}(1,X)\mu_{1}^{1}(X)-\pi_{1}(1,X;\gamma^{*})\mu_{1}^{1}(X;\beta^{*})\right\} \right]\\
 & +\E\left[\left\{ 1-\frac{1-e(X)}{1-e(X;\alpha^{*})}\frac{\pi_{1}(0,X)\pi_{1}(1,X;\gamma^{*})}{\pi_{1}(0,X;\gamma^{*})}\right\} \pi_{1}(1,X;\gamma^{*})\left\{ \mu_{1}^{0}(X)-\mu_{1}^{0}(X;\beta^{*})\right\} \right]\\
 & +\E\left[\left\{ \pi_{1}(1,X)-\pi_{1}(1,X;\gamma^{*})\right\} \left\{ \mu_{1}^{0}(X)-\frac{e(X)}{e(X;\alpha^{*})}\mu_{1}^{0}(X;\beta^{*})\right\} \right]
\end{align*}

{}Therefore, the bias of $\hat{\tau}_{\text{tr}}$ converges to 
\begin{align}
 & \E\left[\left\{ \frac{e(X)}{e(X;\alpha^{*})}-1\right\} \left\{ \pi_{1}(1,X)\mu_{1}^{1}(X)-\pi_{1}(1,X;\gamma^{*})\mu_{1}^{1}(X;\beta^{*})\right\} \right]\label{eq:trbias1_1time}\\
+ & \E\left[\left\{ 1-\frac{1-e(X)}{1-e(X;\alpha^{*})}\frac{\pi_{1}(0,X)\pi_{1}(1,X;\gamma^{*})}{\pi_{1}(0,X;\gamma^{*})}\right\} \pi_{1}(1,X;\gamma^{*})\left\{ \mu_{1}^{0}(X)-\mu_{1}^{0}(X;\beta^{*})\right\} \right]\label{eq:trbias2_1time}\\
+ & \E\left[\left\{ \pi_{1}(1,X)-\pi_{1}(1,X;\gamma^{*})\right\} \left\{ \mu_{1}^{0}(X)-\frac{e(X)}{e(X;\alpha^{*})}\mu_{1}^{0}(X;\beta^{*})\right\} \right]\label{eq:trbias3_1time}
\end{align}

{}Note that \eqref{eq:trbias1_1time} $=0$ under $\mathcal{M}_{\text{rp+om}}\cup\mathcal{M}_{\text{ps}}$,
\eqref{eq:trbias2_1time} $=0$ under $\mathcal{M}_{\text{ps+rp}}\cup\mathcal{M}_{\text{om}}$,
\eqref{eq:trbias3_1time} $=0$ under $\mathcal{M}_{\text{ps+om}}\cup\mathcal{M}_{\text{rp}}$.
Thus, $\hat{\tau}_{\text{tr}}$ is consistent for $\tau_{1}^{\text{\jtr}}$
under $\mathcal{M}_{\text{rp+om}}\cup\mathcal{M}_{\text{ps+om}}\cup\mathcal{M}_{\text{ps+rp}}$.
The triple robustness holds.

\paragraph{Proof of the semiparametric efficiency: }

{}We follow the proof in \citet{kennedy2016semiparametric}. To simplify
the notations, denote $\pr\left\{ N(V;\theta_{0})\right\} =\tau_{1}^{\text{\jtr}}$,
where 
\[
N(V;\theta_{0})=\left\{ \frac{A}{e(X)}-\frac{1-A}{1-e(X)}\frac{\pi_{1}(1,X)}{\pi_{1}(0,X)}\right\} R_{1}\left\{ Y_{1}-\mu_{1}^{0}(X)\right\} -\frac{A-e(X)}{e(X)}\pi_{1}(1,X)\left\{ \mu_{1}^{1}(X)-\mu_{1}^{0}(X)\right\} .
\]
Then $\pr\left\{ N(V;\theta^{*})\right\} =\pr\left\{ N(V;\theta_{0})\right\} =\tau^{\text{\jtr}}$.
Consider the decomposition 
\begin{equation}
\hat{\tau}_{\text{tr}}-\tau^{\text{\jtr}}=\left(\pr_{n}-\pr\right)N(V;\hat{\theta})-\pr\left\{ N(V;\hat{\theta})-N(V;\theta^{*})\right\} .\label{eq:tr_decomp_t1}
\end{equation}
Using empirical process theory, if the nuisance functions take values
in Donsker classes, and satisfy the positivity assumption, i.e., there
exists $\varepsilon>0$, such that $\varepsilon<e(X)<1-\varepsilon$
and $\pi_{1}(0,X)>\varepsilon$ for all $X$, then $N(V;\hat{\theta})$
takes values in Donsker classes, and the first term can be written
as 
\[
\left(\pr_{n}-\pr\right)N(V;\hat{\theta})=\left(\pr_{n}-\pr\right)N(V;\theta_{0})+o_{\pr}(n^{-\frac{1}{2}}).
\]

{}For the second term $\pr\left\{ N(V;\hat{\theta})-N(V;\theta^{*})\right\} $,
by computing the expectations, we have 
\begin{align*}
\pr\left\{ N(V;\hat{\theta})-N(V;\theta^{*})\right\}  & =\pr\left[\left\{ \frac{e(X)}{e(X;\hat{\alpha})}-1\right\} \left\{ \pi_{1}(1,X)\mu_{1}^{1}(X)-\pi_{1}(1,X;\hat{\gamma})\mu_{1}^{1}(X;\hat{\beta})\right\} \right]\\
 & +\pr\left[\left\{ 1-\frac{1-e(X)}{1-e(X;\hat{\alpha})}\frac{\pi_{1}(0,X)}{\pi_{1}(0,X;\hat{\gamma})}\right\} \pi_{1}(1,X;\hat{\gamma})\left\{ \mu_{1}^{0}(X)-\mu_{1}^{0}(X;\hat{\beta})\right\} \right]\\
 & +\pr\left[\left\{ \pi_{1}(1,X)-\pi_{1}(1,X;\hat{\gamma})\right\} \left\{ \mu_{1}^{0}(X)-\frac{e(X)}{e(X;\hat{\alpha})}\mu_{1}^{0}(X;\hat{\beta})\right\} \right].
\end{align*}

{}Under the positivity assumptions, we apply Cauchy-Schwarz inequality
($\pr(fg)\leq\lVert f\rVert\lVert g\rVert$) and obtain a upper bound
for the second term as 
\begin{align*}
\pr\left\{ N(V;\hat{\theta})-N(V;\theta^{*})\right\}  & \leq\Big\lVert\frac{e(X)}{e(X;\hat{\alpha})}-1\Big\rVert\cdot\Big\lVert\pi_{1}(1,X)\mu_{1}^{1}(X)-\pi_{1}(1,X;\hat{\gamma})\mu_{1}^{1}(X;\hat{\beta})\Big\rVert\\
 & +\Big\lVert1-\frac{1-e(X)}{1-e(X;\hat{\alpha})}\frac{\pi_{1}(0,X)\pi_{1}(1,X;\hat{\gamma})}{\pi_{1}(0,X;\hat{\gamma})}\Big\rVert\cdot\Big\lVert\pi_{1}(1,X;\hat{\gamma})\left\{ \mu_{1}^{0}(X)-\mu_{1}^{0}(X;\hat{\beta})\right\} \Big\rVert\\
 & +\Big\lVert\pi_{1}(1,X)-\pi_{1}(1,X;\hat{\gamma})\Big\rVert\cdot\Big\lVert\mu_{1}^{0}(X)-\frac{e(X)}{e(X;\hat{\alpha})}\mu_{1}^{0}(X;\hat{\beta})\Big\rVert\\
\text{} & \leq\Big\lVert\left\{ \frac{e(X)}{e(X;\hat{\alpha})}-1\right\} \left\{ \mu_{1}^{1}(X)-\mu_{1}^{1}(X;\hat{\beta})\right\} \Big\rVert_{1}\cdot\Big\lVert\pi_{1}(1,X)\Big\rVert_{\infty}\\
 & +\Big\lVert\left\{ \frac{e(X)}{e(X;\hat{\alpha})}-1\right\} \left\{ \pi_{1}(1,X)-\pi_{1}(1,X;\hat{\gamma})\right\} \Big\rVert_{1}\cdot\Big\lVert\mu_{1}^{1}(X;\hat{\beta})\Big\rVert_{\infty}\\
 & +\Big\lVert\left\{ 1-\frac{1-e(X)}{1-e(X;\hat{\alpha})}\right\} \left\{ \mu_{1}^{0}(X)-\mu_{1}^{0}(X;\hat{\beta})\right\} \Big\rVert_{1}\cdot\Big\lVert\pi_{1}(1,X;\hat{\gamma})\Big\rVert_{\infty}\\
 & +\Big\lVert\left\{ 1-\frac{\pi_{1}(0,X)}{\pi_{1}(0,X;\hat{\gamma})}\right\} \left\{ \mu_{1}^{0}(X)-\mu_{1}^{0}(X;\hat{\beta})\right\} \Big\rVert_{1}\cdot\Big\lVert\frac{1-e(X)}{1-e(X;\hat{\alpha})}\Big\rVert_{\infty}\\
 & +\Big\lVert\pi_{1}(1,X)-\pi_{1}(1,X;\hat{\gamma})\Big\lVert\cdot\Big\lVert\mu_{1}^{0}(X)-\mu_{1}^{0}(X;\hat{\beta})\Big\lVert\\
 & +\Big\lVert\left\{ \pi_{1}(1,X)-\pi_{1}(1,X;\hat{\gamma})\right\} \left\{ \frac{e(X)}{e(X;\hat{\alpha})}-1\right\} \Big\rVert_{1}\cdot\Big\lVert\mu_{1}^{0}(X;\hat{\beta})\Big\rVert_{\infty}\\
 & \leq M\Big\lVert\frac{e(X)}{e(X;\hat{\alpha})}-1\Big\rVert\cdot\left\{ \Big\lVert\mu_{1}^{1}(X)-\mu_{1}^{1}(X;\hat{\beta})\Big\rVert+\Big\lVert\pi_{1}(1,X)-\pi_{1}(1,X;\hat{\gamma})\Big\lVert\right\} \\
 & +M\Big\lVert\mu_{1}^{0}(X)-\mu_{1}^{0}(X;\hat{\beta})\Big\rVert\cdot\left\{ \Big\lVert1-\frac{1-e(X)}{1-e(X;\hat{\alpha})}\Big\rVert+\Big\lVert1-\frac{\pi_{1}(0,X)}{\pi_{1}(0,X;\hat{\gamma})}\Big\lVert\right\} \\
 & +M\Big\lVert\pi_{1}(1,X)-\pi_{1}(1,X;\hat{\gamma})\Big\rVert\cdot\left\{ \Big\lVert\frac{e(X)}{e(X;\hat{\alpha})}-1\Big\rVert+\Big\lVert\mu_{1}^{0}(X)-\mu_{1}^{0}(X;\hat{\beta})\Big\lVert\right\} .
\end{align*}

{}The second inequality holds by the triangle inequality and Holder's
inequality, and the last inequality holds by Cauchy-Schwarz. Under
$\mathcal{M}_{\text{ps+rp+om}}$, we would expect $\pr\left\{ N(V;\hat{\theta})-N(V;\theta^{*})\right\} =O_{\pr}(n^{-1/2})\cdot o_{\pr}(1)=o_{\pr}(n^{-1/2})$.
Therefore, the EIF-based estimator $\hat{\tau}_{\text{tr}}$ satisfies
$\hat{\tau}_{\text{tr}}-\tau_{1}^{\jtr}=\left(\pr_{n}-\pr\right)N(V;\theta_{0})+o_{\pr}(n^{-\frac{1}{2}})$
and its influence function $N(V;\theta_{0})+\tau_{1}^{\text{\jtr}}$,
which is the same as the EIF in Theorem 2 and completes the proof.

\subsection{Proof of Theorem 4 and Corollary 1}

\paragraph{Proof of Theorem 4: }

{}When using flexible models, we let $\theta$ consist of all the
nuisance functions $\big\{ e(X),\pi_{1}(a,X),\mu_{1}^{a}(X):a=0,1\big\}$,
and $\hat{\theta}$ be its limit. We use the same notations in \ref{subsec:supp_tr_1time},
and consider the same decomposition as formula \eqref{eq:tr_decomp_t1}.
Using empirical process theory, if the nuisance functions take values
in Donsker classes, and satisfy the positivity assumption, i.e., there
exists $\varepsilon>0$, such that $\varepsilon<e(X)<1-\varepsilon$
and $\pi_{1}(0,X)>\varepsilon$ for all $X$, then $N(V;\hat{\theta})$
takes values in Donsker classes, and the first term can be written
as 
\begin{align*}
\left(\pr_{n}-\pr\right)N(V;\hat{\theta}) & =\left(\pr_{n}-\pr\right)N(V;\theta_{0})+o_{\pr}(n^{-\frac{1}{2}}).
\end{align*}

{}For the second term $\pr\left\{ N(V;\hat{\theta})-N(V;\theta^{*})\right\} $,
by computing the expectations, we have 
\begin{align*}
\pr\left\{ N(V;\hat{\theta})-N(V;\theta^{*})\right\}  & =\pr\left[\left\{ \frac{e(X)}{\hat{e}(X)}-1\right\} \left\{ \pi_{1}(1,X)\mu_{1}^{1}(X)-\hat{\pi}_{1}(1,X)\hat{\mu}_{1}^{1}(X)\right\} \right]\\
 & +\pr\left[\left\{ 1-\frac{1-e(X)}{1-\hat{e}(X)}\frac{\pi_{1}(0,X)}{\hat{\pi}_{1}(0,X)}\right\} \hat{\pi}_{1}(1,X)\left\{ \mu_{1}^{0}(X)-\hat{\mu}_{1}^{0}(X)\right\} \right]\\
 & +\pr\left[\left\{ \pi_{1}(1,X)-\hat{\pi}_{1}(1,X)\right\} \left\{ \mu_{1}^{0}(X)-\frac{e(X)}{\hat{e}(X)}\hat{\mu}_{1}^{0}(X)\right\} \right]=\text{Rem}(\hat{\pr},\pr)
\end{align*}
Therefore, $\hat{\tau}_{\text{tr}}-\tau_{1}^{\text{\jtr}}=\left(\pr_{n}-\pr\right)N(V;\theta_{0})+\text{Rem}(\hat{\pr},\pr)+o_{\pr}(n^{-\frac{1}{2}})=\pr_{n}\left\{ \varphi_{1}^{\text{\jtr}}(V_{i};\pr)\right\} +\text{Rem}(\hat{\pr},\pr)+o_{\pr}(n^{-1/2})$.
If $\text{Rem}(\hat{\pr},\pr)=o_{\pr}(n^{-1/2})$, then $\hat{\tau}_{\text{tr}}-\tau_{1}^{\text{\jtr}}=n^{-1}\sum_{i=1}^{n}\varphi^{\text{\jtr}}(V_{i};\pr)+o_{\pr}(n^{-1/2})$.
Apply central limit theorem and we complete the proof.

\paragraph{Proof of Corollary 1:}

{}For the remainder term, based on the uniform bounded condition,
apply Cauchy-Schwarz and Holder's inequality, we have

{} 
\begin{align*}
\pr\left\{ N(V;\hat{\theta})-N(V;\theta^{*})\right\}  & \leq M\Big\lVert\frac{e(X)}{\hat{e}(X)}-1\Big\rVert\cdot\left\{ \Big\lVert\mu_{1}^{1}(X)-\hat{\mu}_{1}^{1}(X)\Big\rVert+\Big\lVert\pi_{1}(1,X)-\hat{\pi}_{1}(1,X)\Big\lVert\right\} \\
 & +M\Big\lVert\mu_{1}^{0}(X)-\hat{\mu}_{1}^{0}(X)\Big\rVert\cdot\left\{ \Big\lVert1-\frac{1-e(X)}{1-\hat{e}(X)}\Big\rVert+\Big\lVert1-\frac{\pi_{1}(0,X)}{\hat{\pi}_{1}(0,X)}\Big\lVert\right\} \\
 & +M\Big\lVert\pi_{1}(1,X)-\hat{\pi}_{1}(1,X)\Big\rVert\cdot\left\{ \Big\lVert\frac{e(X)}{\hat{e}(X)}-1\Big\rVert+\Big\lVert\mu_{1}^{0}(X)-\hat{\mu}_{1}^{0}(X)\Big\lVert\right\} .
\end{align*}
With the convergence rate $\lVert\hat{e}(X)-e(X)\rVert=o_{\pr}(n^{-c_{e}}),\lVert\hat{\mu}_{1}^{a}(X)-\mu_{1}^{a}(X)\rVert=o_{\pr}(n^{-c_{\mu}}),\lVert\hat{\pi}_{1}(a,X)-\pi_{1}(a,X)\rVert=o_{\pr}(n^{-c_{\pi}})$,
and by Theorem 4 based on the central limit theorem, we have $\hat{\tau}_{\text{tr}}-\tau_{1}^{\text{\jtr}}=O_{\pr}(n^{-1/2}+n^{-c})$,
where $c=\min(r_{e}+r_{\pi},r_{e}+r_{\mu},r_{\pi}+r_{\mu})$, which
completes the proof.

\subsection{Proof of Theorem 7 and Corollary 2 \label{subsec:supp_tr_longi}}

\paragraph{Proof of Theorem 7: }

{}When using flexible models, we let $\theta$ consist of all the
nuisance functions{}{} ${\color{black}\big\{ e(H_{s-1}),\pi_{s}(a,H_{s-1}),\mu_{t}^{a}(H_{s-1}),g_{s+1}^{1}(H_{l-1}):l=1,\cdots,s\text{ and }s=1,\cdots,t;a=0,1\big\}}${},
and $\hat{\theta}$ be its limit. We use the same notations in \ref{subsec:supp_tr_1time},
and denote $N(V;\theta):=\varphi_{t}^{\text{\jtr}}(V;\pr)+\tau_{t}^{\text{\jtr}}$.
Consider the same decomposition as formula \eqref{eq:tr_decomp_t1}.

{}Using empirical process theory, if the nuisance functions take
values in Donsker classes, and satisfy the positivity assumption,
i.e., there exists $\varepsilon>0$, such that $\varepsilon<\big\{ e(H_{s-1}),\hat{e}(H_{s-1})\big\}<1-\varepsilon$
and $\big\{\pi_{s}(0,H_{s-1}),\hat{\pi}_{s}(0,H_{s-1})\big\}>\varepsilon$
for all $H_{s-1}$ when $s=1,\cdots,t$, then $N(V;\hat{\theta})$
takes values in Donsker classes, and the first term can be written
as 
\[
\left(\pr_{n}-\pr\right)N(V;\hat{\theta})=\left(\pr_{n}-\pr\right)N(V;\theta_{0})+o_{\pr}(n^{-\frac{1}{2}}).
\]

{}For the second term $\pr\left\{ N(V;\hat{\theta})-N(V;\theta^{*})\right\} $,
we proceed by deriving the expectations of $N(V;\hat{\theta})-N(V;\theta^{*})$.
Note that $\pr\left\{ N(V;\hat{\theta})\right\} $ equals to 
\begin{align}
 & \pr\Bigg\{\frac{A}{\hat{e}(H_{0})}\big\{ R_{t}Y_{t}+\sum_{s=1}^{t}R_{s-1}(1-R_{s})\hat{\mu}_{t}^{0}(H_{s-1})\big\}\label{eq:tr_longi1}\\
 & +\left\{ 1-\frac{A}{\hat{e}(H_{0})}\right\} \left[\hat{\pi}_{1}(1,H_{0})\sum_{s=1}^{t-1}\hat{g}_{s+1}^{1}(H_{0})+\big\{1-\hat{\pi}_{1}(1,H_{0})\big\}\hat{\mu}_{t}^{0}(H_{0})\right]-\hat{\mu}_{t}^{0}(H_{0})\label{eq:tr_longi2}\\
 & +\frac{1-A}{1-\hat{e}(H_{0})}\bigg(\sum_{s=1}^{t}\left[\sum_{k=1}^{s}\hat{\bar{\pi}}_{k-1}(0,H_{k-2})\{1-\hat{\pi}_{k}(1,H_{k-1})\}\hat{\delta}(H_{k-1})-1\right]\frac{R_{s}}{\hat{\bar{\pi}}_{s}(0,H_{s-1})}\big\{\hat{\mu}_{t}^{0}(H_{s})-\hat{\mu}_{t}^{0}(H_{s-1})\big\}\bigg)\Bigg\}.\label{eq:tr_longi3}
\end{align}
By iterated expectations, the first term \eqref{eq:tr_longi1} and
the second term \eqref{eq:tr_longi2} equal to 
\begin{align*}
\pr & \bigg(\frac{e(H_{0})}{\hat{e}(H_{0})}\left[\pi_{1}(1,H_{0})g_{t+1}^{1}(H_{0})+\sum_{s=1}^{t-1}\pi_{1}(1,H_{0})g_{\hat{\mu},s+1}^{1}(H_{0})+\left\{ 1-\pi_{1}(1,H_{0})\right\} \hat{\mu}_{t}^{0}(H_{s-1})\right]\\
 & +\left\{ 1-\frac{e(H_{0})}{\hat{e}(H_{0})}\right\} \left[\hat{\pi}_{1}(1,H_{0})\sum_{s=1}^{t-1}\hat{g}_{s+1}^{1}(H_{0})+\big\{1-\hat{\pi}_{1}(1,H_{0})\big\}\hat{\mu}_{t}^{0}(H_{0})\right]-\hat{\mu}_{t}^{0}(H_{0})\bigg),
\end{align*}
using the notations in the main text.

{}For the third term \eqref{eq:tr_longi3}, for $s=1,\cdots,t$,
we have 
\begin{align*}
 & \E\bigg(\frac{1-A}{1-\hat{e}(H_{0})}\Big[\hat{\bar{\pi}}_{k-1}(0,H_{k-2})\{1-\hat{\pi}_{k}(1,H_{k-1})\}\hat{\delta}(H_{k-1})-1\Big]\frac{R_{s}}{\hat{\bar{\pi}}_{s}(0,H_{s-1})}\big\{\hat{\mu}_{t}^{0}(H_{s})-\hat{\mu}_{t}^{0}(H_{s-1})\big\}\\
 & \qquad\mid H_{s-1},R_{s-1}=1,A=0\bigg)\\
= & \E\Bigg(\frac{1-A}{1-\hat{e}(H_{0})}\text{\ensuremath{\left[\hat{\bar{\pi}}_{k-1}(0,H_{k-2})\{1-\hat{\pi}_{k}(1,H_{k-1})\}\hat{\delta}(H_{k-1})-1\right]}}\left[\E\left\{ \hat{\mu}_{t}^{0}(H_{s})\mid H_{s-1},R_{s}=1,A=0\right\} -\hat{\mu}_{t}^{0}(H_{s-1})\right]\\
 & \qquad\frac{R_{s-1}}{\hat{\bar{\pi}}_{s-1}(0,H_{s-2})}\frac{\pi_{s}(0,H_{s-1})}{\hat{\bar{\pi}}_{s}(0,H_{s-1})}\mid H_{s-1},R_{s-1}=1,A=0\Bigg).
\end{align*}
And for $k=1,\cdots,s$, apply iterated expectations to the above
formula and use the notation in the main text, we have 
\begin{align*}
\E & \bigg(\frac{1-A}{1-\hat{e}(H_{0})}\left[\hat{\bar{\pi}}_{k-1}(0,H_{k-2})\{1-\hat{\pi}_{k}(1,H_{k-1})\}\hat{\delta}(H_{k-1})-1\right]\frac{R_{k}}{\hat{\bar{\pi}}_{k}(0,H_{k-1})}\\
 & \prod_{l=k+1}^{s}\frac{R_{l}}{\hat{\pi}_{l}(0,H_{l-1})}\left[\E\left\{ \hat{\mu}_{t}^{0}(H_{s})\mid H_{s-1},R_{s}=1,A=0\right\} -\hat{\mu}_{t}^{0}(H_{s-1})\right]\bigg)\\
=\E & \Bigg[\frac{1-A}{1-\hat{e}(H_{0})}\left[\hat{\bar{\pi}}_{k-1}(0,H_{k-2})\{1-\hat{\pi}_{k}(1,H_{k-1})\}\hat{\delta}(H_{k-1})-1\right]\frac{R_{k}}{\hat{\bar{\pi}}_{k}(0,H_{k-1})}\\
 & \E\bigg\{\frac{\pi_{k+1}(0,H_{k})}{\hat{\pi}_{k+1}(0,H_{k})}\cdots\\
 & \E\bigg(\frac{\pi_{s}(0,H_{s-1})}{\hat{\pi}_{s}(0,H_{s-1})}\left[\E\left\{ \hat{\mu}_{t}^{0}(H_{s})\mid H_{s-1},R_{s}=1,A=0\right\} -\hat{\mu}_{t}^{0}(H_{s-1})\right]\mid H_{s-2,}R_{s-1}=1,A=0\bigg)\\
 & \cdots\mid H_{k-1},R_{k}=1,A=0\bigg\}\Bigg]\\
=\E & \Bigg\{\frac{1-A}{1-\hat{e}(H_{0})}\left[\hat{\bar{\pi}}_{k-1}(0,H_{k-2})\{1-\hat{\pi}_{k}(1,H_{k-1})\}\hat{\delta}(H_{k-1})-1\right]\frac{R_{k-1}}{\hat{\bar{\pi}}_{k-1}(0,H_{k-2})}\\
 & E_{0,s-2}\left(\prod_{l=k}^{s}\frac{\pi_{l}(0,H_{l-1})}{\hat{\pi}_{l}(0,H_{l-1})}\left[\E\left\{ \hat{\mu}_{t}^{0}(H_{s})\mid H_{s-1},R_{s}=1,A=0\right\} -\hat{\mu}_{t}^{0}(H_{s-1})\right];H_{k-1}\right)\Bigg\}.
\end{align*}
Continue the calculation, the above formula becomes 
\begin{align}
=\E & \Bigg\{\frac{1-A}{1-\hat{e}(H_{0})}\left[R_{k-1}\{1-\hat{\pi}_{k}(1,H_{k-1})\}\hat{\delta}(H_{k-1})-\frac{R_{k-1}}{\hat{\bar{\pi}}_{k-1}(0,H_{k-2})}\right]\nonumber \\
 & E_{0,s-2}\left(\prod_{l=k}^{s}\frac{\pi_{l}(0,H_{l-1})}{\hat{\pi}_{l}(0,H_{l-1})}\left[\E\left\{ \hat{\mu}_{t}^{0}(H_{s})\mid H_{s-1},R_{s}=1,A=0\right\} -\hat{\mu}_{t}^{0}(H_{s-1})\right];H_{k-1}\right)\Bigg\}\nonumber \\
=\E & \left[\frac{1-A}{1-\hat{e}(H_{0})}R_{k-1}\{1-\hat{\pi}_{k}(1,H_{k-1})\}\frac{\hat{\delta}(H_{k-1})}{\delta(H_{k-1})}\delta(H_{k-1})G_{\hat{\mu},\hat{\pi},s-2}(H_{k-1})\right]\label{eq:tr_inter1}\\
- & \E\left\{ \frac{1-A}{1-\hat{e}(H_{0})}\frac{R_{k-1}}{\hat{\bar{\pi}}_{k-1}(0,H_{k-2})}G_{\hat{\mu},\hat{\pi},s-2}(H_{k-1})\right\} \label{eq:tr_inter2}
\end{align}
if we denote

{}$G_{\hat{\mu},\hat{\pi},s-2}(H_{k-1})=E_{0,s-2}\bigg(\prod_{l=k}^{s}\pi_{l}(0,H_{l-1})\Big[\E\left\{ \hat{\mu}_{t}^{0}(H_{s})\mid H_{s-1},R_{s}=1,A=0\right\} -\hat{\mu}_{t}^{0}(H_{s-1})\Big]\big/\hat{\pi}_{l}(0,H_{l-1});H_{k-1}\bigg)$
to indicate the involvement of the estimated nuisance function $\hat{\mu}_{t}^{0}(H_{l-1})$
and $\hat{\pi}_{l}(0,H_{l-1})$ for $l=k,\cdots,s$.

{}For the first term (\ref{eq:tr_inter1}), by Bayes' rule, 
\[
\delta(H_{s-1})=\frac{\bar{\pi}_{s-1}(1,H_{s-2})}{\bar{\pi}_{s-1}(0,H_{s-2})}\prod_{l=1}^{s-1}\frac{f(Y_{l}\mid H_{l-1},R_{l}=1,A=1)}{f(Y_{l}\mid H_{l-1},R_{l}=1,A=0)}.
\]
Take iterated expectations conditional on the historical information,
it equals to 
\begin{align*}
\E & \Big[\frac{1-A}{1-\hat{e}(H_{0})}\frac{R_{k-1}}{\bar{\pi}_{k-1}(0,H_{k-2})}\bar{\pi}_{k-1}(1,H_{k-2})\{1-\hat{\pi}_{k}(1,H_{k-1})\}\frac{\hat{\delta}(H_{k-1})}{\delta(H_{k-1})}\\
 & \prod_{l=1}^{k-1}\frac{f(y_{l}\mid H_{l-1},R_{l}=1,A=1)}{f(y_{l}\mid H_{l-1},R_{l}=1,A=0)}G_{\hat{\mu},\hat{\pi},s-2}(H_{k-1})\Big]\\
= & \E\bigg(\frac{1-A}{1-\hat{e}(H_{0})}\frac{R_{k-1}}{\bar{\pi}_{k-1}(0,H_{k-2})}\bar{\pi}_{k-1}(1,H_{k-2})\prod_{l=1}^{k-2}\frac{f(Y_{l}\mid H_{l-1},R_{l}=1,A=1)}{f(Y_{l}\mid H_{l-1},R_{l}=1,A=0)}\\
 & \E\Big[\{1-\hat{\pi}_{k}(1,H_{k-1})\}\frac{\hat{\delta}(H_{k-1})}{\delta(H_{k-1})}\frac{f(Y_{k-1}\mid H_{k-2},R_{k-1}=1,A=1)}{f(Y_{k-1}\mid H_{k-2},R_{k-1}=1,A=0)}G_{\hat{\mu},\hat{\pi},s-2}(H_{k-1})\mid H_{k-2},R_{k-1}=1,A=0\Big]\bigg)\\
= & \E\bigg(\frac{1-A}{1-\hat{e}(H_{0})}\frac{R_{k-2}}{\bar{\pi}_{k-2}(0,H_{k-1})}\bar{\pi}_{k-1}(1,H_{k-2})\prod_{l=1}^{k-2}\frac{f(Y_{l}\mid H_{l-1},R_{l}=1,A=1)}{f(Y_{l}\mid H_{l-1},R_{l}=1,A=0)}\\
 & \E\left[\{1-\hat{\pi}_{k}(1,H_{k-1})\}\frac{\hat{\delta}(H_{k-1})}{\delta(H_{k-1})}G_{\hat{\mu},\hat{\pi},s-2}(H_{k-1})\mid H_{k-2},R_{k-1}=1,A=1\right]\bigg)\\
= & \E\bigg(\cdots\E\Big[\frac{1-e(H_{0})}{1-\hat{e}(H_{0})}\bar{\pi}_{k-1}(1,H_{k-2})\{1-\hat{\pi}_{k}(1,H_{k-1})\}\frac{\hat{\delta}(H_{k-1})}{\delta(H_{k-1})}\\
 & G_{\hat{\mu},\hat{\pi},s-2}(H_{k-1})\mid H_{k-2},R_{k-1}=1,A=1\Big]\cdots\mid H_{0},R_{1}=1,A=1\bigg)\\
:= & E_{1,k-2}\left[\frac{1-e(H_{0})}{1-\hat{e}(H_{0})}\bar{\pi}_{k-1}(1,H_{k-2})\{1-\hat{\pi}_{k}(1,H_{k-1})\}\frac{\hat{\delta}(H_{k-1})}{\delta(H_{k-1})}G_{\hat{\mu},\hat{\pi},s-2}(H_{k-1});H_{0}\right].
\end{align*}

{}For the second term (\ref{eq:tr_inter2}), again by iterated expectations,\textbf{{}
\begin{align*}
= & \E\Big[\cdots\E\left\{ \frac{1-e(H_{0})}{1-\hat{e}(H_{0})}\frac{\bar{\pi}_{k-1}(0,H_{k-2})}{\hat{\bar{\pi}}_{k-1}(0,H_{k-2})}G_{\hat{\mu},\hat{\pi},s-2}(H_{k-1})\mid H_{k-2},R_{k-1}=1,A=0\right\} \cdots\mid H_{0},R_{1}=1,A=0\Big]\\
= & E_{0,0}\left\{ \frac{1-e(H_{0})}{1-\hat{e}(H_{0})}G_{\hat{\mu},\hat{\pi},s-2}(H_{0})\right\} \text{ (by the definition of \ensuremath{G_{\hat{\mu},\hat{\pi},s-2}(H_{0})}).}
\end{align*}
}

{}Therefore, as the sample size $n\rightarrow\infty$, the multiply
robust estimator $\hat{\tau}_{\text{mr}}$ converges to 
\begin{align*}
\E & \bigg(\frac{e(H_{0})}{\hat{e}(H_{0})}\left[\pi_{1}(1,H_{0})g_{t+1}^{1}(H_{0})+\sum_{s=1}^{t-1}\pi_{1}(1,H_{0})g_{\hat{\mu},s+1}^{1}(H_{0})+\left\{ 1-\pi_{1}(1,H_{0})\right\} \hat{\mu}_{t}^{0}(H_{0})\right]\\
 & +\left\{ 1-\frac{e(H_{0})}{\hat{e}(H_{0})}\right\} \left[\sum_{s=1}^{t}\hat{\pi}_{1}(1,H_{0})\hat{g}_{s+1}^{1}(H_{0})+\left\{ 1-\hat{\pi}_{1}(1,H_{0})\right\} \hat{\mu}_{t}^{0}(H_{0})\right]-\hat{\mu}_{t}^{0}(H_{0})\\
 & +\frac{1-e(H_{0})}{1-\hat{e}(H_{0})}\sum_{s=1}^{t}\left\{ \sum_{k=1}^{s}E_{1,k-2}\left[\bar{\pi}_{k-1}(1,H_{k-2})\{1-\hat{\pi}_{k}(1,H_{k-1})\}\frac{\hat{\delta}(H_{k-1})}{\delta(H_{k-1})}G_{\hat{\mu},\hat{\pi},s-2}(H_{k-1});H_{0}\right]-G_{\hat{\mu},\hat{\pi},s-2}(H_{0})\right\} \bigg).
\end{align*}

{}Rearrange the terms, we can get the formula for $\pr\left\{ N(V;\hat{\theta})-N(V;\theta^{*})\right\} $
as 
\begin{align*}
\E & \bigg(\left\{ \frac{e(H_{0})}{\hat{e}(H_{0})}-1\right\} \pi_{1}(1,H_{0})g_{t+1}^{1}(H_{0})+\frac{e(H_{0})}{\hat{e}(H_{0})}\pi_{1}(1,H_{0})\sum_{s=1}^{t-1}g_{\hat{\mu},s+1}^{1}(H_{0})-\pi_{1}(1,H_{0})\sum_{s=1}^{t-1}g_{s+1}^{1}(H_{0})\\
 & +\frac{e(H_{0})}{\hat{e}(H_{0})}\left\{ 1-\pi_{1}(1,H_{0})\right\} \hat{\mu}_{t}^{0}(H_{0})+\pi_{1}(1,H_{0})\mu_{t}^{0}(H_{0})\\
 & +\left\{ 1-\frac{e(H_{0})}{\hat{e}(H_{0})}\right\} \left[\sum_{s=1}^{t}\hat{\pi}_{1}(1,H_{0})\hat{g}_{s+1}^{1}(H_{0})+\left\{ 1-\hat{\pi}_{1}(1,H_{0})\right\} \hat{\mu}_{t}^{0}(H_{0})\right]-\hat{\mu}_{t}^{0}(H_{0})\\
 & +\frac{1-e(H_{0})}{1-\hat{e}(H_{0})}\sum_{s=1}^{t}\left\{ \sum_{k=1}^{s}E_{1,k-2}\left[\bar{\pi}_{k-1}(1,H_{k-2})\{1-\hat{\pi}_{k}(1,H_{k-1})\}\frac{\hat{\delta}(H_{k-1})}{\delta(H_{k-1})}G_{\hat{\mu},\hat{\pi},s-2}(H_{k-1});H_{0}\right]-G_{\hat{\mu},\hat{\pi},s-2}(H_{0})\right\} \bigg).
\end{align*}

{}For the terms related to $g_{t+1}^{1}(H_{0})$ , we have 
\begin{align*}
 & \E\left[\left\{ \frac{e(H_{0})}{\hat{e}(H_{0})}-1\right\} \pi_{1}(1,H_{0})g_{t+1}^{1}(H_{0})+\left\{ 1-\frac{e(H_{0})}{\hat{e}(H_{0})}\right\} \hat{\pi}_{1}(1,H_{0})\hat{g}_{t+1}^{1}(H_{0})\right]\\
= & \E\left[\left\{ \frac{e(H_{0})}{\hat{e}(H_{0})}-1\right\} \left\{ \pi_{1}(1,H_{0})g_{t+1}^{1}(H_{0})-\hat{\pi}_{1}(1,H_{0})\hat{g}_{t+1}^{1}(H_{0})\right\} \right].
\end{align*}
For the terms with $s$ layers of expectations and the condition $A=1$
for $s=1,\cdots,t$, we have 
\begin{align*}
\E & \Big[\frac{e(H_{0})}{\hat{e}(H_{0})}\pi_{1}(1,H_{0})g_{\hat{\mu},s+1}^{1}(H_{0})-\pi_{1}(1,H_{0})g_{s+1}^{1}(H_{0})\\
 & +\left\{ 1-\frac{e(H_{0})}{\hat{e}(H_{0})}\right\} \hat{\pi}_{1}(1,H_{0})\hat{g}_{s+1}^{1}(H_{0})+\sum_{l=s}^{t}E_{1,s-1}\left[\bar{\pi}_{s}(1,H_{s-1})\{1-\hat{\pi}_{s+1}(1,H_{s})\}\frac{\hat{\delta}(H_{s})}{\delta(H_{s})}G_{\hat{\mu},\hat{\pi},l}(H_{s});H_{0}\right]\Big]\\
=\E\Biggl[ & \left\{ \frac{e(H_{0})}{\hat{e}(H_{0})}-1\right\} \left\{ \pi_{1}(1,H_{0})g_{\hat{\mu},s+1}^{1}(H_{0})-\hat{\pi}_{1}(1,H_{0})\hat{g}_{s+1}^{1}(H_{0})\right\} \\
 & +\sum_{l=s}^{t}\E\Bigg\{\cdots\E\Bigg\{\E\bigg(\cdots\E\bigg(\bar{\pi}_{s}(1,H_{s-1})\Big[\frac{1-e(H_{0})}{1-\hat{e}(H_{0})}\{1-\hat{\pi}_{s+1}(1,H_{s})\}\frac{\hat{\delta}(H_{s})}{\delta(H_{s-1})}\prod_{k=s+1}^{l}\frac{\pi_{l}(0,H_{l-1})}{\hat{\pi}_{l}(0,H_{l-1})}\\
 & -\{1-\pi_{s+1}(1,H_{s})\}\Big]\left\{ \hat{\mu}_{t}^{0}(H_{l})-\hat{\mu}_{t}^{0}(H_{l-1})\right\} \mid H_{l-1},R_{l}=1,A=0\bigg)\cdots\mid H_{s},R_{s+1}=1,A=0\bigg)\\
 & \mid H_{s-1},R_{s}=1,A=1\Bigg\}\cdots\mid H_{0},R_{1}=1,A=1\Bigg\}\Biggl]\\
=\E\Biggl[ & \left\{ \frac{e(H_{0})}{\hat{e}(H_{0})}-1\right\} \left\{ \pi_{1}(1,H_{0})g_{\hat{\mu},s+1}^{1}(H_{0})-\hat{\pi}_{1}(1,H_{0})\hat{g}_{s+1}^{1}(H_{0})\right\} \\
 & +\sum_{l=s}^{t}E_{1,s-1}\Bigg\{ E_{0,l-1}\bigg(\bar{\pi}_{s}(1,H_{s-1})\Big[\frac{1-e(H_{0})}{1-\hat{e}(H_{0})}\{1-\hat{\pi}_{s+1}(1,H_{s})\}\frac{\hat{\delta}(H_{s})}{\delta(H_{s-1})}\prod_{k=s+1}^{l}\frac{\pi_{l}(0,H_{l-1})}{\hat{\pi}_{l}(0,H_{l-1})}\\
 & -\{1-\pi_{s+1}(1,H_{s})\}\Big]\left\{ \hat{\mu}_{t}^{0}(H_{l})-\hat{\mu}_{t}^{0}(H_{l-1})\right\} ;H_{s}\bigg);H_{0}\Bigg\}\Bigg].
\end{align*}
For the rest terms with the condition $A=0$, we have 
\begin{align*}
\E\Big[ & \frac{e(H_{0})}{\hat{e}(H_{0})}\left\{ 1-\pi_{1}(1,H_{0})\right\} \hat{\mu}_{t}^{0}(H_{0})+\pi_{1}(1,H_{0})\mu_{t}^{0}(H_{0})\\
+ & \left\{ 1-\frac{e(H_{0})}{\hat{e}(H_{0})}\right\} \left\{ 1-\hat{\pi}_{1}(1,H_{0})\right\} \hat{\mu}_{t}^{0}(H_{0})-\hat{\mu}_{t}^{0}(H_{0})\\
+ & \frac{1-e(H_{0})}{1-\hat{e}(H_{0})}\left\{ 1-\hat{\pi}_{1}(1,H_{0})-1\right\} \frac{\pi_{1}(0,H_{0})}{\hat{\pi}_{1}(0,H_{0})}\sum_{s=1}^{t}G_{\hat{\mu},\hat{\pi},s-2}(H_{0})\Big]\\
=\E\Biggl[ & \left\{ \hat{\pi}_{1}(1,H_{0})-\pi_{1}(1,H_{0})\right\} \left\{ \frac{e(H_{0})}{\hat{e}(H_{0})}\hat{\mu}_{t}^{0}(H_{0})-\mu_{t}^{0}(H_{0})\right\} +\hat{\pi}_{1}(1,H_{0})\left\{ \mu_{t}^{0}(H_{0})-\hat{\mu}_{t}^{0}(H_{0})\right\} \\
 & -\frac{1-e(H_{0})}{1-\hat{e}(H_{0})}\hat{\pi}_{1}(1,H_{0})\frac{\pi_{1}(0,H_{0})}{\hat{\pi}_{1}(0,H_{0})}\Bigg\{\sum_{s=1}^{t}\E\bigg(\cdots\\
 & \E\left[\prod_{l=2}^{s}\frac{\pi_{l}(0,H_{l-1})}{\hat{\pi}_{l}(0,H_{l-1})}\left\{ \hat{\mu}_{t}^{0}(H_{s})-\hat{\mu}_{t}^{0}(H_{s-1})\right\} \mid H_{s-1},R_{s}=1,A=0\right]\cdots\mid H_{0},R_{1}=1,A=0\bigg)\Bigg\}\Biggl]\\
=\E\Biggl[ & \left\{ \hat{\pi}_{1}(1,H_{0})-\pi_{1}(1,H_{0})\right\} \left\{ \frac{e(H_{0})}{\hat{e}(H_{0})}\hat{\mu}_{t}^{0}(H_{0})-\mu_{t}^{0}(H_{0})\right\} +\hat{\pi}_{1}(1,H_{0})\Bigg\{\sum_{s=1}^{t}\E\bigg(\cdots\\
 & \E\left[\left\{ 1-\frac{1-e(H_{0})}{1-\hat{e}(H_{0})}\frac{\bar{\pi}_{s}(0,H_{s-1})}{\hat{\bar{\pi}}_{s}(0,H_{s-1})}\right\} \left\{ \hat{\mu}_{t}^{0}(H_{s})-\hat{\mu}_{t}^{0}(H_{s-1})\right\} \mid H_{s-1},R_{s}=1,A=0\right]\\
 & \cdots\mid H_{0},R_{1}=1,A=0\bigg)\Bigg\}\Biggl]\text{ (since \ensuremath{\mu_{t}^{0}(H_{0})-\hat{\mu}_{t}^{0}(H_{0})=\sum_{s=1}^{t}\E\left\{ \hat{\mu}_{t}^{0}(H_{s})-\hat{\mu}_{t}^{0}(H_{s-1})\mid H_{0},R_{1}=1,A=0\right\} })}.
\end{align*}

{}Summarize $\pr\left\{ N(V;\hat{\theta})-N(V;\theta^{*})\right\} $,
which is the remainder term $\text{Rem}(\hat{\pr},\pr)$, we have
\begin{align}
=\E\Bigg[ & \left\{ \frac{e(H_{0})}{\hat{e}(H_{0})}-1\right\} \left\{ \pi_{1}(1,H_{0})g_{t+1}^{1}(H_{0})-\hat{\pi}_{1}(1,H_{0})\hat{g}_{t+1}^{1}(H_{0})\right\} \nonumber \\
 & +\left\{ \frac{e(H_{0})}{\hat{e}(H_{0})}-1\right\} \left[\sum_{s=2}^{t}\left\{ \pi_{1}(1,H_{0})g_{\hat{\mu},s+1}^{1}(H_{0})-\hat{\pi}_{1}(1,H_{0})\hat{g}_{t+1}^{1}(H_{0})\right\} \right]\label{eq:tr_ps-rpom}\\
 & +\sum_{s=1}^{t-1}\sum_{l=s+1}^{t}E_{1,s-1}\Bigg\{ E_{0,l-1}\bigg(\bar{\pi}_{s}(1,H_{s-1})\Big[\frac{1-e(H_{0})}{1-\hat{e}(H_{0})}\{1-\hat{\pi}_{s+1}(1,H_{s})\}\frac{\hat{\delta}(H_{s-1})}{\delta(H_{s-1})}\prod_{k=s+1}^{l}\frac{\pi_{k}(0,H_{k-1})}{\hat{\pi}_{k}(0,H_{k-1})}\nonumber \\
 & \qquad\qquad\qquad\qquad\qquad\qquad-\{1-\pi_{s+1}(1,H_{s})\}\Big]\left\{ \hat{\mu}_{t}^{0}(H_{l})-\hat{\mu}_{t}^{0}(H_{l-1})\right\} ;H_{s}\bigg);H_{0}\Bigg\}\label{eq:tr_om-psrp1}\\
 & +\left\{ \hat{\pi}_{1}(1,H_{0})-\pi_{1}(1,H_{0})\right\} \left\{ \frac{e(H_{0})}{\hat{e}(H_{0})}\hat{\mu}_{t}^{0}(H_{0})-\mu_{t}^{0}(H_{0})\right\} \label{eq:tr_rp-psom}\\
 & +\hat{\pi}_{1}(1,H_{0})\left(\sum_{s=1}^{t}E_{0,s-1}\left[\left\{ 1-\frac{1-e(H_{0})}{1-\hat{e}(H_{0})}\frac{\bar{\pi}_{s}(0,H_{s-1})}{\hat{\bar{\pi}}_{s}(0,H_{s-1})}\right\} \left\{ \hat{\mu}_{t}^{0}(H_{s})-\hat{\mu}_{t}^{0}(H_{s-1})\right\} ;H_{0}\right]\right)\Bigg],\label{eq:tr_om-psrp2}
\end{align}
which matches the remainder term in Theorem 7.

{}Therefore, $\hat{\tau}_{\text{mr}}-\tau_{t}^{\text{\jtr}}=\left(\pr_{n}-\pr\right)N(V;\theta_{0})+\text{Rem}(\hat{\pr},\pr)+o_{\pr}(n^{-\frac{1}{2}})=\pr_{n}\left\{ \varphi_{t}^{\text{\jtr}}(V_{i};\pr)\right\} +\text{Rem}(\hat{\pr},\pr)+o_{\pr}(n^{-1/2})$.
If $\text{Rem}(\hat{\pr},\pr)=o_{\pr}(n^{-1/2})$, then $\hat{\tau}_{\text{mr}}-\tau_{t}^{\jtr}=n^{-1}\sum_{i=1}^{n}\varphi_{t}^{\text{\jtr}}(V_{i};\pr)+o_{\pr}(n^{-1/2})$.
Apply the central limit theorem and we complete the proof.

\paragraph{Proof of Corollary 2:}

{}For the remainder term, based on the uniform bounded condition,
we proceed to apply Cauchy-Schwarz and Holder's inequality to obtain
the upper bound for each component. For the first term that corresponds
to (\ref{eq:tr_ps-rpom}), we have 
\begin{align*}
\leq & \Big\lVert\frac{e(H_{0})}{\hat{e}(H_{0})}-1\Big\lVert\cdot\Big\lVert\pi_{1}(1,H_{0})g_{t+1}^{1}(H_{0})-\hat{\pi}_{1}(1,H_{0})\hat{g}_{t+1}^{1}(H_{0})\Big\lVert\\
+ & \Big\lVert\frac{e(H_{0})}{\hat{e}(H_{0})}-1\Big\lVert\cdot\left[\sum_{s=1}^{t-1}\left\{ \Big\lVert\pi_{1}(1,H_{0})g_{\hat{\mu},s+1}^{1}(H_{0})-\hat{\pi}_{1}(1,H_{0})\hat{g}_{s+1}^{1}(H_{0})\Big\lVert\right\} \right]\\
\leq & \Big\lVert\pi_{1}(1,H_{0})\Big\lVert_{\infty}\cdot\Big\lVert\frac{e(H_{0})}{\hat{e}(H_{0})}-1\Big\lVert\cdot\left\{ \Big\lVert g_{t+1}^{1}(H_{0})-\hat{g}_{t+1}^{1}(H_{0})\Big\lVert+\Big\lVert\pi_{1}(1,H_{0})-\hat{\pi}_{1}(1,H_{0})\Big\lVert\right\} \\
 & +\Big\lVert\frac{e(H_{0})}{\hat{e}(H_{0})}-1\Big\lVert\cdot\left[\sum_{s=1}^{t-1}\left\{ \Big\lVert g_{\hat{\mu},s+1}^{1}(H_{0})-\hat{g}_{s+1}^{1}(H_{0})\Big\lVert+\Big\lVert\pi_{1}(1,H_{0})-\hat{\pi}_{1}(1,H_{0})\Big\lVert\right\} \right]\\
\leq & \Big\lVert\frac{e(H_{0})}{\hat{e}(H_{0})}-1\Big\lVert\cdot\Big\{\Big\lVert g_{t+1}^{1}(H_{0})-\hat{g}_{t+1}^{1}(H_{0})\Big\lVert+\sum_{s=1}^{t-1}\Big\lVert g_{\hat{\mu},s+1}^{1}(H_{0})-\hat{g}_{s+1}^{1}(H_{0})\Big\lVert\\
 & +(t-1)\Big\lVert\pi_{1}(1,H_{0})-\hat{\pi}_{1}(1,H_{0})\Big\lVert\Big\}\text{ (since \ensuremath{\pi_{1}(1,H_{0})\leq1})}.
\end{align*}
The second inequality holds by Holder's inequality and triangle inequality.
Based on the derived upper bound, the bound of this term is $O_{\pr}(n^{-\min\left(c_{e}+c_{\pi},c_{e}+c_{g}\right)})$.

{}For the second term that corresponds to \eqref{eq:tr_om-psrp1},
we have 
\begin{align*}
\leq & \sum_{s=1}^{t-1}\sum_{l=s+1}^{t}\bigg[\Big\lVert\bar{\pi}_{s}(1,H_{s-1})\Big\lVert\cdot\Big\lVert\frac{1-e(H_{0})}{1-\hat{e}(H_{0})}\{1-\hat{\pi}_{s+1}(1,H_{s})\}\frac{\hat{\delta}(H_{s-1})}{\delta(H_{s-1})}\prod_{k=s+1}^{l}\frac{\pi_{l}(0,H_{l-1})}{\hat{\pi}_{l}(0,H_{l-1})}\\
 & -\{1-\pi_{s+1}(1,H_{s})\}\Big\lVert\cdot\Big\lVert\E\left\{ \hat{\mu}_{t}^{0}(H_{l})\mid H_{l-1},R_{l}=1,A=0\right\} -\hat{\mu}_{t}^{0}(H_{l-1})\Big\lVert\bigg]\\
\leq & \sum_{s=2}^{t}\sum_{l=s+1}^{t}\bigg[\Big\lVert\frac{1-e(H_{0})}{1-\hat{e}(H_{0})}\{1-\hat{\pi}_{s+1}(1,H_{s})\}\frac{\hat{\delta}(H_{s-1})}{\delta(H_{s-1})}\prod_{k=s+1}^{l}\frac{\pi_{l}(0,H_{l-1})}{\hat{\pi}_{l}(0,H_{l-1})}-\{1-\pi_{s+1}(1,H_{s})\}\Big\lVert\\
 & \cdot\Big\lVert\E\left\{ \hat{\mu}_{t}^{0}(H_{l})\mid H_{l-1},R_{l}=1,A=0\right\} -\hat{\mu}_{t}^{0}(H_{l-1})\Big\lVert\bigg]\text{ (since \ensuremath{\bar{\pi}_{s}(1,H_{s-1})\leq1})}\\
\leq & \sum_{s=1}^{t-1}\sum_{l=s+1}^{t}\Big\lVert\E\left\{ \hat{\mu}_{t}^{0}(H_{l})\mid H_{l-1},R_{l}=1,A=0\right\} -\hat{\mu}_{t}^{0}(H_{l-1})\Big\lVert\cdot\\
 & \bigg[\Big\lVert\frac{1-e(H_{0})}{1-\hat{e}(H_{0})}\frac{\hat{\delta}(H_{s-1})}{\delta(H_{s-1})}\Big\lVert_{\infty}\cdot\Big\lVert\{1-\hat{\pi}_{s+1}(1,H_{s})\}\prod_{k=s+1}^{l}\frac{\pi_{l}(0,H_{l-1})}{\hat{\pi}_{l}(0,H_{l-1})}-\{1-\pi_{s+1}(1,H_{s})\}\Big\lVert\\
 & +\Big\lVert1-\pi_{s+1}(1,H_{s})\Big\lVert_{\infty}\cdot\Big\lVert\frac{1-e(H_{0})}{1-\hat{e}(H_{0})}\frac{\hat{\delta}(H_{s-1})}{\delta(H_{s-1})}-1\Big\rVert\bigg]\\
\leq & M\sum_{s=1}^{t-1}\sum_{l=s+1}^{t}\Big\lVert\E\left\{ \hat{\mu}_{t}^{0}(H_{l})\mid H_{l-1},R_{l}=1,A=0\right\} -\hat{\mu}_{t}^{0}(H_{l-1})\Big\lVert\cdot\\
 & \bigg[\Big\lVert\{1-\hat{\pi}_{s+1}(1,H_{s})\}\prod_{k=s+1}^{l}\frac{\pi_{l}(0,H_{l-1})}{\hat{\pi}_{l}(0,H_{l-1})}-\{1-\pi_{s+1}(1,H_{s})\}\Big\lVert+\Big\lVert\frac{1-e(H_{0})}{1-\hat{e}(H_{0})}\frac{\hat{\delta}(H_{s-1})}{\delta(H_{s-1})}-1\Big\rVert\bigg].
\end{align*}
The second and the third inequalities hold by Holder's inequality
and triangle inequality. The term is $o_{\pr}(n^{-\min\left(c_{e}+c_{\mu},c_{\mu}+c_{\pi}\right)})$.

{}For the third term that corresponds to \eqref{eq:tr_rp-psom},
we have 
\begin{align*}
\leq & \Big\lVert\hat{\pi}_{1}(1,H_{0})-\pi_{1}(1,H_{0})\Big\lVert\cdot\Big\lVert\frac{e(H_{0})}{\hat{e}(H_{0})}\hat{\mu}_{t}^{0}(H_{0})-\mu_{t}^{0}(H_{0})\Big\lVert\\
\leq & \Big\lVert\hat{\pi}_{1}(1,H_{0})-\pi_{1}(1,H_{0})\Big\lVert\cdot\bigg\{\Big\lVert\frac{e(H_{0})}{\hat{e}(H_{0})}\Big\lVert_{\infty}\Big\lVert\hat{\mu}_{t}^{0}(H_{0})-\mu_{t}^{0}(H_{0})\Big\lVert\\
 & +\Big\lVert\mu_{t}^{0}(H_{0})\Big\lVert_{\infty}\Big\lVert\frac{e(H_{0})}{\hat{e}(H_{0})}-1\Big\lVert\bigg\}\\
\leq & M\Big\lVert\hat{\pi}_{1}(1,H_{0})-\pi_{1}(1,H_{0})\Big\lVert\cdot\left\{ \Big\lVert\hat{\mu}_{t}^{0}(H_{0})-\mu_{t}^{0}(H_{0})\Big\lVert+\Big\lVert\frac{e(H_{0})}{\hat{e}(H_{0})}-1\Big\lVert\right\} .
\end{align*}
The second inequality holds by Holder's inequality and triangle inequality.
The term is $o_{\pr}(n^{-\min\left(c_{e}+c_{\pi},c_{\mu}+c_{\pi}\right)})$.

{}For the fourth term that corresponds to \eqref{eq:tr_om-psrp2},
we have 
\begin{align*}
\leq & \sum_{s=1}^{t}\Big\lVert1-\frac{1-e(H_{0})}{1-\hat{e}(H_{0})}\frac{\bar{\pi}_{s}(0,H_{s-1})}{\hat{\bar{\pi}}_{s}(0,H_{s-1})}\Big\lVert\cdot\Big\lVert\E\left\{ \hat{\mu}_{t}^{0}(H_{s})\mid H_{s-1},R_{l}=1,A=0\right\} -\hat{\mu}_{t}^{0}(H_{s-1})\Big\lVert.\\
\leq & \sum_{s=1}^{t}\Big\lVert\E\left\{ \hat{\mu}_{t}^{0}(H_{s})\mid H_{s-1},R_{l}=1,A=0\right\} -\hat{\mu}_{t}^{0}(H_{s-1}))\Big\lVert\cdot\bigg\{\Big\lVert1-\frac{1-e(H_{0})}{1-\hat{e}(H_{0})}\Big\rVert\\
 & +\Big\lVert\frac{1-e(H_{0})}{1-\hat{e}(H_{0})}\Big\rVert_{\infty}\cdot\Big\lVert1-\frac{\bar{\pi}_{s}(0,H_{s-1})}{\hat{\bar{\pi}}_{s}(0,H_{s-1})}\Big\rVert\bigg\}\\
\leq & M\sum_{s=1}^{t}\Big\lVert\E\left\{ \hat{\mu}_{t}^{0}(H_{s})\mid H_{s-1},R_{l}=1,A=0\right\} -\hat{\mu}_{t}^{0}(H_{s-1})\Big\lVert\cdot\bigg\{\Big\lVert1-\frac{1-e(H_{0})}{1-\hat{e}(H_{0})}\Big\rVert+\Big\lVert1-\frac{\bar{\pi}_{s}(0,H_{s-1})}{\hat{\bar{\pi}}_{s}(0,H_{s-1})}\Big\rVert\bigg\}.
\end{align*}
The term is $o_{\pr}(n^{-\min\left(c_{e}+c_{\mu},c_{\mu}+c_{\pi}\right)})$.
Therefore, based on Theorem 7 and apply central limit theorem, we
have $\hat{\tau}_{\text{mr}}-\tau_{t}^{\text{\jtr}}=O_{\pr}\left(n^{-\frac{1}{2}}+n^{-c}\right)$,
where $c=\min(c_{e}+c_{\mu},c_{e}+c_{\pi},c_{\mu}+c_{\pi},c_{\pi}+c_{g})$,
which completes the proof.

\section{{Connections to the conventional augmented inverse
propensity weighted estimator}}\label{sec:supp_aipw}

We try to connect the proposed multiply robust estimators
with the augmented inverse propensity weighted (AIPW; \citealp{robins1994estimation})
estimators in the existing missing data literature (e.g., \citealp{robins1995semiparametric,bang2005doubly}).
Under the cross-sectional setting, we use the identification formula
in Theorem 1 as a starting point to construct
the AIPW estimator. Extending to longitudinal settings follows a similar
idea. 

Since the identification formula in Theorem 1
(b) depends on two of the three models, we can apply the standard
AIPW technique to obtain a doubly robust estimator in the AIPW form
as
\[
\hat{\tau}_{\text{ps-rpom}}=\pr_{n}\left[\left\{ \frac{A}{e(X;\hat{\alpha})}-\frac{1-A}{1-e(X;\hat{\alpha})}\right\} R_{1}\left\{ Y_{1}-\mu_{1}^{0}(X;\hat{\beta})\right\} -\left\{ \frac{A}{e(X;\hat{\alpha})}-1\right\} \pi_{1}(1,X;\hat{\gamma})\left\{ \mu_{1}^{1}(X;\hat{\beta})-\mu_{1}^{0}(X;\hat{\beta})\right\} \right].
\]
The following theorem indicates that it is doubly robust in the sense
that it is consistent under $\mathcal{M}_{\text{ps}}\cup\mathcal{M}_{\text{rp+om}}$
when using parametric modeling strategy to estimate the nuisance functions.

\begin{theorem}\label{thm: ps-rpom-para}

{Under Assumptions 1--4,
suppose that there exists $\varepsilon>0,$ such that $\varepsilon<\big\{ e(X;\alpha^{*}),\allowbreak e(X;\hat{\alpha}),\pi_{1}(a,X;\gamma^{*}),\pi_{1}(a,X;\hat{\gamma})\big\}<1-\varepsilon$
for all $X$ and $a$ almost surely, the estimator $\hat{\tau}_{\text{ps-rpom}}$
is doubly robust in the sense that it is consistent for $\tau_{1}^{\jtr}$
under $\mathcal{M}_{\text{ps}}\cup\mathcal{M}_{\text{rp+om}}$. }

\end{theorem}

\begin{proof} {Suppose the model estimators $\hat{\theta}=(\hat{\alpha},\hat{\beta},\hat{\gamma})^{\text{T}}$
converges to $\theta^{*}=(\alpha^{*},\beta^{*},\gamma^{*})^{\text{T}}$
in the sense that $\lVert\hat{\theta}-\theta^{*}\rVert=o_{p}(1)$,
where at least one component of $\hat{\theta}$ needs to converge
to the true value. As the sample size $n\rightarrow\infty$, we would
expect $\hat{\tau}_{\text{ps-rpom}}$ converges to 
\begin{align*}
 & \E\left[\big\{\frac{A}{e(X;\alpha^{*})}-\frac{1-A}{1-e(X;\alpha^{*})}\big\} R_{1}\big\{ Y_{1}-\mu_{1}^{0}(X;\beta^{*})\big\}\right]-\E\left[\frac{A-e(X;\alpha^{*})}{e(X;\alpha^{*})}\pi_{1}(1,X;\gamma^{*})\big\{\mu_{1}^{1}(X;\beta^{*})-\mu_{1}^{0}(X;\beta^{*})\big\}\right]\\
= & \tau_{1}^{\jtr}-\E\left[\pi_{1}(1,X)\left\{ \mu_{1}^{1}(X)-\mu_{1}^{0}(X)\right\} \right]\\
+ & \E\left[\frac{e(X)}{e(X;\alpha^{*})}\pi_{1}(1,X)\left\{ \mu_{1}^{1}(X)-\mu_{1}^{0}(X)+\mu_{1}^{0}(X)-\mu_{1}^{0}(X;\beta^{*})\right\} \right]-\E\left[\frac{1-e(X)}{1-e(X;\alpha^{*})}\pi_{1}(1,X)\left\{ \mu_{1}^{0}(X)-\mu_{1}^{0}(X;\beta^{*})\right\} \right]\\
- & \E\left[\frac{e(X)-e(X;\alpha^{*})}{e(X;\alpha^{*})}\pi_{1}(1,X;\gamma^{*})\big\{\mu_{1}^{1}(X;\beta^{*})-\mu_{1}^{0}(X;\beta^{*})\big\}\right]\\
= & \E\left[\frac{e(X)}{e(X;\alpha^{*})}\pi_{1}(1,X)\left\{ \mu_{1}^{1}(X)-\mu_{1}^{0}(X)\right\} \right]-\E\left[\pi_{1}(1,X)\left\{ \mu_{1}^{1}(X)-\mu_{1}^{0}(X)\right\} \right]\\
+ & \E\left[\left\{ \frac{e(X)}{e(X;\alpha^{*})}-1\right\} \pi_{1}(1,X)\left\{ \mu_{1}^{0}(X)-\mu_{1}^{0}(X;\beta^{*})\right\} \right]-\E\left[\left\{ \frac{1-e(X)}{1-e(X;\alpha^{*})}-1\right\} \pi_{1}(1,X)\left\{ \mu_{1}^{0}(X)-\mu_{1}^{0}(X;\beta^{*})\right\} \right]\\
- & \E\left[\frac{e(X)-e(X;\alpha^{*})}{e(X;\alpha^{*})}\pi_{1}(1,X;\gamma^{*})\big\{\mu_{1}^{1}(X;\beta^{*})-\mu_{1}^{0}(X;\beta^{*})\big\}\right]\\
= & \tau_{1}^{\jtr}+\E\left(\left\{ \frac{e(X)}{e(X;\alpha^{*})}-1\right\} \left[\pi_{1}(1,X)\left\{ \mu_{1}^{1}(X)-\mu_{1}^{0}(X)\right\} -\pi_{1}(1,X;\gamma^{*})\big\{\mu_{1}^{1}(X;\beta^{*})-\mu_{1}^{0}(X;\beta^{*})\big\}+\mu_{1}^{0}(X)-\mu_{1}^{0}(X;\beta^{*})\right]\right)\\
- & \E\left[\left\{ \frac{1-e(X)}{1-e(X;\alpha^{*})}-1\right\} \pi_{1}(1,X)\left\{ \mu_{1}^{0}(X)-\mu_{1}^{0}(X;\beta^{*})\right\} \right].
\end{align*}
From the expression of the asymptotic bias, the estimator $\hat{\tau}_{\text{ps-rpom}}$
is consistent for $\tau_{1}^{\jtr}$ under $\mathcal{M}_{\text{ps}}\cup\mathcal{M}_{\text{rp+om}}$.}\end{proof} 

When using flexible modeling strategies to approximate
the nuisance functions, a standard AIPW estimator has the form
\[
\hat{\tau}_{\text{ps-rpom}}=\pr_{n}\left[\left\{ \frac{A}{\hat{e}(X)}-\frac{1-A}{1-\hat{e}(X)}\right\} R_{1}\left\{ Y_{1}-\hat{\mu}_{1}^{0}(X)\right\} -\left\{ \frac{A}{\hat{e}(X)}-1\right\} \hat{\pi}_{1}(1,X)\left\{ \hat{\mu}_{1}^{1}(X)-\hat{\mu}_{1}^{0}(X)\right\} \right]
\]
and enjoys the property of rate-double robustness, in the sense that
it reaches $n^{1/2}$-consistency if any nuisance functions converge
at a rate no less than $n^{-1/4}$, as illustrated in Corollary \ref{cor: ps-rpom-flex}.

\begin{corollary}\label{cor: ps-rpom-flex} {Under
the assumptions in Corollary 1, $\hat{\tau}_{\text{ps-rpom}}-\tau_{1}^{\text{\jtr}}=O_{\mathbb{P}}\left(n^{-1/2}+n^{-c}\right)$,
where $c=\min(c_{e}+c_{\mu},c_{e}+c_{\pi})$.}

\end{corollary} 

\begin{proof} {We again follow the proof in \citet{kennedy2016semiparametric}.
To simplify the notations, denote $\pr\left\{ N(V;\theta_{0})\right\} =\tau_{1}^{\text{\jtr}}$,
where 
\begin{align*}
N(V;\theta_{0})= & \left\{ \frac{A}{e(X)}-\frac{1-A}{1-e(X)}\right\} R_{1}\left\{ Y_{1}-\mu_{1}^{0}(X)\right\} -\frac{A-e(X)}{e(X)}\pi_{1}(1,X)\left\{ \mu_{1}^{1}(X)-\mu_{1}^{0}(X)\right\} .
\end{align*}
Then $\pr\left\{ N(V;\theta^{*})\right\} =\pr\left\{ N(V;\theta_{0})\right\} =\tau_{1}^{\text{\jtr}}$.
Consider the decomposition 
\[
\hat{\tau}_{\text{ps-rpom}}-\tau_{1}^{\text{\jtr}}=\left(\pr_{n}-\pr\right)N(V;\hat{\theta})-\pr\left\{ N(V;\hat{\theta})-N(V;\theta^{*})\right\} .
\]
Using empirical process theory, if the nuisance functions take values
in Donsker classes, and satisfy the positivity assumption, i.e., there
exists $\varepsilon>0$, such that $\varepsilon<\{e(X),\pi_{1}(a,X)\}<1-\varepsilon$
for all $X$, then $N(V;\hat{\theta})$ takes values in Donsker classes,
and the first term can be written as 
\[
\left(\pr_{n}-\pr\right)N(V;\hat{\theta})=\left(\pr_{n}-\pr\right)N(V;\theta_{0})+o_{\pr}(n^{-\frac{1}{2}}).
\]
}

{For the second term $\pr\left\{ N(V;\hat{\theta})-N(V;\theta^{*})\right\} $,
by computing the expectations, we have
\begin{align*}
\pr\left\{ N(V;\hat{\theta})-N(V;\theta^{*})\right\}  & =\pr\left(\left\{ \frac{e(X)}{\hat{e}(X)}-1\right\} \left[\pi_{1}(1,X)\left\{ \mu_{1}^{1}(X)-\mu_{1}^{0}(X)\right\} -\hat{\pi}_{1}(1,X)\big\{\hat{\mu}_{1}^{1}(X)-\hat{\mu}_{1}^{0}(X)\big\}\right]\right)\\
 & +\pr\left[\left\{ \frac{e(X)}{\hat{e}(X)}-1\right\} \left\{ \mu_{1}^{0}(X)-\hat{\mu}_{1}^{0}(X)\right\} \right]-\pr\left[\left\{ \frac{1-e(X)}{1-\hat{e}(X)}-1\right\} \pi_{1}(1,X)\left\{ \mu_{1}^{0}(X)-\hat{\mu}_{1}^{0}(X)\right\} \right].
\end{align*}
Under the positivity assumptions, we apply Cauchy-Schwarz inequality
($\pr(fg)\leq\lVert f\rVert\lVert g\rVert$) and obtain a upper bound
for the second term as
\begin{align*}
\pr\left\{ N(V;\hat{\theta})-N(V;\theta^{*})\right\}  & \leq\Big\lVert\frac{e(X)}{\hat{e}(X)}-1\Big\rVert\cdot\Big\lVert\pi_{1}(1,X)\mu_{1}^{1}(X)-\hat{\pi}_{1}(1,X)\hat{\mu}_{1}^{1}(X)\Big\rVert\\
 & +\Big\lVert\frac{e(X)}{\hat{e}(X)}-1\Big\rVert\cdot\Big\lVert\pi_{1}(1,X)\mu_{1}^{0}(X)-\hat{\pi}_{1}(1,X)\hat{\mu}_{1}^{0}(X)\Big\rVert\\
 & +\Big\lVert\frac{e(X)}{\hat{e}(X)}-1\Big\rVert\cdot\Big\lVert\mu_{1}^{0}(X)-\hat{\mu}_{1}^{0}(X)\Big\rVert\\
 & +\Big\lVert\frac{1-e(X)}{1-\hat{e}(X)}-1\Big\rVert\cdot\Big\lVert\pi_{1}(1,X)\left\{ \mu_{1}^{0}(X)-\hat{\mu}_{1}^{0}(X)\right\} \Big\rVert\\
\text{} & \leq\Big\lVert\left\{ \frac{e(X)}{\hat{e}(X)}-1\right\} \left\{ \mu_{1}^{1}(X)-\hat{\mu}_{1}^{1}(X)\right\} \Big\rVert_{1}\cdot\Big\lVert\pi_{1}(1,X)\Big\rVert_{\infty}\\
 & +\Big\lVert\left\{ \frac{e(X)}{\hat{e}(X)}-1\right\} \left\{ \pi_{1}(1,X)-\hat{\pi}_{1}(1,X)\right\} \Big\rVert_{1}\cdot\Big\lVert\hat{\mu}_{1}^{1}(X)\Big\rVert_{\infty}\\
 & +\Big\lVert\left\{ \frac{e(X)}{\hat{e}(X)}-1\right\} \left\{ \mu_{1}^{0}(X)-\hat{\mu}_{1}^{0}(X)\right\} \Big\rVert_{1}\cdot\Big\lVert\pi_{1}(1,X)\Big\rVert_{\infty}\\
 & +\Big\lVert\left\{ \frac{e(X)}{\hat{e}(X)}-1\right\} \left\{ \pi_{1}(1,X)-\hat{\pi}_{1}(1,X)\right\} \Big\rVert_{1}\cdot\Big\lVert\hat{\mu}_{1}^{0}(X)\Big\rVert_{\infty}\\
 & +\Big\lVert\frac{e(X)}{\hat{e}(X)}-1\Big\rVert\cdot\Big\lVert\mu_{1}^{0}(X)-\hat{\mu}_{1}^{0}(X)\Big\rVert\\
 & +\Big\lVert\left\{ 1-\frac{1-e(X)}{1-\hat{e}(X)}\right\} \left\{ \mu_{1}^{0}(X)-\hat{\mu}_{1}^{0}(X)\right\} \Big\rVert_{1}\cdot\Big\lVert\pi_{1}(1,X)\Big\rVert_{\infty}\\
 & \leq M\Big\lVert\frac{e(X)}{\hat{e}(X)}-1\Big\rVert\cdot\left\{ \Big\lVert\mu_{1}^{1}(X)-\hat{\mu}_{1}^{1}(X)\Big\rVert+\Big\lVert\pi_{1}(1,X)-\hat{\pi}_{1}(1,X)\Big\lVert+\Big\lVert\mu_{1}^{0}(X)-\hat{\mu}_{1}^{0}(X)\Big\rVert\right\} \\
 & +M\Big\lVert\mu_{1}^{0}(X)-\hat{\mu}_{1}^{0}(X)\Big\rVert\cdot\Big\lVert1-\frac{1-e(X)}{1-\hat{e}(X)}\Big\rVert.
\end{align*}
The second inequality holds by the triangle inequality and Holder's
inequality, and the last inequality holds by Cauchy-Schwarz. We have
$\hat{\tau}_{\text{ps-rpom}}-\tau_{1}^{\text{\jtr}}=O_{\mathbb{P}}\left(n^{-1/2}+n^{-c}\right)$,
where $c=\min(c_{e}+c_{\mu},c_{e}+c_{\pi})$ }\end{proof}

The triply robust estimator $\hat{\tau}_{\text{tr}}$
consists of all the components in the AIPW estimator $\hat{\tau}_{\text{ps-rpom}}$,
while at the same time including extra augmented terms to guarantee
triple robustness in the sense that it achieves $n^{1/2}$-consistency
if any two of the three nuisance models are correct when using the
parametric modeling strategy or if the nuisance functions converge
at a rate no less than $n^{-1/4}$ when using the flexible modeling
strategy. Those additional augmented terms in the triply robust estimator
constitute one of the major contributions of the paper.

\section{Sensitivity analysis on the partial ignorability
of missingness assumption}\label{sec:supp_sen}
In the main text, we impose the partial ignorability
of missingness assumption on the missing components in the control
group for the treatment effect identification under J2R. While it
may not be realistic in practice, sensitivity analyses can be conducted
to assess the robustness of the ATE estimation against this assumption.
In this section, we provide a way to conduct the sensitivity analysis
against Assumption 3 under the PMM framework in
cross-sectional studies. Extending to longitudinal studies follows
the same logic. 

Using the idea of delta-adjustment \citep{mallinckrodt2016analyzing},
we modify the missingness ignorability assumption (Assumption 3)
by introducing a sensitivity parameter $\delta$ in the outcome mean
in the control group as Assumption \ref{assump:sen2-1time}. In this
way, the discrepancy in the outcome mean among the observed and missing
individuals indicates an MNAR pattern in the control group due to
the dependence between the response status and the outcome. With the
lack of MAR in the control group, the outcome mean $\E\{Y_{1}(0)\mid X\}$
in Assumption 4 cannot be identified solely based
on the observed individuals. Therefore, we replace it with $\E\{Y_{1}(0,1)\mid X\}$
by using the non-dropouts in the control group to characterize the
outcome mean of dropouts in the treated group and adjust the original
Assumptions 3 and 4 as follows. 

\begin{assumptionp}{3$'$}[Delta-adjustment in
the control group]\label{assump:sen2-1time}

{$\E\big\{ Y_{1}(0,0)\mid X\big\}=\E\{Y_{1}(0,1)\mid X\}+\delta$. }

\end{assumptionp}

Assumption \ref{assump:sen2-1time} depicts an MNAR
pattern for the missing components in the control group. The sensitivity
parameter $\delta$ controls the degree of the deviation from the
observed outcome mean, thus indicating a difference in outcome distributions
between the observed and missing individuals when $\delta\neq0$.
Compared with Assumption 3, where we directly
assume the conditional independence between the response status and
the outcome to characterize the MAR assumption under general CBI models,
Assumption \ref{assump:sen2-1time} only specifies the outcome mean
$\E\big\{ Y_{1}(0,0)\mid X\big\}$ that is needed for the ATE identification.
If other types of treatment effect estimands are considered, e.g.,
the risk difference or the quantile treatment effect, one can alternatively
use delta adjustment on the observed distribution $f\{Y_{1}(0,1)\mid X\}$
to describe the unobserved distribution $f\{Y_{1}(0,0)\mid X\}$ and
conduct sensitivity analyses.

\begin{assumptionp}{4$'$}[J2R for the outcome
mean in the treated group]\label{assump:sen4-1time}

{$\E\big\{ Y_{1}(1,0)\mid X,R_{1}(1)=0\big\}=\E\{Y_{1}(0,1)\mid X\}$. }

\end{assumptionp}

We replace the outcome mean $\E\{Y_{1}(0)\mid X\}$
in the original Assumption 4 with $\E\{Y_{1}(0,1)\mid X\}$
in Assumption \ref{assump:sen4-1time} for the sensitivity analysis,
since now the dropouts in the treated group are expected to share
the same outcome mean as the observed subjects in the control group
given the same history. 

{Note that when $\delta=0$, Assumptions \ref{assump:sen2-1time}
and \ref{assump:sen4-1time} do not correspond to Assumptions 3
and 4, as Assumption 3 imposes
a distributional assumption on the outcomes in the control group instead
of an outcome mean profile. Assumption 3 is created
to resemble the conventional MAR assumption, yet a relaxed version
with only the specification of the outcome mean can also result in
the same ATE identification and estimation. Under the sensitivity
analysis, we still use the ITT estimand and define the ATE as $\tau_{1}^{\jtr'}=\E[Y_{1}\{1,R_{1}(1)\}]-\E[Y_{1}\{0,R_{1}(0)\}]=\E\{Y_{1}(1)-Y_{1}(0)\}$.
Similar to Theorem 1 in the main paper, three
identification formulas of $\tau_{1}^{\jtr'}$ can be accomplished
in the following theorem. }

\begin{theorem}\label{thm:iden_1time-sen} 

Under Assumptions 1, 2,
\ref{assump:sen2-1time}, and \ref{assump:sen4-1time}, assume there
exists $\varepsilon>0,$ such that $\varepsilon<\big\{ e(X),\allowbreak\pi_{1}(a,X)\big\}<1-\varepsilon$
for all $X$ and $a$, the following identification formulas hold. 
\begin{enumerate}
\item {Based on the response probability and outcome mean,
\[
\tau_{1}^{\text{\jtr'}}=\mathbb{E}\left[\pi_{1}(1,X)\left\{ \mu_{1}^{1}(X)-\mu_{1}^{0}(X)\right\} -\left\{ 1-\pi_{1}(0,X)\right\} \delta\right].
\]
}
\item {Based on the propensity score and outcome mean, 
\[
\tau_{1}^{\text{\jtr'}}=\mathbb{E}\left[\frac{2A-1}{e(X)^{A}\{1-e(X)\}^{1-A}}\big\{ R_{1}Y_{1}+(1-R_{1})\mu_{1}^{0}(X)\big\}-\frac{1-A}{1-e(X)}(1-R_{1})\delta\right].
\]
}
\item {Based on the propensity score and response probability,
\[
\tau_{1}^{\jtr'}=\mathbb{E}\left[\frac{A}{e(X)}R_{1}Y_{1}-\frac{(1-A)R_{1}}{\{1-e(X)\}\pi_{1}(0,X)}\pi_{1}(1,X)Y_{1}-\left\{ 1-\pi_{1}(0,X)\right\} \delta\right].
\]
}
\end{enumerate}
\end{theorem}

\begin{proof} {We follow the same proof in \ref{subsec:supp_iden_1time}
to get the identification formulas for the ATE. Compared with Theorem
1, an additional term that involves the sensitivity
parameter $\delta$ is contained in each identification formula. The
identification of $\tau_{1,1}=\E[Y_{1}\{1,R_{1}(1)\}]$ remains unchanged
since the specification of the outcome mean $\E\big\{ Y_{1}(1,0)\mid X,R_{1}(1)=0\big\}$
stays the same by Assumption \ref{assump:sen4-1time}. Therefore,
we proceed to identify $\E[Y_{1}\{0,R_{1}(0)\}]$. Following the same
step of identifying $\E[Y_{1}\{0,R_{1}(0)\}]$ in \ref{subsec:supp_iden_1time},
we have
\begin{align*}
\tau_{0,1} & =\E\left[R_{1}(0)Y_{1}^{\text{}}(0,1)+\{1-R_{1}(0)\}Y_{1}(0,0)\right]\\
 & =\E\left[\E\left\{ R_{1}(0)\mid X\right\} \E\left\{ Y_{1}(0,1)\mid X,R_{1}(0)=1\right\} +\E\left\{ 1-R_{1}(0)\mid X\right\} \E\left\{ Y_{1}(0,0)\mid X,R_{1}(0)=0\right\} \right]\\
 & =\E\Big[\E\left(R_{1}\mid X,A=0\right)\E\left\{ Y_{1}(0,1)\mid X,R_{1}(0)=1,A=0\right\} \\
 & \qquad+\E\left(1-R_{1}\mid X,A=0\right)\E\left\{ Y_{1}(0,0)\mid X,R_{1}(0)=0,A=0\right\} \Big]\text{(By A1, A3)}\\
 & =\E\left[\pi_{1}(0,X)\E\left(Y_{1}\mid A=0,R_{1}=1,X\right)+\left\{ 1-\pi_{1}(0,X)\right\} \left\{ \E\left(Y_{1}\mid A=0,R_{1}=1,X\right)+\delta\right\} \right]\text{(By A2', A3, A4') }\\
 & =\E\left[\pi_{1}(0,X)\mu_{1}^{0}(X)+\left\{ 1-\pi_{1}(0,X)\right\} \left\{ \mu_{1}^{0}(X)+\delta\right\} \right]\\
 & =\E\left[\mu_{1}^{0}(X)+\left\{ 1-\pi_{1}(0,X)\right\} \delta\right].
\end{align*}
Therefore, the identification of $\tau$ corresponds to
\begin{align*}
\tau_{1}^{\text{\jtr'}} & =\tau_{1,1}-\tau_{0,1}=\E\left[\pi_{1}(1,X)\mu_{1}^{1}(X)+\{1-\pi_{1}(1,X)\}\mu_{1}^{0}(X)-\mu_{1}^{0}(X)-\left\{ 1-\pi_{1}(0,X)\right\} \delta\right]\\
 & =\E\left[\pi_{1}(1,X)\left\{ \mu_{1}^{1}(X)-\mu_{1}^{0}(X)\right\} -\left\{ 1-\pi_{1}(0,X)\right\} \delta\right],
\end{align*}
which matches the identification formula in Theorem \ref{thm:iden_1time-sen}
(a). }

{We then need to show $\E\left[\left\{ 1-\pi_{1}(0,X)\right\} \delta\right]=\E\left[(1-A)(1-R_{1})\delta/\{1-e(X)\}\right]$.
Note that
\begin{align*}
\E\left[\frac{1-A}{1-e(X)}(1-R_{1})\delta\right]= & \E\left[\E\left\{ \frac{1-A}{1-e(X)}(1-R_{1})\mid X,A\right\} \delta\right]\\
= & \E\left[\frac{1-A}{1-e(X)}\left\{ 1-\pi_{1}(0,X)\right\} \delta\right]\\
= & \E\left[\frac{1-\E(A\mid X)}{1-e(X)}\left\{ 1-\pi_{1}(0,X)\right\} \delta\right]\\
= & \E\left[\left\{ 1-\pi_{1}(0,X)\right\} \delta\right],
\end{align*}
which complete the proof. }\end{proof}

{When $\delta=0$, Theorem \ref{thm:iden_1time-sen}
degenerates to Theorem 1. One can plug in the
nuisance function estimators to get the conventional and stabilized
versions of the ATE estimators. Similarly, we derive the EIF under
the sensitivity analysis to motivate the EIF-based estimator as follows. }

\begin{theorem}\label{thm: eif_1time-sen}

{Under Assumptions 1, 2,
\ref{assump:sen2-1time}, and \ref{assump:sen4-1time}, suppose that
there exists $\varepsilon>0,$ such that $\varepsilon<\big\{ e(X),\pi_{1}(a,X)\big\}<1-\varepsilon$
for all $X$ and $a$, the EIF for $\tau_{1}^{\jtr'}$ is 
\begin{align*}
\varphi_{1}^{\jtr'}(V;\mathbb{P})= & \left\{ \frac{A}{e(X)}-\frac{1-A}{1-e(X)}\frac{\pi_{1}(1,X)}{\pi_{1}(0,X)}\right\} R_{1}\left\{ Y_{1}-\mu_{1}^{0}(X)\right\} -\frac{A-e(X)}{e(X)}\pi_{1}(1,X)\left\{ \mu_{1}^{1}(X)-\mu_{1}^{0}(X)\right\} \\
 & +\left[\{1-\pi_{1}(0,X)\}-\frac{1-A}{1-e(X)}\left\{ R_{1}-\pi_{1}(0,X)\right\} \right]\delta-\tau_{1}^{\jtr'}.
\end{align*}
}\end{theorem} 

{Based on the fact that the mean of the EIF is zero,
we can obtain another identification formula for the ATE under the
sensitivity analysis, which motivates the EIF-based estimator $\hat{\tau}_{\text{tr}}'$
as}{\small{}
\begingroup\makeatletter\def\f@size{9}\check@mathfonts
\begin{eqnarray*}
 & \hat{\tau}_{\text{tr}}'= & \mathbb{P}_{n}\Bigg[\left\{ \frac{A}{e(X;\hat{\alpha})}-\frac{1-A}{1-e(X;\hat{\alpha})}\frac{\pi_{1}(1,X;\hat{\gamma})}{\pi_{1}(0,X;\hat{\gamma})}\right\} R_{1}\left\{ Y_{1}-\mu_{1}^{0}(X;\hat{\beta})\right\} -\frac{A-e(X;\hat{\alpha})}{e(X;\hat{\alpha})}\pi_{1}(1,X;\hat{\gamma})\left\{ \mu_{1}^{1}(X;\hat{\beta})-\mu_{1}^{0}(X;\hat{\beta})\right\} \\
 &  & +\{1-\pi_{1}(0,X;\hat{\gamma})\}\frac{1-A}{1-e(X;\hat{\alpha})}\left\{ R_{1}-\pi_{1}(0,X;\hat{\gamma})\right\} \delta\Bigg].
\end{eqnarray*}
\endgroup
}{One can also apply normalization or calibration
to obtain more stabilized estimators. }

{Next, we investigate the asymptotic properties of
the EIF-based estimator $\hat{\tau}_{\text{tr}}'$. Theorem \ref{thm: tr_1time_cons-sen}
verifies the robustness when using parametric models to approximate
the nuisance functions.}

\begin{theorem}\label{thm: tr_1time_cons-sen}

{Under Assumptions 1, 2,
\ref{assump:sen2-1time}, and \ref{assump:sen4-1time}, suppose that
there exists $\varepsilon>0,$ such that{} $\varepsilon<\big\{ e(X;\alpha^{*}),\allowbreak e(X;\hat{\alpha}),\pi_{1}(a,X;\gamma^{*}),\pi_{1}(a,X;\hat{\gamma})\big\}<1-\varepsilon$
for all $X$ and $a$ almost surely, the estimator $\hat{\tau}_{\text{tr}}$
is triply robust in the sense that it is consistent for $\tau_{1}^{\jtr'}$
under $\mathcal{M}_{\text{rp+om}}\cup\mathcal{M}_{\text{ps+om}}\cup\mathcal{M}_{\text{ps+rp}}$.
Moreover, $\hat{\tau}_{\text{tr}}'$ achieves the semiparametric efficiency
bound under $\mathcal{M}_{\text{ps+rp+om}}$.}

\end{theorem}

\begin{proof} {Suppose the model estimators $\hat{\theta}=(\hat{\alpha},\hat{\beta},\hat{\gamma})^{\text{T}}$
converges to $\theta^{*}=(\alpha^{*},\beta^{*},\gamma^{*})^{\text{T}}$
in the sense that $\lVert\hat{\theta}-\theta^{*}\rVert=o_{p}(1)$,
where at least one component of $\hat{\theta}$ needs to converge
to the true value. As the sample size $n\rightarrow\infty$, we would
expect $\hat{\tau}_{\text{tr}}'$ converges to 
\begin{align}
 & \E\left[\big\{\frac{A}{e(X;\alpha^{*})}-\frac{1-A}{1-e(X;\alpha^{*})}\frac{\pi_{1}(1,X;\gamma^{*})}{\pi_{1}(0,X;\gamma^{*})}\big\} R_{1}\big\{ Y_{1}-\mu_{1}^{0}(X;\beta^{*})\big\}\right]\nonumber \\
- & \E\left[\frac{A-e(X;\alpha^{*})}{e(X;\alpha^{*})}\pi_{1}(1,X;\gamma^{*})\big\{\mu_{1}^{1}(X;\beta^{*})-\mu_{1}^{0}(X;\beta^{*})\big\}\right]\nonumber \\
+ & \E\left[\{1-\pi_{1}(0,X;\gamma^{*})\}-\frac{1-A}{1-e(X;\alpha^{*})}\left\{ R_{1}-\pi_{1}(0,X;\gamma^{*})\right\} \right]\delta\label{eq:tr-1time-sen}
\end{align}
}

{The first two terms are the same as formulas \eqref{eq:trpart1_1time}
and \eqref{eq:trpart2_1time} in \ref{subsec:supp_tr_1time}. Therefore,
we focus on formula \eqref{eq:tr-1time-sen} and rearrange the term
as
\begin{align*}
 & \E\left[\{1-\pi_{1}(0,X;\gamma^{*})\}-\frac{\E(1-A\mid X)}{1-e(X;\alpha^{*})}\left\{ \E(R_{1}\mid X,A=1)-\pi_{1}(0,X;\gamma^{*})\right\} \right]\delta\\
= & \E\left[\{1-\pi_{1}(0,X;\gamma^{*})\}-\frac{1-e(X)}{1-e(X;\alpha^{*})}\left\{ \pi_{1}(0,X)-\pi_{1}(0,X;\gamma^{*})\right\} \right]\delta.
\end{align*}
Combining the three parts together, \eqref{eq:trpart1_1time} + \eqref{eq:trpart2_1time}
+ \eqref{eq:tr-1time-sen}
\begin{align*}
= & \tau_{1}^{\text{\jtr}'}+\E\left[\left\{ \frac{e(X)}{e(X;\alpha^{*})}-1\right\} \left\{ \pi_{1}(1,X)\mu_{1}^{1}(X)-\pi_{1}(1,X;\gamma^{*})\mu_{1}^{1}(X;\beta^{*})\right\} \right]\\
 & +\E\left[\left\{ 1-\frac{1-e(X)}{1-e(X;\alpha^{*})}\frac{\pi_{1}(0,X)\pi_{1}(1,X;\gamma^{*})}{\pi_{1}(0,X;\gamma^{*})}\right\} \pi_{1}(1,X;\gamma^{*})\left\{ \mu_{1}^{0}(X)-\mu_{1}^{0}(X;\beta^{*})\right\} \right]\\
 & +\E\left[\left\{ \pi_{1}(1,X)-\pi_{1}(1,X;\gamma^{*})\right\} \left\{ \mu_{1}^{0}(X)-\frac{e(X)}{e(X;\alpha^{*})}\mu_{1}^{0}(X;\beta^{*})\right\} \right]\\
 & +\E\left[\left\{ 1-\frac{1-e(X)}{1-e(X;\alpha^{*})}\right\} \left\{ \pi_{1}(0,X)-\pi_{1}(0,X;\gamma^{*})\right\} \right]\delta.
\end{align*}
}

{{}Therefore, the bias of $\hat{\tau}_{\text{tr}}'$
converges to 
\begin{align}
 & \E\left[\left\{ \frac{e(X)}{e(X;\alpha^{*})}-1\right\} \left\{ \pi_{1}(1,X)\mu_{1}^{1}(X)-\pi_{1}(1,X;\gamma^{*})\mu_{1}^{1}(X;\beta^{*})\right\} \right]\label{eq:trbias1_1time-1}\\
+ & \E\left[\left\{ 1-\frac{1-e(X)}{1-e(X;\alpha^{*})}\frac{\pi_{1}(0,X)\pi_{1}(1,X;\gamma^{*})}{\pi_{1}(0,X;\gamma^{*})}\right\} \pi_{1}(1,X;\gamma^{*})\left\{ \mu_{1}^{0}(X)-\mu_{1}^{0}(X;\beta^{*})\right\} \right]\label{eq:trbias2_1time-1}\\
+ & \E\left[\left\{ \pi_{1}(1,X)-\pi_{1}(1,X;\gamma^{*})\right\} \left\{ \mu_{1}^{0}(X)-\frac{e(X)}{e(X;\alpha^{*})}\mu_{1}^{0}(X;\beta^{*})\right\} \right]\label{eq:trbias3_1time-1}\\
+ & \E\left[\left\{ 1-\frac{1-e(X)}{1-e(X;\alpha^{*})}\right\} \left\{ \pi_{1}(0,X)-\pi_{1}(0,X;\gamma^{*})\right\} \right]\delta\label{eq:trbias4_1time-sen}
\end{align}
}

{{}Note that \eqref{eq:trbias1_1time-1} $=0$ under
$\mathcal{M}_{\text{rp+om}}\cup\mathcal{M}_{\text{ps}}$, \eqref{eq:trbias2_1time-1}
$=0$ under $\mathcal{M}_{\text{ps+rp}}\cup\mathcal{M}_{\text{om}}$,
\eqref{eq:trbias3_1time-1} $=0$ under $\mathcal{M}_{\text{ps+om}}\cup\mathcal{M}_{\text{rp}}$,
and \eqref{eq:trbias4_1time-sen} $=0$ under $\mathcal{M}_{\text{ps}}\cup\mathcal{M}_{\text{rp}}$.
Thus, $\hat{\tau}_{\text{tr}}'$ is consistent for $\tau_{1}^{\text{\jtr'}}$
under $\mathcal{M}_{\text{rp+om}}\cup\mathcal{M}_{\text{ps+om}}\cup\mathcal{M}_{\text{ps+rp}}$.
The triple robustness holds.}\end{proof}

{When using flexible models to approximate nuisance
functions, Theorem \ref{thm: tr_1time_rate-sen} uncovers the asymptotic
property of the EIF-based estimator and invokes the triple robustness
in terms of rate convergence in Corollary \ref{cor: tr_1time_rate-sen}.}

\begin{theorem}\label{thm: tr_1time_rate-sen}

{Under Assumptions 1, 2,
\ref{assump:sen2-1time}, and \ref{assump:sen4-1time}, suppose that
there exists $\varepsilon>0,$ such that $\varepsilon<\big\{ e(X),\hat{e}(X),\allowbreak\pi_{1}(a,X),\hat{\pi}_{1}(a,X)\big\}<1-\varepsilon$
for all $X$ and $a$ almost surely, and the nuisance functions and
their estimators take value in Donsker classes. Assume $\lVert\varphi_{1}^{\jtr'}(V;\hat{\pr})-\varphi_{1}^{\jtr'}(V;\mathbb{P})\rVert=o_{\mathbb{P}}(1)$.
Then, $\hat{\tau}_{\text{tr}}'=\tau_{1}^{\jtr'}+n^{-1}\sum_{i=1}^{n}\varphi_{1}^{\jtr'}(V;\mathbb{P})+\text{Rem}(\hat{\mathbb{P}},\mathbb{P})+o_{\mathbb{P}}\left(n^{-1/2}\right)$,
where
\begin{align*}
\text{Rem}(\hat{\mathbb{P}},\mathbb{P}) & =\E\bigg[\left\{ \frac{e(X)}{\hat{e}(X)}-1\right\} \left\{ \pi_{1}(1,X)\mu_{1}^{1}(X)-\hat{\pi}_{1}(1,X)\hat{\mu}_{1}^{1}(X)\right\} +\left\{ 1-\frac{1-e(X)}{1-\hat{e}(X)}\frac{\pi_{1}(0,X)}{\hat{\pi}_{1}(0,X)}\right\} \hat{\pi}_{1}(1,X)\big\{\mu_{1}^{0}(X)\\
 & \quad-\hat{\mu}_{1}^{0}(X)\big\}+\left\{ \pi_{1}(1,X)-\hat{\pi}_{1}(1,X)\right\} \left\{ \mu_{1}^{0}(X)-\frac{e(X)}{\hat{e}(X)}\hat{\mu}_{1}^{0}(X)\right\} \\
 & \quad+\left\{ 1-\frac{1-e(X)}{1-\hat{e}(X)}\right\} \left\{ \pi_{1}(0,X)-\hat{\pi}_{1}(0,X)\right\} \delta\bigg].
\end{align*}
}

{If $\text{Rem}(\hat{\mathbb{P}},\mathbb{P})=o_{\mathbb{P}}(n^{-1/2})$,
then $n^{1/2}\left(\hat{\tau}_{\text{tr}}'-\tau_{1}^{\jtr'}\right)\xrightarrow{d}\mathcal{N}\left(0,\mathbb{V}\left\{ \varphi_{1}^{\jtr'}(V;\mathbb{P})\right\} \right)$,
where the asymptotic variance of $\hat{\tau}_{\text{tr}}'$ reaches
the semiparametric efficiency bound. }

\end{theorem}

\begin{corollary}\label{cor: tr_1time_rate-sen} {Under
the assumptions in Theorem \ref{thm: tr_1time_rate-sen}, suppose
$\lVert\varphi_{1}^{\text{\jtr}'}(V;\hat{\mathbb{P}})-\varphi_{1}^{\text{\jtr'}}(V;\mathbb{P})\rVert=o_{\mathbb{P}}(1)$,
and further suppose that there exists $0<M<\infty$, such that $\pr\bigg(\max\Big\{\big|\hat{\mu}_{1}^{0}(X)\big|,\big|\hat{\mu}_{1}^{1}(X)\big|,\allowbreak\Big|\{1-e(X)\}/\{1-\hat{e}(X)\}\Big|,\delta\Big\}\leq M\bigg)=1$,
then $\hat{\tau}_{\text{tr}}^{'}-\tau_{1}^{\text{\jtr'}}=O_{\mathbb{P}}\left(n^{-1/2}+n^{-c}\right)$,
where $c=\min(c_{e}+c_{\mu},c_{e}+c_{\pi},c_{\mu}+c_{\pi})$.}

\end{corollary} 

\begin{proof} {We again follow the proof in \citet{kennedy2016semiparametric}.
To simplify the notations, denote $\pr\left\{ N(V;\theta_{0})\right\} =\tau_{1}^{\text{\jtr'}}$,
where 
\begin{align*}
N(V;\theta_{0})= & \left\{ \frac{A}{e(X)}-\frac{1-A}{1-e(X)}\frac{\pi_{1}(1,X)}{\pi_{1}(0,X)}\right\} R_{1}\left\{ Y_{1}-\mu_{1}^{0}(X)\right\} -\frac{A-e(X)}{e(X)}\pi_{1}(1,X)\left\{ \mu_{1}^{1}(X)-\mu_{1}^{0}(X)\right\} \\
 & +\left[\{1-\pi_{1}(0,X)\}-\frac{1-A}{1-e(X)}\left\{ R_{1}-\pi_{1}(0,X)\right\} \right]\delta.
\end{align*}
Then $\pr\left\{ N(V;\theta^{*})\right\} =\pr\left\{ N(V;\theta_{0})\right\} =\tau_{1}^{\text{\jtr}'}$.
Consider the decomposition 
\begin{equation}
\hat{\tau}_{\text{tr}}^{'}-\tau_{1}^{\text{\jtr'}}=\left(\pr_{n}-\pr\right)N(V;\hat{\theta})-\pr\left\{ N(V;\hat{\theta})-N(V;\theta^{*})\right\} .\label{eq:tr_decomp_t1-1}
\end{equation}
Using empirical process theory, if the nuisance functions take values
in Donsker classes, and satisfy the positivity assumption, i.e., there
exists $\varepsilon>0$, such that $\varepsilon<e(X)<1-\varepsilon$
and $\pi_{1}(a,X)>\varepsilon$ for all $X$, then $N(V;\hat{\theta})$
takes values in Donsker classes, and the first term can be written
as 
\[
\left(\pr_{n}-\pr\right)N(V;\hat{\theta})=\left(\pr_{n}-\pr\right)N(V;\theta_{0})+o_{\pr}(n^{-\frac{1}{2}}).
\]
}

{{}For the second term $\pr\left\{ N(V;\hat{\theta})-N(V;\theta^{*})\right\} $,
by computing the expectations, we have
\begin{align*}
\pr\left\{ N(V;\hat{\theta})-N(V;\theta^{*})\right\}  & =\pr\left[\left\{ \frac{e(X)}{\hat{e}(X)}-1\right\} \left\{ \pi_{1}(1,X)\mu_{1}^{1}(X)-\hat{\pi}_{1}(1,X)\hat{\mu}_{1}^{1}(X)\right\} \right]\\
 & +\pr\left[\left\{ 1-\frac{1-e(X)}{1-\hat{e}(X)}\frac{\pi_{1}(0,X)}{\hat{\pi}_{1}(0,X)}\right\} \hat{\pi}_{1}(1,X)\left\{ \mu_{1}^{0}(X)-\hat{\mu}_{1}^{0}(X)\right\} \right]\\
 & +\pr\left[\left\{ \pi_{1}(1,X)-\hat{\pi}_{1}(1,X)\right\} \left\{ \mu_{1}^{0}(X)-\frac{e(X)}{\hat{e}(X)}\hat{\mu}_{1}^{0}(X)\right\} \right]\\
 & +\pr\left[\left\{ 1-\frac{1-e(X)}{1-\hat{e}(X)}\right\} \left\{ \pi_{1}(0,X)-\hat{\pi}_{1}(0,X)\right\} \right]\delta
\end{align*}
}

{{}Under the positivity assumptions, we apply Cauchy-Schwarz
inequality ($\pr(fg)\leq\lVert f\rVert\lVert g\rVert$) and obtain
a upper bound for the second term as
\begin{align*}
\pr\left\{ N(V;\hat{\theta})-N(V;\theta^{*})\right\}  & \leq\Big\lVert\frac{e(X)}{\hat{e}(X)}-1\Big\rVert\cdot\Big\lVert\pi_{1}(1,X)\mu_{1}^{1}(X)-\hat{\pi}_{1}(1,X)\hat{\mu}_{1}^{1}(X)\Big\rVert\\
 & +\Big\lVert1-\frac{1-e(X)}{1-\hat{e}(X)}\frac{\pi_{1}(0,X)}{\hat{\pi}_{1}(0,X)}\Big\rVert\cdot\Big\lVert\hat{\pi}_{1}(1,X)\left\{ \mu_{1}^{0}(X)-\hat{\mu}_{1}^{0}(X)\right\} \Big\rVert\\
 & +\Big\lVert\pi_{1}(1,X)-\hat{\pi}_{1}(1,X)\Big\rVert\cdot\Big\lVert\mu_{1}^{0}(X)-\frac{e(X)}{\hat{e}(X)}\hat{\mu}_{1}^{0}(X)\Big\rVert\\
 & +\Big\lVert1-\frac{1-e(X)}{1-\hat{e}(X)}\Big\rVert\cdot\Big\lVert\pi_{1}(0,X)-\hat{\pi}_{1}(0,X)\Big\rVert\\
\text{} & \leq\Big\lVert\left\{ \frac{e(X)}{\hat{e}(X)}-1\right\} \left\{ \mu_{1}^{1}(X)-\hat{\mu}_{1}^{1}(X)\right\} \Big\rVert_{1}\cdot\Big\lVert\pi_{1}(1,X)\Big\rVert_{\infty}\\
 & +\Big\lVert\left\{ \frac{e(X)}{\hat{e}(X)}-1\right\} \left\{ \pi_{1}(1,X)-\hat{\pi}_{1}(1,X)\right\} \Big\rVert_{1}\cdot\Big\lVert\hat{\mu}_{1}^{1}(X)\Big\rVert_{\infty}\\
 & +\Big\lVert\left\{ 1-\frac{1-e(X)}{1-\hat{e}(X)}\right\} \left\{ \mu_{1}^{0}(X)-\hat{\mu}_{1}^{0}(X)\right\} \Big\rVert_{1}\cdot\Big\lVert\hat{\pi}_{1}(1,X)\Big\rVert_{\infty}\\
 & +\Big\lVert\left\{ 1-\frac{\pi_{1}(0,X)}{\hat{\pi}_{1}(0,X)}\right\} \left\{ \mu_{1}^{0}(X)-\hat{\mu}_{1}^{0}(X)\right\} \Big\rVert_{1}\cdot\Big\lVert\frac{1-e(X)}{1-\hat{e}(X)}\Big\rVert_{\infty}\\
 & +\Big\lVert\pi_{1}(1,X)-\hat{\pi}_{1}(1,X)\Big\lVert\cdot\Big\lVert\mu_{1}^{0}(X)-\hat{\mu}_{1}^{0}(X)\Big\lVert\\
 & +\Big\lVert\left\{ \pi_{1}(1,X)-\hat{\pi}_{1}(1,X)\right\} \left\{ \frac{e(X)}{\hat{e}(X)}-1\right\} \Big\rVert_{1}\cdot\Big\lVert\hat{\mu}_{1}^{0}(X)\Big\rVert_{\infty}\\
 & +\Big\lVert1-\frac{1-e(X)}{1-\hat{e}(X)}\Big\lVert\cdot\Big\lVert\pi_{1}(0,X)-\hat{\pi}_{1}(0,X)\Big\lVert|\delta|\\
 & \leq M\Big\lVert\frac{e(X)}{\hat{e}(X)}-1\Big\rVert\cdot\left\{ \Big\lVert\mu_{1}^{1}(X)-\hat{\mu}_{1}^{1}(X)\Big\rVert+\Big\lVert\pi_{1}(1,X)-\hat{\pi}_{1}(1,X)\Big\lVert\right\} \\
 & +M\Big\lVert\mu_{1}^{0}(X)-\hat{\mu}_{1}^{0}(X)\Big\rVert\cdot\left\{ \Big\lVert1-\frac{1-e(X)}{1-\hat{e}(X)}\Big\rVert+\Big\lVert1-\frac{\pi_{1}(0,X)}{\hat{\pi}_{1}(0,X)}\Big\lVert\right\} \\
 & +M\Big\lVert\pi_{1}(1,X)-\hat{\pi}_{1}(1,X)\Big\rVert\cdot\left\{ \Big\lVert\frac{e(X)}{\hat{e}(X)}-1\Big\rVert+\Big\lVert\mu_{1}^{0}(X)-\hat{\mu}_{1}^{0}(X)\Big\lVert\right\} \\
 & +M\Big\lVert1-\frac{1-e(X)}{1-\hat{e}(X)}\Big\lVert\cdot\Big\lVert\pi_{1}(0,X)-\hat{\pi}_{1}(0,X)\Big\lVert.
\end{align*}
}

{{}The second inequality holds by the triangle inequality
and Holder's inequality, and the last inequality holds by Cauchy-Schwarz.
Under $\mathcal{M}_{\text{ps+rp+om}}$, we would expect $\pr\left\{ N(V;\hat{\theta})-N(V;\theta^{*})\right\} =O_{\pr}(n^{-1/2})\cdot o_{\pr}(1)=o_{\pr}(n^{-1/2})$.
Therefore, the EIF-based estimator $\hat{\tau}_{\text{tr}}^{'}$ satisfies
$\hat{\tau}_{\text{tr}}^{'}-\tau_{1}^{\jtr'}=\left(\pr_{n}-\pr\right)N(V;\theta_{0})+o_{\pr}(n^{-\frac{1}{2}})$
and its influence function $N(V;\theta_{0})+\tau_{1}^{\text{\jtr}'}$,
which is the same as the EIF in Theorem \ref{thm: eif_1time-sen}
and completes the proof of Theorem \ref{thm: tr_1time_rate-sen} and
Corollary \ref{cor: tr_1time_rate-sen}.}\end{proof}
\endgroup

\section{Additional results from simulation \label{sec:supp_simu}}

\subsection{Cross-sectional setting}

{}Web Table \ref{tab:supp_sim_1time} shows the simulation results
of the eight estimators for single-time-point outcomes under 8 different
model specifications in terms of the bias and the Monte Carlo standard
deviation (denoted as SD) based on 1000 simulated datasets. The proposed
triply robust estimators are unbiased if any two of the three models
are correct. The calibration-based estimator has the smallest variation
among the three triply robust estimators. Under the correct specification
of all the models, the calibration-based triply robust estimator has
a comparable SD compared to $\hat{\tau}_{\text{ps-om}}$ and $\hat{\tau}_{\text{rp-om}}$.

{} 
\begin{table}[ht]
{} 
\global\long\def\arraystretch{0.75}%
{} \caption{Point estimation in the cross-sectional setting under 8 different
model specifications. \label{tab:supp_sim_1time}}

\noindent {}\resizebox{\textwidth}{!}{%
\begin{tabular}{>{\centering}p{0.04\columnwidth}>{\centering}p{0.04\columnwidth}>{\centering}p{0.04\columnwidth}>{\centering}p{0.1\columnwidth}>{\centering}p{0.07\columnwidth}>{\centering}p{0.07\columnwidth}>{\centering}p{0.07\columnwidth}>{\centering}p{0.07\columnwidth}>{\centering}p{0.07\columnwidth}>{\centering}p{0.07\columnwidth}>{\centering}p{0.07\columnwidth}>{\centering}p{0.07\columnwidth}}
\hline 
\multicolumn{3}{c}{Correct specification} &  &  &  & \multicolumn{4}{c}{Estimators} &  & \tabularnewline
\hline 
PS  & RP  & OM  &  & $\hat{\tau}_{\text{tr}}$  & $\hat{\tau}_{\text{tr-N}}$  & $\hat{\tau}_{\text{tr-C}}$  & $\hat{\tau}_{\text{ps-rp}}$  & $\hat{\tau}_{\text{ps-rp-N}}$  & $\hat{\tau}_{\text{ps-om}}$  & $\hat{\tau}_{\text{ps-om-N}}$  & $\hat{\tau}_{\text{rp-om}}$\tabularnewline
\hline 
yes  & yes  & yes  & Bias (\%)  & -0.04  & -0.05  & \textbf{-0.15}  & -0.21  & 0.00  & \textbf{-0.10}  & \textbf{-0.09}  & \textbf{-0.15}\tabularnewline
 &  &  & SD (\%)  & 7.40  & 7.33  & \textbf{7.10}  & 11.68  & 10.38  & \textbf{7.10 }  & \textbf{7.09}  & \textbf{7.03}\tabularnewline
yes  & yes  & no  & Bias (\%)  & 0.59  & 0.67  & \textbf{-0.22}  & -0.21  & 0.00  & 9.29  & 9.29  & 15.60\tabularnewline
 &  &  & SD (\%)  & 9.53  & 9.22  & \textbf{8.59}  & 11.68  & 10.38  & 8.09  & 8.09  & 8.70\tabularnewline
yes  & no  & yes  & Bias (\%)  & -0.11  & -0.11  & \textbf{-0.32}  & 9.23  & 9.22  & \textbf{-0.10}  & \textbf{-0.09}  & -2.28\tabularnewline
 &  &  & SD (\%)  & 7.24  & 7.24  & \textbf{7.08}  & 8.70  & 8.65  & \textbf{7.10}  & \textbf{7.09}  & 6.91\tabularnewline
no  & yes  & yes  & Bias (\%)  & -0.08  & -0.09  & \textbf{-0.15}  & 7.85  & 7.89  & 12.75  & 12.72  & \textbf{-0.15}\tabularnewline
 &  &  & SD (\%)  & 7.25  & 7.24  & \textbf{7.10}  & 10.34  & 10.01  & 9.85  & 9.84  & \textbf{7.03}\tabularnewline
yes  & no  & no  & Bias (\%)  & 8.16  & 8.11  & 3.05  & 9.23  & 9.22  & 9.29  & 9.29  & 16.35\tabularnewline
 &  &  & SD (\%)  & 8.16  & 8.14  & 8.43  & 8.70  & 8.65  & 8.09  & 8.09  & 8.65\tabularnewline
no  & yes  & no  & Bias (\%)  & 8.30  & 8.31  & -0.22  & 7.85  & 7.89  & 16.56  & 16.56  & 15.60\tabularnewline
 &  &  & SD (\%)  & 9.18  & 9.10  & 8.59  & 10.34  & 10.01  & 8.71  & 8.72  & 8.70\tabularnewline
no  & no  & yes  & Bias (\%)  & 0.06  & 0.05  & -0.32  & 16.33  & 16.31  & 12.75  & 12.72  & -2.28\tabularnewline
 &  &  & SD (\%)  & 7.39  & 6.39  & 7.08  & 9.34  & 9.32  & 9.85  & 9.84  & 6.91\tabularnewline
no  & no  & no  & Bias (\%)  & 14.87  & 14.86  & 3.05  & 16.33  & 16.31  & 16.56  & 16.56  & 16.35\tabularnewline
 &  &  & SD (\%)  & 8.83  & 8.83  & 8.43  & 9.34  & 9.32  & 8.71  & 8.72  & 8.65\tabularnewline
\hline 
\end{tabular}}

\end{table}

{}We compare three types of CIs, including the Wald-type CI with
the variance estimated by nonparametric bootstrap, the Wald-type CI
with the variance estimated by the asymptotic theory as $\hat{\mathbb{V}}(\hat{\tau})=n^{-2}\sum_{i=1}^{n}\left\{ \varphi_{1}^{\text{\jtr}}(V_{i};\hat{\mathbb{P}})-\hat{\tau}\right\} ^{2}$,
and the symmetric t bootstrap CI as $(\hat{\tau}-c^{*}\hat{\mathbb{V}}^{1/2}(\hat{\tau}),\hat{\tau}+c^{*}\hat{\mathbb{V}}^{1/2}(\hat{\tau}))$,
with $c^{*}$ as the $95\%$ quantile of $\{|(\hat{\tau}^{(b)}-\hat{\tau})/\hat{\mathbb{V}}^{1/2}(\hat{\tau}^{(b)})|:b=1,\cdots,B\}$.
Note that the CI comparison is only conducted for the three EIF-based
estimators $\hat{\tau}_{\text{tr}}$, $\hat{\tau}_{\text{tr-N}}$,
and $\hat{\tau}_{\text{tr-C}}$ under the scenario where all the three
models are correctly specified, since Theorem 3
entails that the EIF-based estimators achieve the semiparametric efficiency
bound under $\mathcal{M}_{\text{ps+rp+om}}$. Given that bootstrap
is now used to obtain the CIs, we set the number of bootstrap replicates
to $B=500$. Web Table \ref{tab:supp_sim1time_ci} presents the coverage
rate and the mean CI length for the three types of CIs. The Wald-type
CI with the variance estimated by the asymptotic theory produces an
anti-conservative coverage rate and the smallest mean CI length, while
the Wald-type CI with the variance estimated by nonparametric bootstrap
and the symmetric t bootstrap CI produce comparable coverage rates
and mean CI lengths for each EIF-based estimator. As obtaining the
Wald-type CI with the nonparametric bootstrap variance estimator does
not involve the calculation of the bootstrap CI, which saves computation
time, we recommend using it in the cross-sectional setting.

{} 
\begin{table}
{}\caption{Comparison among the three types of CIs of the EIF-based estimators
in the cross-sectional setting under $\mathcal{M}_{\text{ps+rp+om}}$.\label{tab:supp_sim1time_ci}}
 
\global\long\def\arraystretch{0.75}%
{} 

{}\resizebox{\textwidth}{!}{{\large{}{}}%
\begin{tabular}{c>{\centering}p{0.2\columnwidth}>{\centering}p{0.22\columnwidth}>{\centering}p{0.01\columnwidth}>{\centering}p{0.2\columnwidth}>{\centering}p{0.22\columnwidth}>{\centering}p{0.01\columnwidth}>{\centering}p{0.2\columnwidth}>{\centering}p{0.22\columnwidth}}
\toprule 
 & \multicolumn{2}{c}{{}Wald-type CI by nonparametric bootstrap} &  & \multicolumn{2}{c}{{}Wald-type CI by asymptotic theory} &  & \multicolumn{2}{c}{{}Symmetric t bootstrap CI}\tabularnewline
\midrule 
{}Estimator  & {}Coverage rate ($\%$)  & {}Mean CI length ($\%$)  &  & {}Coverage rate ($\%$)  & {}Mean CI length ($\%$)  &  & {}Coverage rate ($\%$)  & {}Mean CI length ($\%$)\tabularnewline
\midrule 
{}$\hat{\tau}_{\text{tr}}$  & {}95.2  & {}29.82  &  & {}94.2  & {}27.80  &  & {}94.9  & {}29.88\tabularnewline
{}$\hat{\tau}_{\text{tr-N}}$  & {}95.2  & {}29.11  &  & {}93.9  & {}27.72  &  & {}95.0  & {}29.44\tabularnewline
{}$\hat{\tau}_{\text{tr-C}}$  & {}95.0  & {}28.53  &  & {}93.0  & {}26.04  &  & {}95.2  & {}28.28\tabularnewline
\bottomrule
\end{tabular}} 
\end{table}

{To explore the effect of calibration on the proposed
estimators, we additionally incorporate two simple estimators $\hat{\tau}_{\text{ps-rp-C}}$
and $\hat{\tau}_{\text{ps-om-C}}$, where we use calibration to obtain
the propensity score and response probability weights . Web Figure
\ref{fig:point_1time-addcal} and Web Table \ref{tab:sim_1time-addcal}
present the corresponding simulation results. While calibration fails
to improve the performance of $\hat{\tau}_{\text{ps-om-C}}$ as the
true propensity score weights are not extreme under this simulation
setting, it reveals a significant improvement in the estimators $\hat{\tau}_{\text{ps-rp-C}}$
and $\hat{\tau}_{\text{tr-C}}$, since combining the propensity score
and response probability weights together is more likely to generate
extreme values. Among the three calibration-based estimators, the
EIF-based estimator $\hat{\tau}_{\text{tr-C}}$ has the most satisfying
performance with the greatest precision and robustness.}

\begin{figure}
\centering{}\caption{Performance of the estimators in the cross-sectional setting under
8 different model specifications, where ps, rp, and om are shorthands
for the propensity score, response probability, and outcome mean.
In the x-axis, tr, tr-N, and tr-C denote the three EIF-based estimators
$\hat{\tau}_{\text{tr}}$, $\hat{\tau}_{\text{tr-N}}$, and $\hat{\tau}_{\text{tr-C}}$;
psrp, psrp-N, {and psrp-C} denote the estimators $\hat{\tau}_{\text{ps-rp}}$,
$\hat{\tau}_{\text{ps-rp-N}}$, {and $\hat{\tau}_{\text{ps-rp-C}}$};
psom, psom-N, {and psom-C} denote the estimators $\hat{\tau}_{\text{ps-om}}$,
$\hat{\tau}_{\text{ps-om-N}}$, {and $\hat{\tau}_{\text{ps-om-C}}$};
and rpom denotes the estimator $\hat{\tau}_{\text{rp-om}}$ in Example
1. \label{fig:point_1time-addcal}}
\includegraphics[scale=0.5]{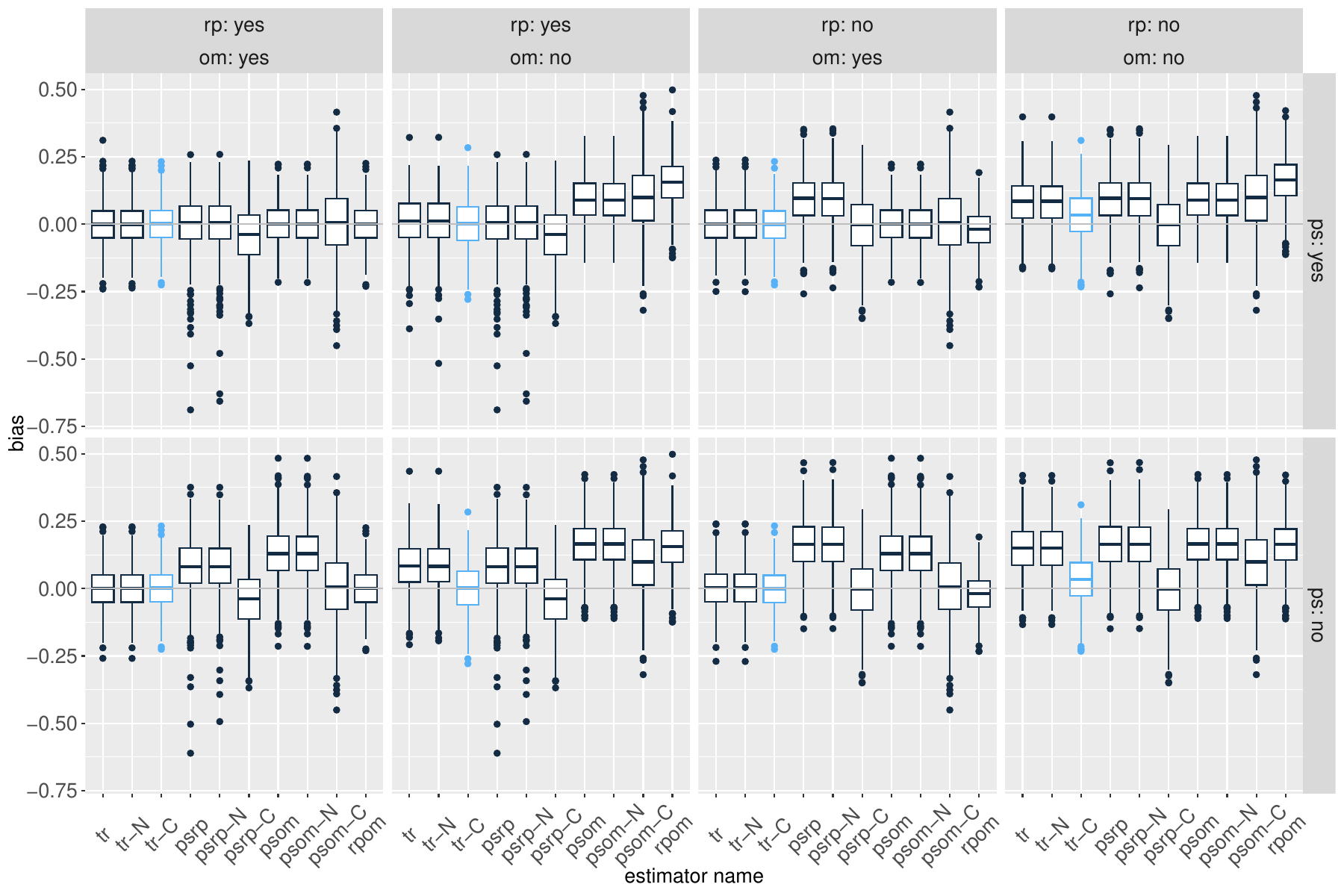}
\end{figure}

\begin{table}[ht]
\global\long\def\arraystretch{0.75}%
\caption{Coverage rates and mean CI lengths in the cross-sectional setting
under 8 different model specifications, where PS, RP, and OM are shorthands
for the propensity score, response probability, and outcome mean.\label{tab:sim_1time-addcal}}

\noindent \resizebox{\textwidth}{!}{%
\begin{tabular}{>{\centering}m{0.05\columnwidth}>{\centering}m{0.05\columnwidth}>{\centering}m{0.05\columnwidth}>{\centering}p{0.01\columnwidth}>{\centering}p{0.07\columnwidth}>{\centering}p{0.07\columnwidth}>{\centering}p{0.07\columnwidth}>{\centering}p{0.07\columnwidth}>{\centering}p{0.07\columnwidth}>{\centering}p{0.07\columnwidth}>{\centering}p{0.07\columnwidth}>{\centering}p{0.07\columnwidth}>{\centering}p{0.07\columnwidth}>{\centering}p{0.07\columnwidth}}
\hline 
\multicolumn{3}{c}{Model specification} &  & \multicolumn{10}{c}{Coverage rate (\%)}\tabularnewline
 &  &  &  & \multicolumn{10}{c}{(Mean CI length, \%)}\tabularnewline
\hline 
PS & RP & OM &  & $\hat{\tau}_{\text{tr}}$ & $\hat{\tau}_{\text{tr-N}}$ & $\hat{\tau}_{\text{tr-C}}$ & $\hat{\tau}_{\text{ps-rp}}$ & $\hat{\tau}_{\text{ps-rp-N}}$ & {$\hat{\tau}_{\text{ps-rp-C}}$} & $\hat{\tau}_{\text{ps-om}}$ & $\hat{\tau}_{\text{ps-om-N}}$ & {$\hat{\tau}_{\text{ps-om-C}}$} & $\hat{\tau}_{\text{rp-om}}$\tabularnewline
\hline 
yes & yes & yes &  & 94.7 & 94.7 & \textbf{94.4} & 95.7 & 95.5 & {93.4} & 94.9 & 94.9 & {95.0} & \textbf{94.3}\tabularnewline
 &  &  &  & (30.9) & (29.5) & \textbf{(28.5)} & (59.9) & (41.8) & {(39.9)} & (29.1) & (29.0) & {(52.3)} & \textbf{(28.2)}\tabularnewline
yes & yes & no &  & 95.3 & 94.8 & \textbf{94.3} & 95.7 & 95.5 & {93.4} & 80.6 & 80.6 & {81.4} & 57.6\tabularnewline
 &  &  &  & (41.8) & (36.1) & \textbf{(33.7)} & (59.9) & (41.8) & {(39.9)} & (33.1) & (33.1) & {(49.3)} & (34.0)\tabularnewline
yes & no & yes &  & 94.1 & 94.1 & \textbf{94.2} & 79.7 & 80.0 & {93.8} & 94.9 & 94.9 & {95.0} & 93.5\tabularnewline
 &  &  &  & (28.8) & (28.3) & \textbf{(28.2)} & (36.7) & (35.3) & {(41.4)} & (29.1) & (29.0) & {(52.3)} & (27.7)\tabularnewline
no & yes & yes &  & 94.4 & 94.4 & \textbf{94.4} & 85.8 & 86.0 & {93.4} & 72.8 & 72.9 & {95.0} & \textbf{94.3}\tabularnewline
 &  &  &  & (29.5) & (29.1) & \textbf{(28.5)} & (45.7) & (40.7) & {(39.9)} & (37.9) & (37.9) & {(52.3)} & \textbf{(28.2)}\tabularnewline
yes & no & no &  & 83.0 & 82.9 & 93.1 & 79.7 & 80.0 & {93.8} & 80.6 & 80.6 & {81.4} & 53.4\tabularnewline
 &  &  &  & (32.7) & (32.3) & (33.8) & (36.7) & (35.3) & {(41.4)} & (33.1) & (33.1) & {(49.3)} & (34.1)\tabularnewline
no & yes & no &  & 84.1 & 83.9 & 94.3 & 85.8 & 86.0 & {93.4} & 53.8 & 53.8 & {81.4} & 57.6\tabularnewline
 &  &  &  & (37.4) & (35.9) & (33.7) & (45.7) & (40.7) & {(39.9)} & (34.6) & (34.7) & {(49.3)} & (34.0)\tabularnewline
no & no & yes &  & 94.6 & 94.6 & 94.2 & 56.1 & 56.1 & {93.8} & 72.8 & 72.9 & {95.0} & 93.5\tabularnewline
 &  &  &  & (29.2) & (29.2) & (28.2) & (38.0) & (37.4) & {(41.4)} & (37.9) & (37.9) & {(52.3)} & (27.7)\tabularnewline
no & no & no &  & 61.3 & 61.3 & 93.1 & 56.1 & 56.1 & {93.8} & 53.8 & 53.8 & {81.4} & 53.4\tabularnewline
 &  &  &  & (35.1) & (34.9) & (33.8) & (38.0) & (37.4) & {(41.4)} & (34.7) & (34.7) & {(49.3)} & (34.1)\tabularnewline
\hline 
\end{tabular}}
\end{table}

\subsection{Longitudinal setting \label{subsec:supp_sim2time}}

{}We use the original covariates $X_{1},\cdots,X_{5}$ in GAM to
approximate each nuisance function separately in each group. {}For
calibration, we incorporate the first two moments of the transformed
covariates $Z$ and all the interactions to calibrate the propensity
score weights, and use the first two moments of the historical information
and all the interactions to calibrate the response probability weights
sequentially. {}

{}Web Table \ref{tab:supp_sim_2time} shows the simulation results
of the eight estimators for longitudinal outcomes under J2R in detail.
The SD in the table refers to the Monte Carlo standard deviation.
{}From the table, the multiply robust estimators are unbiased, while
other estimators suffer from larger deviations from the true value.
Applying calibration tends to improve efficiency, as we observe a
smaller Monte Carlo variation compared to the other two multiply robust
estimators.{}

{} 
\begin{table}
\centering{}{}\caption{Simulation results of the estimators in the longitudinal setting.\label{tab:supp_sim_2time}}
 
\global\long\def\arraystretch{0.75}%
{} \centering{}%
\begin{tabular}{c>{\centering}p{0.14\columnwidth}>{\centering}p{0.14\columnwidth}>{\centering}p{0.14\columnwidth}>{\centering}p{0.14\columnwidth}}
\toprule 
Estimator   & Bias ($\%$)   & SD ($\%$)   & Coverage rate ($\%$)   & Mean CI length ($\%$)\tabularnewline
\midrule 
$\hat{\tau}_{\text{mr}}$   & 4.54   & 10.35   & 95.40   & 43.76 \tabularnewline
$\hat{\tau}_{\text{mr-N}}$   & 4.59   & 10.37   & 95.20   & 43.73 \tabularnewline
$\hat{\tau}_{\text{mr-C}}$   & 3.36   & 9.94   & 96.60   & 42.77 \tabularnewline
$\hat{\tau}_{\text{ps-rp}}$   & 44.55   & 15.18   & 26.70   & 72.04 \tabularnewline
$\hat{\tau}_{\text{ps-rp-N}}$   & 44.56   & 15.26   & 27.10   & 72.13 \tabularnewline
$\hat{\tau}_{\text{ps-om}}$   & 17.31   & 12.03   & 93.10   & 51.47 \tabularnewline
$\hat{\tau}_{\text{ps-om-N}}$   & 17.56   & 12.08   & 92.50   & 52.02 \tabularnewline
$\hat{\tau}_{\text{rp-pm}}$   & -13.14   & 8.57   & 77.20   & 39.50 \tabularnewline
\bottomrule
\end{tabular}
\end{table}

Similar to the cross-sectional setting, we compare {three}
types of CIs of the EIF-based estimators, including the {Wald-type
CI with the variance estimated by nonparametric bootstrap}, the Wald-type
CI with the variance estimated by asymptotic theory, and the symmetric
t bootstrap CI in Web Table \ref{tab:supp_simlongi_ci}, with the
number of bootstrap replicates $B=500$. {While applying
nonparametric bootstrap produces a slightly conservative Wald-type
CI with a wider CI length, the anti-conservative issue of Wald-type
CI with the variance estimated by the asymptotic theory is more pronounced
in the longitudinal setting, resulting in low coverage rates and smaller
mean CI lengths. Using symmetric t bootstrap CI eases those issues
and leads to satisfying coverage rates and mean CI lengths.} Therefore,
we recommend the use of symmetric t bootstrap CI in the longitudinal
setting to obtain reasonable CIs for the multiply robust estimators.

\begin{table}
{}\caption{Comparison among the three types of CIs of the EIF-based estimators
in the longitudinal setting.\label{tab:supp_simlongi_ci}}
 
\global\long\def\arraystretch{0.75}%

{}\resizebox{\textwidth}{!}{{\large{}{}}%
\begin{tabular}{c>{\centering}p{0.14\columnwidth}>{\centering}p{0.14\columnwidth}>{\centering}p{0.14\columnwidth}>{\centering}p{0.14\columnwidth}>{\centering}p{0.01\columnwidth}>{\centering}p{0.14\columnwidth}>{\centering}p{0.14\columnwidth}}
\toprule 
 & \multicolumn{2}{c}{{Wald-type CI by nonparametric bootstrap}} & \multicolumn{2}{c}{Wald-type CI by asymptotic theory} &  & \multicolumn{2}{c}{Symmetric t bootstrap CI}\tabularnewline
\midrule 
Estimator  & {Coverage rate ($\%$) } & {Mean CI length ($\%$) } & Coverage rate ($\%$)  & Mean CI length ($\%$)  &  & Coverage rate ($\%$)  & Mean CI length ($\%$)\tabularnewline
\midrule 
$\hat{\tau}_{\text{mr}}$  & {96.3 } & {45.76 } & 83.7  & 31.08   &  & 95.4 & 43.76 \tabularnewline
$\hat{\tau}_{\text{mr-N}}$  & {96.2 } & {45.90 } & 84.1  & 31.61   &  & 95.2 & 43.73 \tabularnewline
$\hat{\tau}_{\text{mr-C}}$  & {98.3 } & {48.70 } & 82.9  & 28.36   &  & 96.6 & 42.77 \tabularnewline
\bottomrule
\end{tabular}}
\end{table}

\section{Additional results from application\label{sec:supp_app}}

{}The antidepressant clinical trial data is available on \url{https://www.lshtm.ac.uk/research/centres-projects-groups/missing-data#dia-missing-data}
prepared by \citet{mallinckrodt2014recent}. The longitudinal outcomes
in the data suffer from missingness at weeks 2, 4, 6, and 8. All the
missingness in the control group follows a monotone missingness pattern,
while 1 participant in the treatment group has intermittent missing
data. We first delete three individuals with the unobserved investigation
site numbers, and one individual with intermittent missing data for
simplicity, since our proposed framework is only valid under a monotone
missingness pattern. After data preprocessing, 39 participants in
the control group and 30 participants in the treatment group suffered
from monotone missingness. We fit models of the propensity score,
response probability and outcome mean sequentially in backward order,
starting from the last time point. For outcome mean models, we regress
the observed outcome $Y_{4}$ at the last time point on the historical
information $H_{3}$ in the group with $A=a$ to get $\hat{\mu}_{4}^{a}(H_{3})$,
and then regress the predicted value $\hat{\mu}_{4}^{a}(H_{s})$ at
time $s$ on the historical information $H_{s-1}$ using the subset
of the data with $\left(R_{s-1}=1,A=a\right)$ to get $\mu_{4}^{a}(H_{s-1})$
for $s=1,\cdots,3$, recursively. For response probability, we fit
the observed indicator $R_{s}$ with the incorporation of the historical
information $H_{s-1}$ on the data with $\left(R_{s-1}=1,A=a\right)$
to get $\hat{\pi}_{s}(a,H_{s-1})$ for $s=1,\cdots,4$ sequentially.
For propensity score models, the treatment indicator $A$ is regressed
on $H_{s-1}$ using the subset of the data with $R_{s-1}=1$ to get
$\hat{e}(H_{s-1})$. For the pattern mean models ${\color{black}\big\{ g_{s+1}^{1}(H_{l-1}):l=1,\cdots,s\text{ and }s=1,\cdots,4\big\}}$
that rely on both the response probability and outcome mean models,
we regress the predicted value on the historical information $H_{s-1}$
on the subset of the data with $\left(R_{s-1}=1,A=1\right)$.

{}The distributions of the normalized estimated weights involved
in the multiply robust estimators are visualized in Web Figure \ref{fig:weight_app}
(type = ``original''). The weights that correspond to weeks 4 ($A=0$
and $R_{2}=1$), 6 ($A=0$ and $R_{3}=1$) and 8 ($A=0$ and $R_{4}=1$)
suffer from extreme outliers. The existence of outliers explains a
distinct difference in the point estimation of $\hat{\tau}_{\text{ps-rp}}$
and $\hat{\tau}_{\text{ps-rp-N}}$ in Table 3 in the main text. Therefore,
we consider using calibration to mitigate the impact. The distributions
of calibrated weights are also presented in Figure \ref{fig:weight_app}.
As shown by the figure, calibration tends to scatter the concentrated
estimated weights when no outstanding outliers exist in the original
weights, for weights when $A=1$ and $\left(A=0,R_{1}=1\right)$.
However, it stabilizes the extreme weights at weeks 4, 6, and 8, which
explains the narrower CI produced by $\hat{\tau}_{\text{mr-C}}$ compared
to the other two multiply robust estimators.

\begin{figure}
\centering{}{}\caption{Weight distributions of the HAMD-17 data\label{fig:weight_app}}
\includegraphics[scale=0.8]{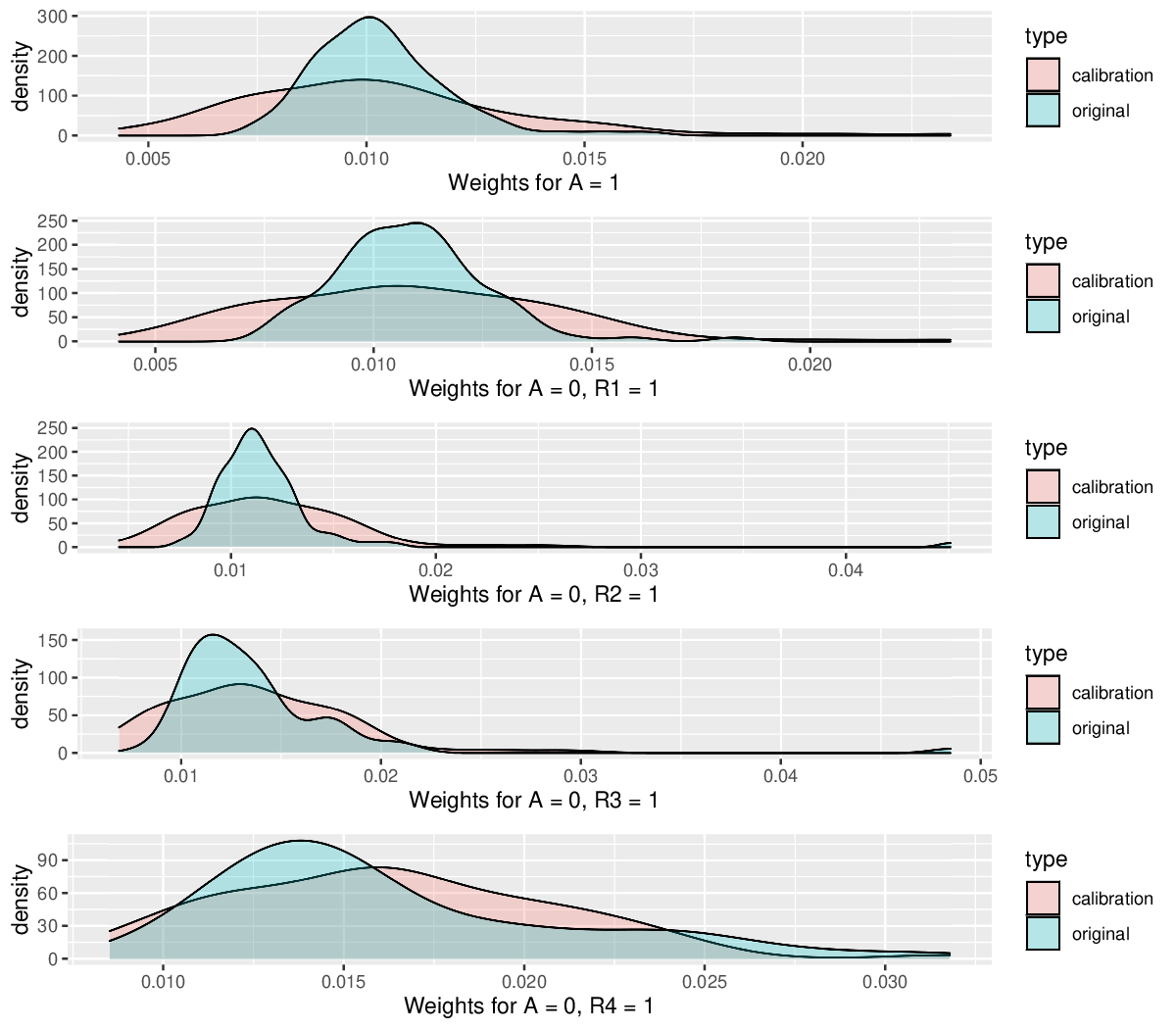}  
\end{figure}

\end{document}